\documentclass[a4paper,onecolumn,allowtoday,unpublished]{quantumarticle}
\pdfoutput=1
\usepackage[T1]{fontenc}
\usepackage[british]{babel}
\usepackage[unicode,colorlinks]{hyperref}
\usepackage{graphicx}
\usepackage{amsmath,amssymb,amsthm,bm}
\usepackage{xspace}
\usepackage{mathtools}
\usepackage{dsfont}
\usepackage{xcolor}
\usepackage{tikz,pgfplots}\pgfplotsset{compat=1.18}

\pgfplotsset{
    legend image with text/.style={
        legend image code/.code={%
            \node[anchor=center] at (0.3cm,0cm) {#1};
        }
    },
}

\usepackage{subfig}
\usepackage{lipsum}
\usepackage{mathrsfs}
\usepackage[all]{xy}

\usepackage{csquotes}
\usepackage[backend=bibtex,style=phys,biblabel=brackets,eprint=true,doi=false,url=true,maxnames=10,giveninits=true]{biblatex}
\let\quantumthanks\thanks   
\usepackage[convert-footnotes=true]{notes2bib}
\let\thanks\quantumthanks  

\usepackage{changes}

\DeclareMathOperator{\logdet}{logdet}

\DeclareMathOperator{\cl}{cl}
\DeclareMathOperator{\im}{Im}

\let\originalleft\left
\let\originalright\right
\renewcommand{\left}{\mathopen{}\mathclose\bgroup\originalleft}
\renewcommand{\right}{\aftergroup\egroup\originalright}

\newcommand{\inprod}[2]{\left\langle #1, #2\right\rangle}

\newcommand{\bra}[1]{\left\langle #1 \right|}
\newcommand{\ket}[1]{\left| #1 \right\rangle}

\newcommand{\ketbra}[2]{\left|#1\middle\rangle\middle\langle#2\right|}

\newcommand{\de}[1]{\left(#1\right)}
\newcommand{\De}[1]{\left[#1\right]}

\newcommand{\inv}[1]{ #1^{-1} }

\newcommand{\id}{\mathds{1}}

\newcommand{\N}{\mathbb{N}}
\newcommand{\R}{\mathbb{R}}
\newcommand{\C}{\mathbb{C}}

\newcommand{\HH}{\mathbb{H}}

\newtheorem{theorem}{Theorem}[section]
\newtheorem*{theorem*}{Theorem}
\newtheorem{lemma}[theorem]{Lemma}
\newtheorem{proposition}[theorem]{Proposition}
\newtheorem{definition}[theorem]{Definition}
\newtheorem{corollary}[theorem]{Corollary}
\newtheorem{remark}{Remark}

\def\supp{\operatorname{supp}}

\newcommand{\claseC}{{\mathscr C}}
\newcommand{\claseCk}[1]{{\mathscr C}^{#1}}

\newcommand{\hypo}{\operatorname{hypo}}

\def\freq{\operatorname{freq}}
\def\Pr{\operatorname{Pr}}
\def\Tr{\operatorname{Tr}}
\def\interior{\operatorname{int}}

\newcommand{\lmap}{\mathcal L}
\newcommand{\g}{\mathcal G}
\newcommand{\z}{\mathcal Z}
\newcommand{\zg}{\mathcal Z \circ \mathcal G}

\newcommand{\gmap}{\widehat{\mathcal G}}
\newcommand{\zmap}{\widehat{\mathcal Z}}
\newcommand{\zgmap}{\widehat{\mathcal{ZG}}}
\newcommand{\pkey}{p^\mathrm{K}}

\newcommand{\Psihat}{\widehat\Psi_\alpha}
\newcommand{\Psihathat}{\widehat{\widehat\Psi}_\alpha}
\usepackage{booktabs}

\newcommand{\pro}[1]{\ket{#1}\bra{#1}}
\newcommand{\mc}[1]{\mathcal{#1}}

\hypersetup{pdftitle={Finite-size quantum key distribution rates from Renyi entropies using conic optimization}}

\begin{document}

\author[1,2]{Mariana Navarro}
\orcid{0000-0001-9381-369X}

\author[3,4]{Andrés González Lorente}
\orcid{0009-0009-1689-5682}

\author[3,4]{Pablo V. Parellada}
\orcid{0000-0002-6768-6671}

\author[1,2]{Carlos Pascual-García} \orcid{0000-0003-0659-7349}

\author[3,4]{Mateus Araújo}\orcid{0000-0003-0155-354X}

\affil[1]{Luxquanta Technologies S.L., Av. Joan Carles I, 30, 1º1ª. 08908 L´Hospitalet de Llobregat, Barcelona, Spain}
\affil[2]{ICFO - Institut de Ciencies Fotoniques, The Barcelona Institute of Science and Technology, Castelldefels, Spain}
\affil[3]{Departamento de Física Teórica, Atómica y Óptica, Universidad de Valladolid, 47011 Valladolid, Spain}
\affil[4]{Laboratory for Disruptive Interdisciplinary Science (LaDIS), Universidad de Valladolid, 47011 Valladolid, Spain.}

\listcsgadd{author1affiliations}{$*$}\listcsgadd{author1affiliations}{$\ddagger$}
\listcsgadd{author2affiliations}{$\dagger$}\listcsgadd{author2affiliations}{$\ddagger$}

\title{Finite-size quantum key distribution rates from Rényi entropies using conic optimization}
\date{\today}

\begin{abstract}
    Finite-size general security proofs for quantum key distribution based on Rényi entropies have recently been introduced. These approaches are more flexible and provide tighter bounds on the secret key rate than traditional formulations based on the von Neumann entropy. However, deploying them requires minimizing the conditional Rényi entropy, a difficult optimization problem that has hitherto been tackled using ad-hoc techniques based on the Frank-Wolfe algorithm, which are unstable and can only handle particular cases. In this work, we introduce a method based on non-symmetric conic optimization for solving this problem. Our technique is fast, reliable, and completely general. We illustrate its performance on several protocols, and demonstrate a dramatic runtime improvement over the state of the art.
\end{abstract}

\maketitle

\insert\footins{\footnotesize\parindent=0pt
  $^{*}$Contact author: mariana.navarro@icfo.eu\par
  $^{\dagger}$Contact author: andres.gonzalez.lorente@uva.es\par
  $^{\ddagger}$These authors contributed equally to this work.\par}

\section{Introduction}

Quantum key distribution (QKD) \cite{BB84,E91} allows two honest parties, traditionally named Alice and Bob, to create a shared key that remains secret against any quantum third parties thanks to the principles of quantum information theory. In particular, QKD is secure against adversaries with unlimited resources and the ability to perform cross-correlations between the different signals that compose the protocol---a scenario known as general security \cite{renner2008security,Renner2022Security}. 

Recent works have successfully achieved general security for different protocols (for instance, see Refs. \cite{tan2020improved,baeuml2023security,kamin25MEATsecurity}) thanks to different versions of the so-called entropy accumulation theorem (EAT) \cite{dupuis2016entropy,dupuis2019entropy}. The EAT framework reduces any multi-round QKD protocol into a single-round scenario, together with a finite-size analysis, which allows security to be proven for real QKD protocols that run for a finite number of rounds. Among its variants, the marginal-constrained EAT (MEAT) \cite{arqand2025MEAT} constitutes the most up-to-date version. It relies on a formulation based on Rényi entropies, removing impractical limitations from previous approaches \cite{metger2022security} while providing tighter finite-size corrections compared to former results \cite{arqand2024generalized}.

However, Rényi-based secret key rate computations require optimizing the conditional Rényi entropy, which is challenging to perform and is not covered by standard optimization software. Existing works \cite{chung2025,kamin25MEATsecurity} handle it with ad-hoc methods based on the Frank-Wolfe algorithm \cite{frankwolfe56,winick2018reliable}. Unfortunately, this approach suffers from low convergence rates and instability. Furthermore, the existing implementations skipped facial reduction \cite{drusvyatskiy2017,hu2022}, a critical step necessary to make the gradient well-defined, opting instead to add artificial depolarization noise, which reduces both precision and computational performance. Another approach would be to use existing non-symmetric conic methods for optimizing the sandwiched Rényi relative entropy \cite{he2025}. However, this would also require adding artificial depolarization noise, because performing facial reduction makes the method no longer applicable. This is analogous to the case of von Neumann relative entropy, where facial reduction necessitated introducing a new cone dedicated to the QKD problem \cite{Lorente2025quantumkey,he2024}.

In this work we perform facial reduction on the sandwiched Rényi relative entropy cone, and introduce a dedicated cone for optimizing over the resulting problem, named RényiQKD. We do the same for a lower bound introduced in Ref. \cite{chung2025}, resulting in another cone, named FastRényiQKD. The latter cone gives very close bounds and is much faster, and is therefore the most appropriate one for numerical computations. These cones can then be used for computing secret key rates using Rényi entropies within the MEAT framework combined with the Rényi leftover hashing lemma \cite{Dupuis23RenyiHashLemma}. We provide numerical examples for several protocols: qubit BB84 \cite{BB84,BBM92}, mutually unbiased bases (MUB) \cite{Cerf2001,Sheridan2010}, and discrete modulated continuous-variable (DMCV) QKD \cite{zhang2024continuousvariable,pascualgarcia2024}. Our implementation of the cones is developed in the programming language Julia \cite{JuliaLang} using the conic solver Hypatia \cite{coey2022solving,coey2022performance}, which uses an improved version of the non-symmetric optimization algorithm by Skajaa and Ye \cite{skajaa2015,papp2017}. It is interfaced by the modeller JuMP \cite{Lubin2023}, which makes it easy to use and mix with extra conic constraints that may come from specific QKD protocols, such as the partial state characterization constraints in Ref. \cite{pereira2025optimalkeyratesquantum}.

The rest of this document goes as follows: in Sec \ref{Sec:Prolegomena} we introduce diverse mathematical concepts that will be required along the manuscript, in Sec \ref{Sec:Coherent_Security} we provide a generic QKD protocol amenable for an analysis under the MEAT, as well as the main results to derive the secret key rate. Sec \ref{sec:ComputeSKR} shows the application of facial reduction and the construction of our cones. Sec \ref{Sec:Experiments} presents our numerical results and Sec \ref{Subsec:Benchmarks} benchmarks our method against previous approaches. We conclude in Sec \ref{Sec:Conclusion}.

\section{Notation and definitions}\label{Sec:Prolegomena}

Given a Hilbert space $\mc{H}_X$, we define $\mc{D}(X)$ as the set of all trace one, positive semidefinite matrices acting on this space, whereas $\mathbb{H}^d_{\succeq}$ (respec. $\mathbb{H}^d_{\succ}$) will be used to denote Hermitian positive semidefinite (respec. definite) matrices of dimension $d$. For a register $X$, we denote $n$ identical copies as $X^{(n)}$, and define the notation $[XY]_j := X_jY_j$. Let us denote a state $\rho_{XY}$ as classical-quantum (cq-state)  when it can be expressed as 
\begin{equation}
    \rho_{XY} = \sum_{x\in \mathcal{X}} p(x) \ketbra{x}{x}_X\otimes \rho^x_Y,
\end{equation}
for some countable alphabet $\mathcal{X}$, which may induce a space of probabilities denoted as $\mathbb{P}_{\mc{X}}$. Any subset $\Omega \subseteq \mathcal{X} $ can be used to define an event, which allows a restriction for cq-states according to a conditioning, as

\begin{equation} \label{eq:TrPr}
    \rho_{XY|\Omega} = \frac{1}{\mathrm{Pr}_\rho [\Omega]}\sum_{x\in \Omega} p(x) \ketbra{x}{x}_X\otimes \rho^x_Y,
\end{equation}
with $\mathrm{Pr}_\rho[{\Omega}] = \sum_{x\in \tilde{\Omega}} p(x)$ representing the probability of observing the event. When only one case is contained in the event, i.e. $\Omega=\{x\}$ for some $x\in\mc{X}$, we will directly refer to the conditioned state as $\rho_{|x}$. We define the generalized trace distance as
\begin{equation} \label{eq:GenTraceDist}
    T(\rho,\sigma) = \frac{1}{2} \lVert \rho - \sigma \rVert_1 +\frac{1}{2}\left|\Tr[\rho-\sigma]\right|,
\end{equation}
and $\lVert p \rVert_1$ to denote the trace norm of a vector $p$. We also make use of several entropic quantities; particularly the Rényi entropy, expressed as $H_\alpha(A)_\rho$ for $\alpha \in (0,\infty)\backslash \{1\}$, and its sandwiched conditional version, defined as
\begin{equation}
    H^\uparrow_{\alpha}(A|B)_\rho = \sup_{\sigma_B \in \mc{D}(\mc{H}_B)} - D_\alpha(\rho_{AB}\|\mathds{1}\otimes\sigma_B),
\end{equation}
where we use the sandwiched Rényi relative entropy
\begin{equation} \label{eq:QRenyiDiv}
    D_\alpha(\rho\|\sigma) = \frac1{\alpha-1}\log\De{ \frac{\Psi_\alpha(\rho, \sigma)}{\Tr[\rho]}},
\end{equation}
where 
\begin{equation}\label{eq;PsiRenyi}
    \Psi_\alpha(\rho, \sigma) = \Tr\left[(\sigma^\frac{1-\alpha}{2\alpha}\rho\sigma^\frac{1-\alpha}{2\alpha})^\alpha\right],
\end{equation}
provided that $\alpha \in (0,\infty)\backslash\{1\}$ and $\mathrm{supp}(\rho) \subseteq \mathrm{supp}(\sigma)$. Unless explicitly stated otherwise, we will consider trace-one density matrices, and directly set $\Tr[\rho]=1$ for \eqref{eq:QRenyiDiv}. The case $\alpha \to 1$ is well-defined and leads to the usual quantum relative entropy
\begin{equation}
    D(\rho\|\sigma) = \Tr[\rho \log(\rho) - \rho \log(\sigma)], 
\end{equation}
which we will denote as the Kullback-Leibler divergence $D_\mathrm{KL}(\cdot \| \cdot)$ when both states are classical. To conclude, we caution the reader that the numerical optimization methods used in this manuscript typically employ natural logarithms, while the calculations for the key rates make use of the base 2 logarithm.

\section{Protocol outline}\label{Sec:Coherent_Security}

Let us follow Ref. \cite{chung2025} and provide a brief, generic description of a prepare-and-measure QKD protocol that is valid under the MEAT \cite{arqand2025MEAT}.

\paragraph{Preparation and measurement.}
For each round $j \in \{1,\ldots,n\}$, Alice prepares a state $\ket{\psi^x}_{A'}$ according to the value $x$ of a random variable $X_j$, and sends it to Bob. Bob receives the state, applies one (of possibly many) POVM, and records the outcome in a register $Y_j$. Then, Alice (possibly also Bob) designates the round as either a key generation or parameter estimation round, with probabilities $p^\mathrm{K}$ and $1-p^\mathrm{K}$, respectively. Additionally, they build a reference key register $R_j$ for the key generation. 

\paragraph{Public announcement.} Alice and Bob publicly disclose information about the registers $[XY]^{(n)}$, with all such communication stored in a register $I^{(n)}$. They use this register to construct $C^{(n)}$, which contains all the public classical information associated to parameter estimation rounds.

\paragraph{Sifting} Using $[RI]^{(n)}$, Alice builds her key register $S^{(n)}$ \footnote{In the case of continuous variable QKD, it should be understood that $S^{(n)}$ and ${C'}^{(n)}$ are actually held by Bob due to the reverse nature of information reconciliation for continuous variable QKD.}. Optionally, Alice may also create a register ${C'}^{(n)}$ with private information \footnote{Such as a virtual tomography, for implementations as in Ref. \cite{baeuml2023security}.} for statistical checks.

\paragraph{Parameter estimation} Alice and Bob use $C^{(n)}$ (optionally, Alice uses ${C'}^{(n)}$) to perform parameter estimation in order to bound Eve's information. If the process reveals that Eve holds excessive information, they abort.

\paragraph{Information reconciliation.} Alice and Bob use an error correction subroutine, which might follow a reverse reconciliation, to remove the discrepancies between their private keys. Alice and Bob then certify that their keys are the same using a universal$_2$ hash function. If their hashed keys disagree, they abort.
\paragraph{Privacy amplification.} Alice and Bob distil a shorter, fully private shared key recorded as $K_A$ and $K_B$ respectively, by using another universal$_2$ hash function. \\

For future reference, we define $\mc{C}$ as the alphabet for the registers $C^{(n)}$. To specify the values observed by Alice and Bob during parameter estimation rounds, let us also define $\tilde{\mc{C}}:= \mc{C}\backslash\{\perp\}$. On similar grounds, we denote Alice's key registers $S^{(n)}$ as defined on an alphabet $\mc{S}$, where $S_j=\perp$ indicates a non-key round, and $\tilde{\mc{S}} := \mc{S}\backslash \{\perp\}$ for non-trivial key values. Finally, we assume that, within a given session of the protocol, no public announcements are needed before all measurements are performed, which allows an simplified analysis via Ref. \cite[Corollary 4.2]{arqand2025MEAT}. Furthermore, this approach allows an interpretation of the protocol in terms of an entanglement-based picture \cite{Curty2004,Ferenczi2012}, where Alice prepares an entangled state $\ket{\psi}_{AA'}$ in every round.

\subsection{Security and completeness}

The security of QKD protocols is customarily defined via the $\varepsilon$-security framework by Renner \cite{renner2008security}, which establishes that Alice's and Bob's final key will be $\varepsilon_\mathrm{s}$-secure if their shared key is $\varepsilon_\mathrm{s}$-close to being completely decoupled from Eve's total information. Quantitatively, this is expressed according to the formula
\begin{equation} \label{eq:secure}
    \mathrm{Pr}_\rho[\Omega] T( \rho_{K_A K_BE|\Omega},\bar{\tau}_{K_A K_BE|\Omega} )\leq \varepsilon_\mathrm{s},
\end{equation}
for a given $\varepsilon_\mathrm{s}>0$. Here, $E$ represents Eve's information, and $\rho_{K_AK_BE}$ is the final state shared by the three parties after the execution of the protocol. The set $\Omega$ represents the event of not aborting at any step of the protocol, and $\bar{\tau}_{K_AK_BE}$ is the ideal output state, where the keys share ideal local randomness and no correlations with respect to any other registers. Formally, we define
\begin{align}\label{eq:IdealCCQ}
    \bar{\tau}_{K_A K_B E | \Omega} &= \tau_{K_A K_B} \otimes \rho_{E|\Omega} = \frac{1}{d}\sum_{x=0}^{d-1} \pro{x,x}_{K_A K_B}\otimes \rho_{E|\Omega}.
\end{align}
Condition \eqref{eq:secure} can be simplified according to a triangle inequality into the statements of correctness and secrecy, such that we may define an  $\varepsilon_\mathrm{EC}$-correct implementation according to
\begin{equation}\label{eq:correctness}
    \mathrm{Pr} [(K_A \neq K_B) \wedge \Omega] \leq \varepsilon_\mathrm{EC}.
\end{equation}
In particular, correctness can be enforced by Alice sending her key to Bob after applying a universal$_2$ hash function once error correction has finished, such that the total length of the hashed key does not exceed $\lceil\log(1/\varepsilon_\mathrm{EC})\rceil$ bits (see Ref. \cite{Tupkary2024} for a complete derivation).

Secrecy is induced by the fact that Eve does not hold excessive information about one of the keys. Namely, a protocol is $\varepsilon_\mathrm{PA}$-secret whenever
\begin{equation} \label{eq:secret}
    \mathrm{Pr}_\rho[\Omega] T( \rho_{K_AE|\Omega}, \bar{\tau}_{K_AE|\Omega} ) \leq \varepsilon_\mathrm{PA}.
\end{equation}
We defer the discussion about the security of our approach to Appendix \ref{App:Proof_Th}, which will be based on a proof against general attacks via the MEAT \cite{arqand2025MEAT}.

Although the above definitions indeed certify that the final key remains unknown against any quantum adversary, it is still necessary to ensure that an execution will not frequently abort. This is given by the notion of completeness, which certifies that any honest implementation of the protocol (i.e., when Eve performs no attack) will be successful with very high probability. Formally, a protocol is $\varepsilon_\mathrm{cm}$-complete whenever 
\begin{equation}\label{eq:CompletenessBound}
1-\Pr^\mathrm{h}[\Omega]\leq\varepsilon_\mathrm{cm},
\end{equation}
where the superscript $(\cdot)^h$ denotes that the implementation is honest. In general, completeness must be enforced for both parameter estimation and error correction. For the former, this is done via $C^{(n)}$ and a reference probability distribution $p\in \mathbb{P}_\mathcal{C}$ on the alphabet $\mathcal{C}$, which constitutes the vector of probabilities expected by Alice and Bob in such an honest implementation. We define $\hat{\Sigma}_{\Omega}$ as the set of all accepted frequency estimators according to $\{p(c)\}_{c \in \mc{C}}$, and a scalar $\delta>0$. Explicitly, 
\begin{align}\label{eq:AcceptanceSet}
    \hat{\Sigma}_{\Omega} = \left\{\freq_{c^{(n)}} \in \mathbb{P}_{\mc{C}}: \lVert \mathrm{freq}_{c^{(n)}}  -p \rVert_1 \leq \delta \right\}.
\end{align}
During parameter estimation, Alice and Bob will verify that their observed finite distribution $\freq_{c^{(n)}}$ is actually contained in $ \hat{\Sigma}_{\Omega}$, or abort otherwise. This analysis is carried out according to the margin of confidence for the statistical estimators; given an honest implementation (such that we may assume that the rounds are iid), there exists a probability $\bar{\varepsilon}_\mathrm{PE}> 0 $ such that
\begin{equation}\label{eq:Completeness1}
    \Pr^\mathrm{h}\left[\lVert \mathrm{freq}_{c^{(n)}}  -p \rVert_1 \leq \delta \right] \geq 1 - \bar{\varepsilon}_\mathrm{PE}.
\end{equation}
In order to bound this probability quantitatively, we make use of the following concentration bound.

\begin{proposition}[Bretagnolle-Huber–Carol inequality] \emph{\cite{BHC_Ineq,vanderVaart2023} }Let $p$ be a multinomial distribution defined on an alphabet $\mc{X}$. After $n \in \mathbb{N}$ iid trials, the empirical frequency distribution $\hat{p}$ verifies
\begin{equation}
    \Pr\left[ \lVert \hat{p}  -p \rVert_1 \leq \delta \right] \geq 1 - 2^{|\mc{X}|} e^{-n \delta^2/2}.
\end{equation}
\end{proposition}
Provided that this proposition holds in an implementation, and identifying the right-hand side with the one of \eqref{eq:Completeness1}, the step of parameter estimation will be successful.
Hence, by splitting the nonabortion event as $\Omega = \Omega_\mathrm{PE}\cap \Omega_\mathrm{EC}$ (i.e. success at both parameter estimation and error correction), we can certify via the union bound that the protocol is $(\bar{\varepsilon}_\mathrm{PE}+\bar{\varepsilon}_\mathrm{EC})$-complete. This conclusion assumes the observation of an $\bar{\varepsilon}_\mathrm{EC}$-complete information reconciliation \footnote{Information reconciliation is exclusively performed through the classical channel, which cannot be attacked by Eve beyond wiretapping. Therefore, it can be assumed that its implementation is always honest.}, which we further develop in Sec \ref{sec:error correction}.

\subsection{Finite-size secret key}

Using the definitions of Secs \ref{Sec:Prolegomena} and \ref{Sec:Coherent_Security}, we characterize the final secret key achievable via the following result.

\begin{theorem}\label{Th:MainTheorem}
    Let  $n\in \mathbb{N}$, $\sigma_A \in \mc{D}(A)$ and $\tilde{\mc{M}}:[AB]_j \to [SCI]_j$ a quantum channel concatenated $n$ times to form the protocol. Let further $C$  a classical register defined on an alphabet $\mc{C}$ and $\Omega$ the nonabortion event for the protocol.
    Then, for a choice of coefficients $\varepsilon_\mathrm{EC},\varepsilon_\mathrm{PA},\bar{\varepsilon}_\mathrm{PE}, \bar{\varepsilon}_\mathrm{EC}>0$, the protocol is $\varepsilon_\mathrm{EC}$-correct, $\varepsilon_\mathrm{PA}$-secret and $(\bar{\varepsilon}_\mathrm{PE} + \bar{\varepsilon}_\mathrm{EC})$-complete providing a binary key whose length verifies
    \begin{align} \label{eq:TheoremLength}
    \ell \leq&  \, n h^\uparrow_{\alpha} 
    -  \frac{\alpha}{\alpha-1} \log \left(\frac{1}{\varepsilon_\mathrm{PA}} \right)- \mathrm{leak}_\mathrm{EC} + 2,
    \end{align}
    where $\alpha \in (1,2)$; $\mathrm{leak}_\mathrm{EC}$ is the information lost to Eve during an $\varepsilon_\mathrm{EC}$-correct and $\bar{\varepsilon}_\mathrm{EC}$-complete error correction, and
    \begin{equation}\label{eq:mainopt}
    h^\uparrow_{\alpha} \geq \inf_{{q}\in \hat{\Sigma}_{\Omega}} \inf_{\omega\in \mc{D}(AB)}  \inf_{\psi\in \mc{D}(RAB)}
      \frac{\alpha}{\alpha - 1} D_\mathrm{KL}({q}\|\tilde{\mc{M}}(\omega)_{C}) +  q(\perp) D_{\gamma} (\g(\omega)\|\z (\psi))
    \end{equation}
    where $\gamma = \alpha/(2 \alpha -1) \in (2/3, 1)$, $\Tr_B[\omega_{AB}] = \sigma_A$, and $\hat{\Sigma}_{\Omega}$ is the set defined in \eqref{eq:AcceptanceSet} given a reference distribution $p \in \mathbb{P}_\mc{C}$ and $\delta \geq 0$ verifying $\bar{\varepsilon}_\mathrm{PE} \geq 2^{|\mc{C}|}  e^{-n \delta^2/2}$. The maps $\mathcal{G}:AB\to RAB$ and $\mathcal{Z} : RAB \to RAB$ denote the CPTP coherent and pinching maps involved in a single-round key distillation.
\end{theorem}

 Additionally, we provide the following lower bound, which greatly simplifies the computation of the secret key rate at the cost of a slightly suboptimal yield.

\begin{corollary}\label{Th:MainCorollary}
    Given the conditions stated in Theorem \ref{Th:MainTheorem}, the bound $h^\uparrow_\alpha \geq h^\downarrow_\alpha$ holds at \eqref{eq:TheoremLength}, where
\begin{equation}\label{eq:fastmainopt}
     h^\downarrow_\alpha = \inf_{q \in \hat{\Sigma}_\Omega } \inf_{\omega \in \mc{D}(AB)} \frac{\alpha}{\alpha -1} D_\mathrm{KL}({q}\|\tilde{\mc{M}}(\omega)_{C}) + q(\perp) D_\beta (\mc{G}(\omega)\|\zg{(\omega)} ),
\end{equation}
with $\beta = 1/\alpha \in (1/2, 1)$.
\end{corollary}
We defer the security proof of Theorem \ref{Th:MainTheorem} to Appendix \ref{App:Proof_Th}.

\subsection{Error correction cost} \label{sec:error correction}

To conclude the theoretical framework, we define the information leaked during error correction. Albeit this step is carried through an authenticated classical channel, there is no universal, information-theoretic bound on the length of $\mathrm{leak}_\mathrm{EC}$. Some implementations can indeed characterize said length \cite{Renes_2012,renner2008security}, although they require algorithms that are challenging to apply. 

As an efficient alternative, Alice and Bob can impose a limit to the total information exchanged during error correction before executing the protocol, and verify the success of the implementation by using a universal$_2$ hash function. Taking this approach,  $\mathrm{leak}_\mathrm{EC}$ can be reasonably estimated via the Slepian-Wolf bound \cite{SlepianWolf}, defining an  $\varepsilon_\mathrm{EC}-$correct and $\varepsilon_\mathrm{EC}-$complete implementation as
\begin{align}\label{eq:EC_P2P}
    \mathrm{leak}_\mathrm{EC} \leq p^\mathrm{s} n f H(S|Y) + \left\lceil \log \left(\frac{1}{\varepsilon_\mathrm{EC}} \right) \right\rceil ,
\end{align}
where $p^\mathrm{s}$ denotes the sifting probability, i.e., the fraction of key rounds that are not lost due to no-clicks and/or postselection, and thus require to be error-corrected. The conditional Shannon entropy $H(S|Y)$ is provided by Alice's key register conditioned on Bob's (after removing the $\{\perp\}$ symbol from said registers), which can be analytically calculated from their measurements. Furthermore, $\lceil \log \left(1/\varepsilon_\mathrm{EC} \right) \rceil $ denotes the cost in bits of using the verification hash, and $f\geq 1$ represents the error correction efficiency, where $f = 1$ constitutes the ideal case (known as the Shannon limit).

\subsection{Non-trace-preserving maps}\label{sec:tracepreserving}

Although in Theorem \ref{Th:MainTheorem} the coherent map $\g$ is assumed to be trace-preserving, we often need to model non-detections or postselection, which leads to trace-non-increasing maps. To address this, we extend such maps into trace-preserving ones by adding a symbol $\{\bot\}$ corresponding to non-detection or discarded events, similar to Ref. \cite[Corollary 9]{kamin25MEATsecurity}. 

Let $\g' :AB \to \tilde{R}AB$ be the original map with register $\tilde{R}$ defined on an alphabet $\tilde{\mc{R}}$ as embedded on an extended register $R$ on $\mc{R} = \{\bot\}\cup\tilde{\mc{R}}$. We then consider its trivial extension $\g_\top : AB \to RAB$, given explicitly by $\g_\top(\omega) = (V\otimes \id_{AB})\g'(\omega)(V^\dagger \otimes \id_{AB})$ where $V = \sum_{r \in \tilde{\mc{R}}} \ket{r}_{R}\bra{r}_{\tilde{R}}$ is an isometry. Analogously, for the map $\z' : \tilde{R}AB \to \tilde{R}AB$ we define its extension $\z_\top : RAB\to RAB$.
 Our goal now is to construct additional maps $\g_\bot$ and $\z_\bot$ such that $\g := \g_\top + \g_\bot$ and $\z := \z_\top + \z_\bot$ are trace-preserving and satisfy the orthogonality relations
 \begin{equation} \label{eq: ortoghonality relations}
     \g_\top(\omega)\g_\bot(\omega') = \z_\top(\psi)\z_\bot(\psi')=\g_\top(\omega)\z_\bot(\psi')=\z_\top(\psi)\g_\bot(\omega') = 0,\quad\forall \omega,\omega',\psi,\psi'.
 \end{equation}
This construction is not unique, and a generic choice is
 \begin{gather}
     \g_\bot(\omega) := [\Tr(\omega)-\Tr(\g_\top(\omega))]\pro{\bot}_R\otimes \frac{1}{d_{AB}}\id_{AB}, \nonumber\\
     \z_\bot(\psi) := (\pro{\bot}_R\otimes \id_{AB}) \psi (\pro{\bot}_R\otimes \id_{AB}). \nonumber\\
 \end{gather}
Using the orthogonality relations in \eqref{eq: ortoghonality relations} on these maps, the expression $\Psi_\gamma(\g(\omega),\z(\psi))$ used in \eqref{eq:mainopt} decomposes as
 \begin{equation}
    \Psi_\gamma(\g(\omega),\z(\psi)) = \Psi_\gamma(\g_\top(\omega),\z_\top(\psi)) + \Psi_\gamma(\g_\bot(\omega),\z_\bot(\psi)),
\end{equation}
such that we can split the optimization into two smaller cones. In addition, the fact that $\z_\bot\circ\g_\bot(\omega) = \g_\bot(\omega)$ allows a further simplification for the expression $\Psi_\beta(\g(\omega),\z\circ\g(\omega))$ used in \eqref{eq:fastmainopt}:
 \begin{equation}
    \Psi_\beta(\g(\omega),\z\circ\g(\omega)) = \Psi_\beta(\g_\top(\omega),\z_\top\circ\g_\top(\omega)) + \Tr(\omega)-\Tr(\g_\top(\omega)),
\end{equation}
which implies that for the lower bound $h^\downarrow_\alpha$, we never need to construct $\g_\bot$ and $\z_\bot$ explicitly.

\section{Computing the key rate}\label{sec:ComputeSKR}

The problem of computing the secret key rate thus reduces to solving the optimization problem in \eqref{eq:mainopt} (or  \eqref{eq:fastmainopt}). It consists of a Kullback-Leibler divergence and a Rényi relative entropy. The Kullback-Leibler divergence is straightforward to optimize: it is a jointly convex function that can be tackled by standard convex optimization techniques. Moreover, it can be handled by conic solvers that support non-symmetric cones, such as Hypatia \cite{coey2022solving} and QICS \cite{QICS}, which makes the optimization particularly effective and convenient.

The challenging term is the sandwiched Rényi relative entropy. It is also jointly convex, and thus in principle it can also be tackled by standard convex optimization techniques, such as the Frank-Wolfe algorithm \cite{frankwolfe56}. However, its arguments $\g(\omega), \z(\psi)$ often have eigenvalues that are constrained to be equal to zero. This makes the derivatives undefined, which causes the optimization to become unstable and fail. Indeed, previous attempts suffered from numerical instability and could only handle particular cases \cite{chung2025, kamin25MEATsecurity}. Technically, we say that the problem is not strictly feasible, and therefore Slater's condition is not satisfied. There is a standard technique to handle this, consisting of reformulating the problem in terms of the support of the matrices, known as facial reduction \cite{drusvyatskiy2017}. It has been successfully applied in the case of von Neumann entropies \cite{hu2022, Lorente2025quantumkey}.

\subsection{Facial reduction}

Before going further, note that the Rényi relative entropy is a composition of the function $\Psi_\alpha$ with the logarithm. Since the logarithm is a monotonic function, it is enough to optimize over $\Psi_\alpha$, and we will do so for simplicity.

What we need is to reformulate the function $\Psi_\alpha(\g(\omega),\z(\psi))$ in terms of the supports of $\omega\in \HH^p_\succeq, \; \psi\in \HH^m_\succeq,\; \g(\omega)\in \HH^m_\succeq,$ and $ \z(\psi)\in \HH^m_\succeq$. First of all, let us consider the support of $\omega$. In general, there are several constraints acting on it, which might constrain some of its eigenvalues to be zero. We write then $\omega = V\omega'V^\dagger$ for some isometry $V\in \C^{p\times q}$ such that $\omega'$ is not constrained to be singular. On the other hand, the only constraint acting on $\psi$ is that it has trace one, so it can already be full rank; $\psi\in \HH^m_\succ$, and no reduction is needed. Let us define then the strictly positive maps $\gmap, \zmap$ such that 
\begin{subequations}
    \begin{gather}
        \g(V\omega'V^\dagger) = W_\g\gmap(\omega')W_\g^\dagger, \label{eq:greduction} \\
        \z(\psi) = W_\z\zmap(\psi)W_\z^\dagger, \label{eq:zreduction}
    \end{gather}
\end{subequations} 
for some isometries $W_\g\in \C^{m\times k}, \; W_{\z}\in \C^{m\times n}$. Consequently, we have that 
\begin{align}
    \Psi_\alpha(\g(\omega),\z(\psi)) = \Psi_\alpha(W_\g\gmap(\omega')W_\g^\dagger,W_\z\zmap(\psi)W_\z^\dagger).
\end{align}
In order to put this in a convenient form, we will use the identities
\begin{equation}\label{eq:equivalence}
\Psi_\alpha(\rho,\sigma) = \Tr\left[(\sigma^\frac{1-\alpha}{2\alpha}\rho\sigma^\frac{1-\alpha}{2\alpha})^\alpha\right] = \Big\|\sigma^\frac{1-\alpha}{2\alpha}\rho^\frac12\Big\|^{2\alpha}_{2\alpha} = \Tr\left[(\rho^\frac12\sigma^\frac{1-\alpha}{\alpha}\rho^\frac12)^\alpha\right].
\end{equation}
which then imply that
\begin{equation}
\Psi_\alpha(W_\g\gmap(\omega')W_\g^\dagger,W_\z\zmap(\psi)W_\z^\dagger) = \Big\|\zmap(\psi)^\frac{1-\alpha}{2\alpha}W_\z^\dagger W_\g \gmap(\omega')^\frac12\Big\|^{2\alpha}_{2\alpha}
\end{equation}
Letting then $S = W_\z^\dagger W_\g$, which as we prove in Appendix \ref{sec:isometry} is an isometry, we define
\begin{equation}\label{eq:truepsi}
    \Psihat(\omega',\psi) := \Big\|\zmap(\psi)^\frac{1-\alpha}{2\alpha}S \gmap(\omega')^\frac12\Big\|^{2\alpha}_{2\alpha},
\end{equation}
and the facial reduction of the sandwiched Rényi relative entropy is done. 

It is also necessary to do the facial reduction of $\Psi_\alpha(\g(\omega),\z \circ \g(\omega))$. It proceeds along the same lines, with the difference that we define the strictly positive map $\zgmap$ such that 
\begin{equation}
    \zg(V\omega' V^\dagger) = W_{\z\g}\zgmap(\omega')W_{\z\g}^\dagger, \label{eq:zgreduction}
\end{equation}
for some isometry $W_{\z\g} \in \C^{m \times r}$, we define $S = W_{\z\g}^\dagger W_\g$, and end up with
\begin{equation}\label{eq:fastpsi}
    \Psihat(\omega') := \Big\|\zgmap(\omega')^\frac{1-\alpha}{2\alpha}S \gmap(\omega')^\frac12\Big\|^{2\alpha}_{2\alpha}.
\end{equation}
One can always find adequate isometries $V, W_\g, W_\z$, and $W_{\z\g}$ numerically via the spectral decomposition, as shown in Ref. \cite{hu2022}. The numerical solution inevitably has numerical noise, however, which might negatively affect the performance of the solver. It is frequently the case, though, that the structure of the maps is simple enough that one can find suitable isometries analytically. Indeed we do so in all the examples shown in Sec \ref{Sec:Experiments}.

\subsection{The Rényi QKD cones}

The expressions \eqref{eq:truepsi} and \eqref{eq:fastpsi} are suitable for convex optimization. They can, for example, be used to implement a Frank-Wolfe algorithm with facial reduction. However, our goal here is to implement a conic method, and for that we need to represent these expressions via convex cones:
\begin{equation}\label{eq: qkd cone}
    \mc{K}^{\alpha, \gmap, \zmap, S}_\text{RényiQKD} = \cl\left\{(u,\rho,\sigma) \in \R \times \mathbb H^q_\succeq \times \mathbb H^m_\succ; u \ge s_\alpha \Psihat(\rho, \sigma) \right\},
\end{equation}
suitable for optimizing expression \eqref{eq:mainopt}, and
\begin{equation}\label{eq: fast qkd cone}
    \mc{K}^{\alpha, \gmap, \zgmap, S}_\text{FastRényiQKD} = \cl\left\{(u,\rho) \in \R \times \mathbb H^q_\succ; u \ge s_\alpha \Psihat(\rho) \right\},
\end{equation}
suitable for optimizing the lower bound \eqref{eq:fastmainopt}. The latter we denominate fast cone as it has lower dimension and is thus considerably faster to optimize over. In both expressions $s_\alpha = \operatorname{sign}(\alpha-1)$ is a convenience factor to make the cones convex for all $\alpha \ge 1/2$ \cite{frank2013}. Note, however, that for computing key rates the relevant ranges are $(2/3,1)$ and $(1/2,1)$, respectively, as given by Theorem \ref{Th:MainTheorem} and Corollary \ref{Th:MainCorollary}.

In order to implement optimization over these cones, the most crucial part is to find logarithmically homogeneous self-concordant barriers for them (see definition \ref{Def: LHSCB}).
We use the natural barriers
\begin{gather}
    f(u, \rho, \sigma) = -\log(u - s_\alpha \Psihat(\rho, \sigma)) - \logdet(\rho) - \logdet(\sigma) \label{eq:truebarrier},\\
    f(u, \rho) = -\log(u - s_\alpha \Psihat(\rho)) - \logdet(\rho) \label{eq:fastbarrier},
\end{gather}
which can be proven to be self-concordant for $\alpha \in [1/2, 1)$ following the argument of Ref. \cite[Theorem 1.2]{he2025}, where it is shown that the natural barrier of the quasi-Rényi cone 
\begin{equation} \label{eq:quasi Renyi cone}
    \mc{K}^\alpha_\text{quasi-Rényi} = \cl\left\{(u,\rho,\sigma) \in \R \times \mathbb H^n_\succeq \times \mathbb H^n_\succeq; u \ge s_\alpha \Psi_\alpha(\rho, \sigma) \right\}
\end{equation}
is self-concordant \footnote{One might hope to be able to use this cone directly for optimizing \eqref{eq:truepsi} or \eqref{eq:fastpsi}, but this is not possible, as they are not expressible in terms of $\Psi_\alpha(\rho, \sigma)$. Analogously, the von Neumann QKD cone is not expressible in terms of the relative entropy cone \cite{Lorente2025quantumkey}.}. The key difference in our proof is that we need to do a composition with the strictly positive map $X \mapsto S^\dagger X S$ when using Ref. \cite[Theorem 2.1]{HiaiTr2}. Note that this map is strictly positive because $S$ is an isometry. The complete proof is shown in Appendix \ref{sec:selfconcordance}.

It is also necessary to compute the gradient, Hessian, and third-order derivatives of the barriers to run the conic solver. This computation also follows along the lines of Ref. \cite{he2025}, except that we need to find an analogue of the identities \eqref{eq:equivalence} where all the operators involved are full-rank. It is given by
\begin{subequations}\label{eq:goodequivalence}
\begin{align}
     \Psihat(\rho,\sigma) &= \Tr\Big[\big( (S^\dagger\zmap(\sigma)^\frac{1-\alpha}{\alpha}S)^\frac12\gmap(\rho) (S^\dagger\zmap(\sigma)^\frac{1-\alpha}{\alpha}S)^\frac12\big)^\alpha\Big] \label{eq:szs} \\
     &= \quad \Big\|\zmap( \sigma )^{\frac{1-\alpha}{2\alpha}}S \gmap(\rho)^\frac12\Big\|^{2\alpha}_{2\alpha} \\
     &= \Tr\Big[\big(\gmap(\rho)^\frac12 S^\dagger \zmap(\sigma)^\frac{1-\alpha}{\alpha}S\gmap(\rho)^\frac12\big)^\alpha\Big].
\end{align}
\end{subequations}
Note that the right-hand side of Eq. \eqref{eq:szs} simplifies to $\Tr\Big[\big(\zmap(\sigma)^\frac{1-\alpha}{2\alpha}S\gmap(\rho) S^\dagger \zmap(\sigma)^\frac{1-\alpha}{2\alpha}\big)^\alpha\Big]$, but the simpler expression is not usable as the matrix inside the trace has null eigenvalues that make the derivatives ill-defined. The complete derivation is shown in Appendix \ref{sec:derivatives}. Finally, we need to compute a starting point for the optimization of each cone, which is shown in Appendix \ref{sec:startingpoint}.

We implemented these cones in the programming language Julia \cite{JuliaLang}, as an extension of the solver Hypatia \cite{coey2022solving}. It is interfaced with the modeller JuMP \cite{Lubin2023} and uses a generic type system, allowing the computation of key rates with precision greater than 64 bits when desired. The code is part of the ConicQKD.jl package \cite{git}.

\subsection{Explicit conic formulation}\label{App:NumFormulation}

We are now in a position to explicitly reformulate the optimization problems \eqref{eq:mainopt} and \eqref{eq:fastmainopt} as conic programs, suitable for our conic solver. First of all, let us rewrite \eqref{eq:mainopt} with all its constraints explicitly
\begin{equation}\label{eq:SecondMEATmain}
\begin{gathered}
    h^\uparrow_{\alpha} \geq \min_{q,\omega,\psi}   \frac{\alpha}{\alpha - 1} D_\mathrm{KL}({q}\|\tilde{\mc{M}}(\omega)_{C}) +  q(\perp) D_{\gamma} (\g(\omega)\|\z(\psi)) \\
    \mathrm{s.t.} \quad \omega_{AB}, \psi_{RAB} \succeq 0,  \\
    \Tr_B[\omega_{AB}] = \sigma_{A}, \; \Tr[\psi_{RAB}] = 1,  \\
    \sum_{c \in {\mc{C}}} q(c) = 1,\;q \ge 0, \\
    \;  \lVert p - q \rVert_1 \leq \delta.
\end{gathered}
\end{equation}
Here, we point out that the Kullback-Leibler divergence and the 1-norm can be expressed via the well-known cones \cite{coey2022solving} 
\begin{gather}
    \mc{K}_\mathrm{KL} = \cl\{(u,q,p) \in \R \times \R_>^d \times \R_>^d ; u \geq D_\mathrm{KL}(q\|p)\}, \nonumber \\
    \mc{K}_1  = \left\{(u,w) \in \R \times \R^d : u \geq \lVert w \rVert_1 \right\}.
\end{gather}
In conjunction with the RényiQKD cone \eqref{eq: qkd cone}, this allows us to rewrite \eqref{eq:SecondMEATmain} as 
\begin{equation}\label{eq:Conic_1}
\begin{gathered}
    h^\uparrow_{\alpha} \geq 
    \min_{h_\mathrm{KL},u, q, \omega, \psi}  \frac{\alpha}{\alpha - 1} h_\mathrm{KL} + \frac{ q(\perp)}{\gamma -1} \log \left(s_\gamma u\right) \\ %
    \mathrm{s.t.} \quad \Tr_B[\omega_{AB}] = \sigma_{A}, \; \Tr[\psi_{RAB}]=1,\; \sum_{c \in {\mc{C}}} q(c) = 1 \\
    (u,\omega_{AB},\psi_{RAB}) \in \mc{K}^{\gamma, \gmap, \zmap, S}_\text{RényiQKD} \\
    \left(h_\mathrm{KL}, q, \mc{\tilde{M}}(\omega_{AB})_C \right) \in \mc{K}_\mathrm{KL} \\
    (\delta,p - q) \in \mc{K}^1.
\end{gathered}
\end{equation}
Note that the factor $s_\gamma = -1$ is necessary to cancel out the factor $s_\gamma$ inside the definition of the RényiQKD cone, and positivity constraints are implicitly enforced by the cones.

We are still not done because the objective function must be affine. The logarithm term can be handled as before by using the logarithm cone \footnote{This cone is better known as the exponential cone, formulated in terms of the exponential function, but both formulations are equivalent.}
\begin{equation}
    \mc{K}_{\log} = \cl\left\{ (u, v, w) \in \mathbb{R} \times \mathbb{R}_> \times \mathbb{R}_> : u \leq v \log\left(w/v\right) \right\},
\end{equation}
but the term $q(\bot)$ multiplying the logarithm requires more care. One possibility is to do a nested optimization, where the inner is a conic optimization with fixed $q(\bot)$, and the outer optimizes over the single parameter $q(\bot)$ via a non-convex method. Another, more practical option is to use the bound on the probabilities at the acceptance set
\begin{equation}
    p(\bot)  - q(\perp) \le \lVert p - q \rVert_1 \leq \delta,
\end{equation}
to lower-bound $q(\perp)$ at the objective function, resulting in the relaxed problem
\begin{equation}\label{eq:Conic_2}
\begin{gathered}
    h^\uparrow_{\alpha} \geq 
    \min_{h_\mathrm{KL},h_\mathrm{QKD}, u, q, \omega, \psi}  \frac{\alpha}{\alpha - 1} h_\mathrm{KL} + \frac{p(\perp)-\delta}{\gamma-1} h_\mathrm{QKD} \\ %
    \mathrm{s.t.} \quad \Tr_B[\omega_{AB}] = \sigma_{A}, \; \Tr[\psi_{RAB}]=1,\; \sum_{c \in {\mc{C}}} q(c) = 1 \\
    (u,\omega_{AB},\psi_{RAB}) \in \mc{K}^{\gamma, \gmap, \zmap, S}_\text{RényiQKD} \\
    \left(h_\mathrm{KL}, q, \mc{\tilde{M}}(\omega_{AB})_C \right) \in \mc{K}_\mathrm{KL} \\
    (h_\mathrm{QKD},1,s_\gamma u) \in \mc{K}_{\log} \\
    (\delta,p - q) \in \mc{K}^1.
\end{gathered}
\end{equation}
Numerical experiments show that the effect of this relaxation on the key rate is small but noticeable for small block sizes. A more elegant solution can be applied to the FastRényiQKD cone \eqref{eq: fast qkd cone}, following the argument from Ref. \cite[Appendix F]{kamin25MEATsecurity}. Let $q = (q_{\tilde{C}}, q(\bot))$ and $\tilde{\mc{M}}(\omega)_{C} = (\tilde{\mc{M}}(\omega)_{\tilde{C}}, p(\bot))$ where ${\tilde{C}}$ is defined on the alphabet $\tilde{\mc{C}}$, so that
\begin{equation}
    D_\mathrm{KL}({q}\|\tilde{\mc{M}}(\omega)_{C}) = D_\mathrm{KL}(q_{\tilde{C}}\|\tilde{\mc{M}}(\omega)_{\tilde{C}}) + q(\bot)\log\de{\frac{q(\bot)}{p(\bot)}}.
\end{equation}
Then write $\widehat\Psi_\beta = \frac{p(\bot)\widehat\Psi_\beta}{q(\bot)}\frac{q(\bot)}{p(\bot)}$, so that
\begin{equation}
    q(\bot)\log(\widehat\Psi_\beta) = q(\bot)\log\de{\frac{ p(\bot)\widehat\Psi_\beta}{q(\bot)}} + q(\bot)\log_2\de{\frac{q(\bot)}{p(\bot)}}
\end{equation}
Then the expression in \eqref{eq:fastmainopt} can be rewritten as
\begin{equation}
    \frac{\alpha}{\alpha -1} D_\mathrm{KL}({q}\|\tilde{\mc{M}}(\omega)_{C}) + \frac{q(\perp)}{\beta-1} \log_2 \widehat\Psi_\beta(\omega) = \frac{\alpha}{\alpha -1}\de{D_\mathrm{KL}(q_{\tilde{C}}\|\tilde{\mc{M}}(\omega)_{\tilde{C}}) - q(\bot)\log\de{\frac{p(\bot)\widehat\Psi_\beta}{q(\bot)}}}
\end{equation}
and using the same cones as before gives us the exact conic formulation:
\begin{equation}\label{eq:Conic_3}
\begin{gathered}
    h^\downarrow_{\alpha} \geq 
    \min_{h_\mathrm{QKD},h_\mathrm{KL},u,q,\omega}  \frac{\alpha}{\alpha - 1} \de{h_\mathrm{KL} - h_\mathrm{QKD}} \\ %
    \mathrm{s.t.} \quad \Tr_B[\omega_{AB}] = \sigma_{A}, \; \sum_{c \in {\mc{C}}} q(c) = 1 \\
    (u,\omega_{AB}) \in \mc{K}^{\beta, \gmap, \zgmap, S}_\text{FastRényiQKD} \\
    \left(h_\mathrm{KL}, q_{\tilde{C}}, \tilde{\mc{M}}(\omega)_{\tilde{C}} \right) \in \mc{K}_\mathrm{KL} \\
    (h_\mathrm{QKD},q(\bot), p(\bot) s_\beta u) \in \mc{K}_{\log}  \\
    (\delta,p - q) \in \mc{K}^1.
\end{gathered}
\end{equation}
\section{Examples}\label{Sec:Experiments}

In this section, we define diverse QKD protocols, the necessary tools to adapt them into the formalism of the MEAT, and compute key rates for them using our technique. We exclusively use the FastRényiQKD cone \eqref{eq: fast qkd cone}, because it is significantly faster and more stable than the RényiQKD cone \eqref{eq: qkd cone}, and as illustrated in Fig \ref{fig:MUB2}, the key rates they produce are essentially indistinguishable.

Unless explicitly stated otherwise, we assume the following values for the tolerance coefficients
\begin{subequations}\label{eq:eps honest value}
    \begin{align}
        \varepsilon_\mathrm{EC} &= 10^{-11},\\
        \varepsilon_\mathrm{PA} &= 9 \times 10^{-11}, \\
        \bar{\varepsilon}_\mathrm{PE} &= 9 \times 10^{-11},
    \end{align}
\end{subequations}
and set $\bar{\varepsilon}_\mathrm{PE} = 2^{|\mc{C}|}  e^{-n \delta^2/2}$, which defines the value of $\delta$. On the other hand, we define for each protocol the corresponding entanglement-based picture, the state $\ket{\sigma}_{AA'}$ prepared by Alice in every round, the channel representing Alice and Bob's measurements, and the maps $\g$ and $\z$ representing the key distillation. Additionally, the entanglement-based approach enforces a constraint on Alice's marginal $\sigma_A$ in the minimization (since said register is not accessible to Eve), which may require a further facial reduction on the state. 

\subsection{Qubit BB84}

Following the protocol outlined in Ref. \cite[Sec VII]{kamin25MEATsecurity}, we consider an implementation of the qubit BB84 protocol \cite{BB84} that accounts for no-click events \footnote{We note, in particular, that the protocol outline in Ref. \cite{kamin25MEATsecurity} is compatible with Ref. \cite[Corollary 4.2]{arqand2025MEAT} as explained in Ref. \cite[Sec IV.B]{kamin25MEATsecurity}.}. Under the source-replacement scheme, qubit BB84 can be reformulated as an entanglement-based protocol \cite{E91,BBM92}, in which Alice prepares the maximally entangled state
\begin{align}
    \ket{\sigma}_{AA'} =& \frac{1}{\sqrt{2}} (\ket{0}\otimes \ket{0} + \ket{1}\otimes \ket{1})_{AA'},
\end{align}
so that her marginal state is maximally mixed, $\sigma_A = \mathds{1}/2$. To model imperfections, we consider that the state sent by Alice in the subsystem $A'$ is subjected to a depolarizing channel
\begin{align}
    \mathcal{E} (\rho) = v \rho + (1-v) \Tr[\rho] \frac{\mathds{1}}{d},
\end{align}
where $v \in [0,1]$ describes the visibility. Subsequently, Alice sends her state through a lossy channel with transmittance $\eta = 10^{-\frac{\chi}{10}}$, where $\chi$ represents the total channel transmittance in dB.

In each round, Alice chooses with probability $p^K$ to measure her subsystem $A$ in the $Z$-basis, which will be used exclusively for key generation rounds. With probability $1-p^K$, she instead measures in the $X$-basis, which is reserved for parameter estimation. The measurements performed by Alice to generate the state are given by the operators
\begin{align}
    & P^{X,0} = \frac{1-\pkey}{2}\begin{pmatrix}
        1 & 1 \\
        1 & 1
    \end{pmatrix}, \quad P^{X,1} = \frac{1-\pkey}{2}\begin{pmatrix}
        1 & -1 \\
        -1 & 1
    \end{pmatrix}, \nonumber \\
    & P^{Z,0} = \pkey \begin{pmatrix}
        1 & 0 \\
        0 & 0
    \end{pmatrix}, \quad P^{Z,1} = \pkey
    \begin{pmatrix}
        0 & 0 \\
        0 & 1
    \end{pmatrix}.
\end{align}
Bob performs the same measurements with the same probabilities as Alice, but also includes the possibility of observing a no-click event denoted by the symbol $\perp$. Thus, Bob's measurement operators are given by
\begin{align}
     Q^{X,0}  &= \frac{1-\pkey}{2} \begin{pmatrix}
        1 & 1 & 0 \\
        1 & 1 & 0 \\
        0 & 0 & 0
    \end{pmatrix}, \quad Q^{X,1} = \frac{1-\pkey}{2} \begin{pmatrix}
        1 & -1 & 0 \\
        -1 & 1 & 0 \\
        0 & 0 & 0
    \end{pmatrix},  \nonumber \\
    Q^{Z,0}  &= \pkey
    \begin{pmatrix}
        1 & 0 & 0 \\
        0 & 0 & 0 \\
        0 & 0 & 0
    \end{pmatrix}, \quad Q^{Z,1} = \pkey \begin{pmatrix}
        0 & 0 & 0 \\
        0 & 1 & 0 \\
        0 & 0 & 0
    \end{pmatrix},  \nonumber \\
     Q^{Z, \perp} \;  &= \; \pkey \begin{pmatrix}
        0 & 0 & 0 \\
        0 & 0 & 0 \\
        0 & 0 & 1 
    \end{pmatrix}
    , \quad
    Q^{X, \perp} =  (1-\pkey)\begin{pmatrix}
        0 & 0 & 0 \\
        0 & 0 & 0 \\
        0 & 0 & 1  
    \end{pmatrix}.
\end{align}
Once all quantum states have been sent and measured, Alice and Bob proceed with the public announcements. She announces her basis choice $Z$ for key rounds, while for parameter estimation rounds she reveals her basis choice $X$ together with her measurement outcome (either 0 or 1).  Bob's announcement depends on Alice's choice of measurement; if she indicates a key generation round and Bob measures in the $Z$-basis, he discloses only whether a detection occurred, using the symbol $I=\top$ for a successful detection and $I=\perp$ for a non-detection event. If their bases do not coincide, the round is discarded and Bob sets $I=\perp$ as well. For parameter estimation rounds, he announces both his basis, which could be either $Z$ or $X$, and the corresponding outcome, which may be 0, 1, or $\perp$.

From these announcements, Alice and Bob construct the register $C$ with alphabet $\mc{C} = \{\perp\} \cup (\{0,1\} \times \{(Z,0),(Z,1),(Z,\perp),(X,0),(X,1),(X,\perp) \} )$ where $C=\perp$ correspond to a key generation round, independent of Bob's basis choice. We also define the restricted alphabet $\tilde{\mc{C}} = \mc{C} \setminus \{\perp\}$. On the other hand, Alice stores her key in the secret register $S$. Whenever she announces a generation round and Bob measures in the $Z$ basis with a detection event, she sets $S\in\{0,1\}$ according to her measurement outcome; for no-clicks, she assigns $S=\perp$. Given this description, the channel is as follows 
\begin{align} \label{eq:RateBoundingChannel_BB84}
    \tilde{\mc{M}}_{AB\to SCI}(\cdot) =&  \sum_{s=0,1} \Tr_{AB} \left[P_A^{Z,s} \otimes Q_B^{Z,\top} (\cdot) \right] \otimes \pro{s}_S \otimes \pro{\perp}_{C} \otimes \pro{\top}_I  \nonumber \\ 
    & + \sum_{s=0,1} \Tr_{AB} \left[P_A^{Z,s} \otimes Q_B^{Z,\bot} (\cdot) \right] \otimes \pro{\perp}_S \otimes \pro{\bot}_{C} \otimes \pro{\bot}_I   \nonumber \\ 
    & + \sum_{s=0,1} \sum_{b=0,1,\perp} \Tr_{AB} \left[P_A^{Z,s} \otimes Q_B^{X,b} (\cdot) \right] \otimes \pro{\perp}_S \otimes \pro{\bot}_{C} \otimes \pro{\bot}_I   \nonumber \\ 
    &+ \sum_{a,(y,b)\in \tilde{\mc{C}}} \Tr[P_A^{X,a} \otimes Q_B^{y,b} (\cdot)] \pro{\perp}_S \otimes \pro{a,(y,b)}_C \otimes \pro{\bot}_I .
\end{align}
From this structure, we can find the trace-preserving key map 
\begin{equation}
    \Tr_I[\g(\omega)_{ABRI}] = \g(\omega) = \g_\top(\omega) + \g_\bot(\omega),
\end{equation}
where  $\g_\top$ and $\g_\bot$ have a single Kraus operator each, namely,
\begin{align}
    G_\top &= \sqrt{\pkey} \sum_{r=0,1} \ket{r}_R  \otimes \pro{r}_A \otimes (\pro{0} + \pro{1})_B, \nonumber\\
    G_\bot &=  \ket{\perp}_R \otimes \id_A \otimes \left[\sqrt{1-p^K}(\pro{0} + \pro{1})_B + \pro{\bot}_B\right].
\end{align}
In turn, the pinching map is defined as $\z(\psi) = \z_\top(\psi)+\z_\bot(\psi)$, where $\z_\top$ has Kraus operators $\pro{r}_R\otimes \id_{AB}$, with $r \in \{0,1\}$, and $\z_\bot$ has a single Kraus operator $\pro{\bot}_R\otimes \id_{AB}$. Note that these maps obey the structure introduced in Sec \ref{sec:tracepreserving}, and therefore, we can use the simplifications discussed there.

For the RényiQKD cone \eqref{eq: qkd cone} we must perform the facial reduction of the maps $\g_\top, \z_\top$ and $\g_\bot, \z_\bot$. In the case of $\g_\top$, we require an isometry $W_{\g_\top}$ such that $W_{\g_\top} W_{\g_\top}^\dagger$ is a projector onto the support of $G_\top G_\top^\dagger$. A suitable choice is
\begin{equation}
    W_{\g_\top} = \sum_{r=0,1}\ket{r}_R \otimes \pro{r}_A  \otimes (\ket{0}_{B}\bra{0}_{B'} + \ket{1}_{B}\bra{1}_{B'}),
\end{equation}
where $\{\ket{0}_{B'}, \ket{1}_{B'}\}$ is an orthonormal basis. Then, $\gmap_\top$ has a single Kraus operator given by
\begin{equation}
    W_{\g_\top}^\dagger G_\top = \sqrt{p^K}\id_A \otimes (\ket{0}_{B'}\bra{0}_{B} + \ket{1}_{B'}\bra{1}_{B}),
\end{equation}
which satisfies \eqref{eq:greduction} when $V = \id$. In its turn, $\z_\top$ is strictly positive when restricted to the subspace $\{\ket{0},\ket{1}\}$ of the $R$ subsystem. The isometry that does the job is therefore
\begin{equation}
    W_{\z_\top} = (\ket{0}_{R}\bra{0}_{R'} + \ket{1}_{R}\bra{1}_{R'}) \otimes \id_{AB},
\end{equation}
where $\{\ket{0}_{R'}, \ket{1}_{R'}\}$ is an orthonormal basis. Thus, $\zmap_\top$ has Kraus operators $\ket{r}_{R'}\bra{r}_{R} \otimes \id_{AB}$ for $r \in \{0,1\}$, which can be verified to satisfy equation \eqref{eq:zreduction}. The facial reduction of $\g_\bot$ and $\z_\bot$ is trivial: we eliminate the $R$ subsystem employing $W_{\g_\bot} = W_{\z_\bot} = \ket{\bot}_R \otimes \id_{AB}$. Hence, $\gmap_\bot$ has a single Kraus operator $\id_A \otimes \Big[\sqrt{1-p^K}(\pro{0} + \pro{1})_B + \pro{\bot}_B\Big]$, and $\zmap_\bot$ has a single Kraus operator $\bra{\bot}_R \otimes \id_{AB}$. All together, we use the facial reduction of the maps $\g_\top, \z_\top$ and $\g_\bot, \z_\bot$ with parameters $S=W_{\z_\top}^\dagger W_{\g_\top}$ and $S=W_{\z_\bot}^\dagger W_{\g_\bot} = \id_{AB}$, respectively.

\begin{figure}[h]
	\centering
 	\begin{tikzpicture}
		\begin{axis}[%
			scale only axis,
            ymode = log,
			xmin=0,
			xmax=32,
			ymin=1e-6,
			ymax=1e-0,
			grid=major,
			xlabel={$\chi$ (dB)},
            ylabel = {Secret key rate (bits/pulse)},
			axis background/.style={fill=white},
			legend style={at={(0.97,0.97)},legend cell align=left, align=left, draw=white!15!black, font=\footnotesize}
			]
            \addplot[black!70, mark=*] table[col sep=comma] {plot_data/bb84/bb84_NInf.csv};
        	\addlegendentry{$n \rightarrow \infty$ }
             \addplot[purple, mark=*] table[col sep=comma] {plot_data/bb84/bb84_N1e9_fast.csv};
        	\addlegendentry{$n = 10^{9}$ }
            \addplot[olive!90, mark=*] table[col sep=comma] {plot_data/bb84/bb84_N1e7_fast.csv};
        	\addlegendentry{$n = 10^{7}$}
            \addplot[teal, mark=*] table[col sep=comma] {plot_data/bb84/bb84_N1e6_fast.csv};
        	\addlegendentry{$n = 10^{6}$}
            \addplot[blue!70, mark=*] table[col sep=comma] {plot_data/bb84/bb84_N1e5_fast.csv};
        	\addlegendentry{$n = 10^{5}$}
        \end{axis}
	\end{tikzpicture}
    \caption{Finite secret key rate for qubit BB84 protocol using the FastRényiQKD cone for a different number of rounds $n$. All curves consider $v=0.97$, and $f=1.16$. The probability $p^K$ was optimized for each point according to a coarse-grained tuning, while the Rényi parameter $\alpha$ was optimized numerically. The asymptotic key rate is included as a reference.}
    \label{fig:bb84 real estimator}
\end{figure} 

For the FastRényiQKD cone \eqref{eq: fast qkd cone}, we only need the facial reduction of $\g_\top$ and $\z_\top \circ \g_\top$ (see Sec \ref{sec:tracepreserving}). While the former is already done above, the latter is given by the map $\zgmap_\top$ with Kraus operators $\sqrt{p^K}\pro{r}_A \otimes (\ket{0}_{B'}\bra{0}_{B}+ \ket{1}_{B'}\bra{1}_{B})$ for $r \in \{0,1\}$. One can check that equation \eqref{eq:zgreduction} is respected with $W_{\z_\top\g_\top} = W_{\g_\top}$, and thus the parameter $S$ is given by $W_{\z_\top\g_\top}^\dagger W_{\g_\top} = \id_{AB'}$. \\

 Our results are shown in Fig \ref{fig:bb84 real estimator}, where we plot the secret key rates against the total transmittance $\chi$ for the qubit BB84 protocol via the FastRényiQKD cone. Here we consider $f=1.16$, different numbers of rounds $n$, visibility $v=0.97$, and optimal values for $\alpha$ and $\pkey$ in every point. We observe that our implementation reaches 14 dB of total loss for a number of rounds $n \leq 10^6$, compatible with current industrial capacity \cite{Walenta_2014}. To further validate our approach, we implement the variable-length QKD framework described in Ref. \cite[Sec VII]{kamin25MEATsecurity}. In contrast to the fixed-length scenario discussed above, this formulation requires solving two distinct optimization problems to compute the final key rate (see Appendix \ref{App:VarLength}). As shown in Fig \ref{fig:bb84comparisonWaterloo}, this improved numerical convergence leads to slightly better key rates and an increase in maximal transmittance. Moreover, the use of our conic methods provides a neat speedup as indicated in Sec \ref{Subsec:Benchmarks}.

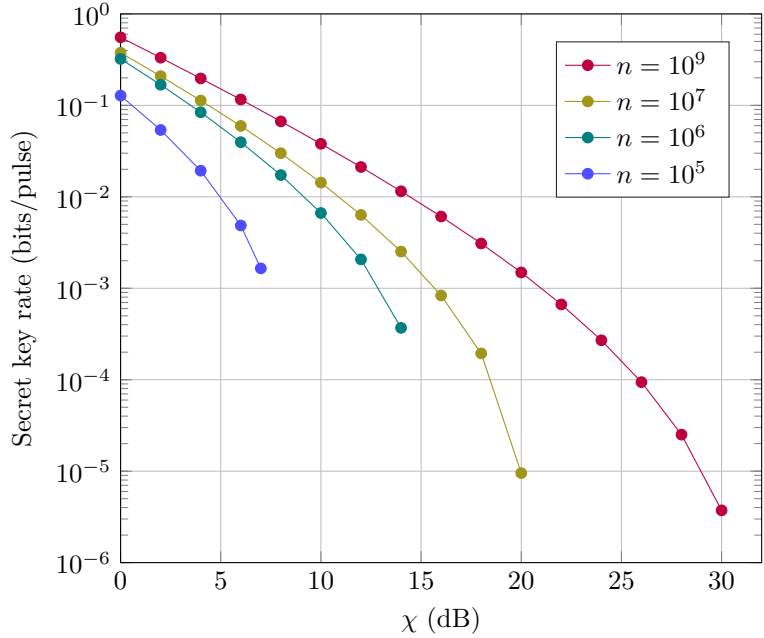
\begin{figure}[h!]
	\centering
 	\begin{tikzpicture}
		\begin{axis}[%
			scale only axis,
            ymode = log,
			xmin=0,
			xmax=55,
			ymin=1e-7,
			ymax=1e-0,
			grid=major,
			xlabel={Transmittance $\chi$ (dB)},
            ylabel = {Secret key rate (bits/pulse)},
			axis background/.style={fill=white},
			legend style={at={(0.97,0.97)},legend cell align=left, align=left, draw=white!15!black, font=\footnotesize}
			]
            \addplot[black!70, mark=*] table[col sep=comma] {plot_data/bb84/bb84_NInf.csv};
            \addlegendentry{$n \rightarrow \infty$}
            \addplot[purple!50, mark=*] table[col sep=comma] {plot_data/bb84/bb84_N1e9.csv};
        	\addlegendentry{$n = 10^{9}$}
            \addplot[purple!100, mark=x,mark options={solid}, dashed] table[col sep=comma] {plot_data/bb84/waterloo_N1e9.csv};
        	\addlegendentry{$n = 10^{9}$ \cite{kamin25MEATsecurity}}
            \addplot[olive!50, mark=*] table[col sep=comma] {plot_data/bb84/bb84_N1e7.csv};
        	\addlegendentry{$n = 10^{7}$}
            \addplot[olive, mark=x,mark options={solid}, dashed] table[col sep=comma] {plot_data/bb84/waterloo_N1e7.csv};
        	\addlegendentry{$n = 10^{7}$ \cite{kamin25MEATsecurity}}
            \addplot[blue!40, mark=*] table[col sep=comma] {plot_data/bb84/bb84_N1e5.csv};
        	\addlegendentry{$n = 10^{5}$}
            \addplot[blue,mark=x,mark options={solid}, dashed] table[col sep=comma] {plot_data/bb84/waterloo_N1e5.csv};
        	\addlegendentry{$n = 10^{5}$ \cite{kamin25MEATsecurity}}
        \end{axis}
	\end{tikzpicture}
    \caption{Variable-length finite secret key rate for qubit BB84 protocol using the FastRényiQKD cone (solid lines), compared with the results showcased in Ref. \cite{kamin25MEATsecurity} (dashed lines) for different numbers of rounds $n$. All curves consider $v=0.97$, and $f=1.16$. The probability $p^K$ was optimized for each point according to a coarse-grained tuning, while the Rényi parameter $\alpha$ was optimized numerically. For a fair comparison, we set the security parameters to $\varepsilon_\mathrm{PA} = \varepsilon_\mathrm{EC} = \frac{1}{2}10^{-80}$. The asymptotic key rate is included as a reference.} \label{fig:bb84comparisonWaterloo}
\end{figure} 
\subsection{MUBs}

The protocol described in Ref. \cite{Araujo2023quantumkey} (originally outlined in Ref. \cite{Cerf2001,Sheridan2010}) consists of Alice and Bob using a set of $m$ MUBs to measure a state whose dimension is $d$, such that $2 \leq m \leq d+1$. This protocol is directly provided in an entanglement-based picture, where the initial state prepared by an external source is the generalized, maximally entangled state 
\begin{equation}
    \ket{\sigma}_{AA'} = \frac{1}{\sqrt{d}} \sum_{j=0}^{d-1} \ket{j}_{A} \otimes \ket{j}_{A'}.
\end{equation}
This state when distributed to Alice and Bob, evolves into an isotropic state $\sigma^\mathrm{iso}_{AA'}(v) = v{\sigma}_{AA'} + (1-v) \mathds{1}/d^2$, where Alice's marginal is always $\sigma_A = \mathds{1}/d$. Following Ref. \cite{Araujo2023quantumkey}, the key is distilled according to Alice's projective measurements on the computational basis $\{\pro{s}_A\}_{s=0}^{d-1}$. On the other hand, since no-click events are not considered in this formulation (or in a broader sense, postselection events), register $I$ becomes trivial and thus can be eliminated for this protocol.

Parameter estimation is performed by Alice using the POVMs $\{\{\Pi^{j,l} \}_{j=0}^{d-1}\}_{l=0}^{m-1}$ where $\Pi^{j,l}$ denotes the $j-$th vector of the $l-$th MUB (in particular, $l=0$ represents the computational basis). Bob's POVMs are the transposed version of Alice's. The choice of basis is defined via the probability vector
\begin{equation}
    \{p(l)\}_{l=0}^{m-1} = (\pkey, (1-\pkey)/(m-1),\ldots,(1-\pkey)/(m-1)).
\end{equation}
Hence, only rounds where both Alice and Bob use the zeroth basis will be spent for the key distillation, implying that the sifting probability is given by $p^\mathrm{s} = (\pkey)^2$. Note further that the original implementation requires only the statistics generated by Alice and Bob when their basis choices coincide, so we can merge the statistics from all other cases into a single probability outcome for parameter estimation. All in all, we define the alphabet 
$\mc{C} =  \{\perp, 1,\ldots,m\}$, such that $\tilde{\mc{C}} = \mc{C} \backslash \{\perp\}$. 
Furthermore, $\mc{S} = \{\perp,0,\ldots,d-1\}$ with $\tilde{\mc{S}} = \mc{S} \backslash \{\perp\}$ for the secret key alphabet. According to this notation and reducing the POVMs that add up to the identity, one round of the MUB protocol is described by the following channel
\begin{align} 
    \tilde{\mc{M}}_{AB\to SC}(\cdot) =& p(0)^2  \sum_{s\in\tilde{\mc{S}}}  \Tr[\pro{s}_A \otimes \mathds{1}_B (\cdot)]\pro{s}_S \otimes \pro{\perp}_{C} \nonumber \\ 
    & + \sum_{l=1}^{m-1} \sum_{j=0}^{d-1} p(l)^2 \Tr[\Pi^{j,l} \otimes (\Pi^{j,l})^T (\cdot)] \pro{\perp}_S \otimes \pro{l}_C  \nonumber\\
    & + \left[1 - \sum_{l=0}^{m-1} p(l)^2 \right] \pro{\perp}_S \otimes \pro{m}_C .\label{eq:RateBoundingChannel_MUB}
\end{align}
The key map $\g$ is given by $\mathcal {G}(\rho) = G\rho G^\dagger$, where
\begin{equation}
    G = \sum_{r=0}^{d-1} \ket{r}_{R} \otimes \pro{r}_A \otimes \mathds{1}_B.
\end{equation}

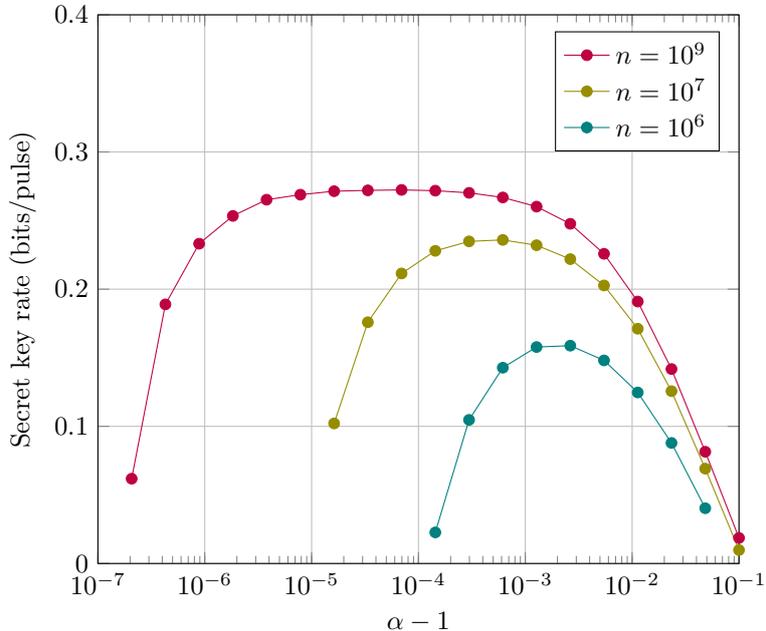
\begin{figure}[h]
	\centering
 	\begin{tikzpicture}
		\begin{axis}[%
			scale only axis,
            xmode = log,
			xmin=1e-7,
			xmax=1e-1,
			ymin=0,
			ymax=0.4,
			grid=major,
			xlabel={$\alpha-1$},
            ylabel = {Secret key rate (bits/pulse)},
			axis background/.style={fill=white},
            legend pos=north east,
        ]
            \addplot[black!70, mark=*] table[col sep=comma]{plot_data/mub/mub_asym.csv};
            \addlegendentry{$n \rightarrow \infty$}
            \addplot[purple, mark=*] table[col sep=comma]{plot_data/mub/mub_1e9.txt};
        	\addlegendentry{$n = 10^{9}$ }
            \addplot[olive, mark=*] table[col sep=comma] {plot_data/mub/mub_1e7.txt};
        	\addlegendentry{$n = 10^{7}$ }
            \addplot[teal, mark=*] table[col sep=comma] {plot_data/mub/mub_1e6.txt};
        	\addlegendentry{$n = 10^{6}$ }
        \end{axis}
	\end{tikzpicture}
    \caption{Finite secret key generation rate for the MUB protocol using the FastRényiQKD \eqref{eq: fast qkd cone} cone with $d=5$, $m=6$, $v = 0.9$, $f = 1.16$, $\pkey=0.5$, and different values for the number of rounds $n$ and the Rényi parameter $\alpha$. The asymptotic key rate is included as a reference.}
    \label{fig:MUB}
\end{figure} 

As this is an isometry, the facial reduction is trivial: $\gmap$ is the identity map, and thus equation \eqref{eq:greduction} is satisfied with $V = \id$ and $W_\g = G$. In turn, $\mc{Z}$ is provided by the usual dephasing superoperators $Z^r =\pro{r}_R \otimes \mathds{1}_{AB}$ for $r \in \{0,\ldots,d-1\}$, and is already a strictly positive map, so equation \eqref{eq:zreduction} is satisfied with $\zmap = \z$ and $W_\z = \id$. The last parameter we need for the RényiQKD cone \eqref{eq: qkd cone} is $S = W_\z^\dagger W_\g = G$. For the FastRényiQKD cone \eqref{eq: fast qkd cone} we need in addition the facial reduction of $\zg$. It is the map $\zgmap$ with Kraus operators $\pro{r}_A \otimes \mathds{1}_{B}$ for $r \in \{0,\ldots,d-1\}$, and equation \eqref{eq:zgreduction} is satisfied with $W_{\z\g}=G$, and thus the parameter $S$ is simply $W_{\z\g}^\dagger W_\g = G^\dagger G = \id$.

In Fig \ref{fig:MUB} we illustrate the secret key generation rates against the Rényi parameter $\alpha$ for diverse values of the block size $n$ with parameters $d=5$, $m=6$, and $v = 0.9$ for the FastRényiQKD cone \eqref{eq: fast qkd cone}. Overall, we observe (as is the case for the other protocols) that the optimal value for $\alpha$ does not require fine-tuning beyond the order of magnitude to yield values close to the actual maximum at the curves. Fig \ref{fig:MUB2} compares the FastRényiQKD cone \eqref{eq: fast qkd cone} with the RényiQKD cone \eqref{eq: qkd cone} for $\alpha$ varying from $1.1$ to $2$. We see that the RényiQKD cone delivers superior performance, as expected, but the advantage is diminutive and only visible for values of $\alpha$ where the bounds on the secret key rate drop below zero. We emphasize that negative key rates have no operational meaning in QKD, and we consider this regime strictly to compare the two cones. In order to deliver a fair comparison, we did not use the relaxed formulation \eqref{eq:Conic_2} for the RényiQKD cone, but the nested optimization discussed just above it.

\begin{figure}[h!]
	\centering
 	\begin{tikzpicture}
		\begin{axis}[%
			scale only axis,
			xmin=0.1,
			xmax=1,
			grid=major,
			xlabel={$\alpha-1$},
            ylabel = {Secret key rate (bits/pulse)},
			axis background/.style={fill=white},
            legend pos=north east,
        ]
            \addplot[purple, mark=*] table[col sep=comma] {plot_data/mub_comparison/mub_1e9.txt};
            \addlegendentry{FastRényiQKD}
            \addplot[purple!50, dashed, mark=*] table[col sep=comma] {plot_data/mub_comparison/mub_1e9_true.txt};
            \addlegendentry{RényiQKD}
        \end{axis}
	\end{tikzpicture}
    \caption{Lower bound on the finite secret key rate for the MUB protocol using the RényiQKD \eqref{eq: qkd cone} and FastRényiQKD \eqref{eq: fast qkd cone} cones with $d=5$, $m=6$, $v = 0.9$, $f = 1.16$, $\pkey=0.5$, $n = 10^9$, and Rényi parameter $\alpha$ varying from $1.1$ to $2$. We emphasize that this range of $\alpha$ is not relevant for secret key generation, and negative key rates have no operational meaning. The purpose here is only to illustrate the difference between the two cones.}
    \label{fig:MUB2}
\end{figure}
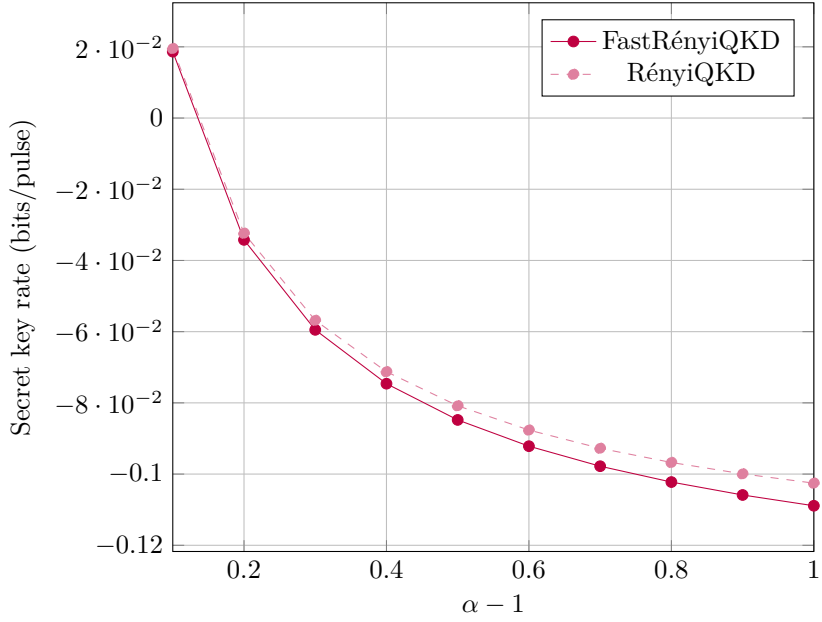 

\subsection{DMCV}

For the QKD protocol based on a discrete-modulated continuous variable (DMCV) approach, we implement a variant \cite{baeuml2023security} of the well-known quadrature phase shift keying (QPSK) \cite{Leverrier2011,ghorai2019asymptotic,lin2019asymptotic} based on Alice sampling states with equal probabilities from a distribution of coherent states $\{\ket{i^x \gamma}\}_{x=0}^3$, and Bob implementing a full discretization for the measurement outcomes. We caution the reader to understand $S$ as the key register generated by Bob, as it is customary in CVQKD due to its reverse nature for information reconciliation \footnote{Similarly, the entropy $H(S|Y)$ in Inequality \eqref{eq:EC_P2P} should be taken with $S$ as Bob's secret key after the public announcements, and $Y$ as Alice's.}.

Under the entanglement-based picture, Alice will send for all rounds the state
\begin{equation}
    \ket{\sigma}_{AA'} = \frac{1}{2} \sum_{x=0}^3 \ket{x}_A \otimes \ket{i^x \gamma}_{A'}
\end{equation}
from which her marginal follows directly. While Alice generates her key registers via measurements on her computational basis, Bob extracts said register using a discretization according to a coarse-grained heterodyne detection
\begin{equation}
    \bar{R} = \frac{1}{\pi} \int_{\mc{R}} \pro{\gamma} d\gamma^2.
\end{equation}
Here, $\bar{R}$ is known as a region operator, and $\mathcal{R}\subseteq \mathbb{C}$ an element of a partition in $\mathbb{C}$, such that the collection of region operators completely covers the phase space, constituting a POVM. In particular, we consider the sets of region operators $\{\hat{R}^z_B\}_{z=0}^3$ and $\{\tilde{R}^z_B\}_{z=0}^5$ for key and parameter estimation rounds, respectively, defined in Ref. \cite{pascualgarcia2024} according to a modulation $(\Delta,\Delta_s)$. We refer the reader to Ref. \cite{pascualgarcia2024} for further implementation specifics.

Provided these considerations we define the alphabets $\mc{C} = \{\perp,\perp\} \cup (\{0,\ldots,3\} \times \{0,\ldots,5\})$ and $\mc{S} = \{\perp,0,\ldots,3\}$. We can now formulate the map constituting the DMCV protocol via
\begin{align} \label{eq:RateBoundingChannel2}
    \mc{M}_{AB\to SC}(\cdot) =& p^\mathrm{K} \sum_{s\in\tilde{\mc{S}}}  \Tr[\mathds{1}_A\otimes \hat{R}_B^s (\cdot)]\pro{s}_S \otimes \pro{\perp,\perp}_{C} \nonumber \\
    &+ (1-p^\mathrm{K})\sum_{(x,z)\in \tilde{\mc{C}}} \Tr[\pro{x}_A\otimes \tilde{R}_B^z (\cdot)]\pro{\perp}_S \otimes \pro{x,z}_{C}.
\end{align}
As in the MUB protocol, we have that no postselection is applied on the key distillation, therefore $I$ can be removed. Furthermore, Alice always implements the same measurement such that the sifting probability coincides with the key probability $p^\mathrm{s} = p^\mathrm{K}$. From the key generation rounds, we can build the key map $\g(\rho) = G \rho G^\dagger$ where \cite{lin2019asymptotic}
\begin{equation}
    G = \sum_{r=0}^3 \ket{r}_R \otimes \mathds{1}_A \otimes \sqrt{\hat{R}^r_B}.
\end{equation}
The $\z$ map is again provided by a set of pinching operators $Z^r=\pro{r}_R \otimes \mathds{1}_{AB}$ for $r\in \{0,\ldots,3\}$. 

\begin{figure}[t]
	\centering
 	\begin{tikzpicture}
		\begin{axis}[%
			scale only axis,
            ymode = log,
			xmin=0,
			xmax=40,
			ymin=1e-3,
			ymax=1e-0,
			grid=major,
			xlabel={$L$ (km)},
            ylabel = {Secret key rate (bits/pulse)},
			axis background/.style={fill=white},
			legend style={at={(0.97,0.97)},legend cell align=left, align=left, draw=white!15!black, font=\footnotesize}
			]
            \addplot[black!70, mark=*] table[col sep=comma] {plot_data/dmcv/dmcv_Asym.csv};
        	\addlegendentry{$n \rightarrow \infty$ }
            \addplot[red, mark=*] table[col sep=comma] {plot_data/dmcv/dmcv_1e10.csv};
        	\addlegendentry{$n = 10^{10}$}
            \addplot[purple, mark=*] table[col sep=comma] {plot_data/dmcv/dmcv_1e9.csv};
        	\addlegendentry{$n = 10^{9}$}
            \addplot[violet, mark=*] table[col sep=comma] {plot_data/dmcv/dmcv_1e8.csv};
        	\addlegendentry{$n = 10^{8}$}
            \addplot[olive, mark=*] table[col sep=comma] {plot_data/dmcv/dmcv_1e7.csv};
        	\addlegendentry{$n = 10^{7}$}
            \addplot[teal, mark=*] table[col sep=comma] {plot_data/dmcv/dmcv_5e6.csv};
        	\addlegendentry{$n = 5\times 10^{6}$}
            \addplot[blue!70, mark=*] table[col sep=comma] {plot_data/dmcv/dmcv_3e6.csv};
        	\addlegendentry{$n = 3\times 10^{6}$}
        \end{axis}
	\end{tikzpicture}
    \caption{Secret key generation rates for the DMCV protocol using the FastRényiQKD cone \eqref{eq: fast qkd cone} for different values of the number of rounds $n$, using $N_c = 10$, $\Delta = 4.0$ and $\Delta_s = 1.5$ and error correction at the Shannon limit. $\alpha$, $\pkey$ and the coherent state amplitude were optimized for each point according to a coarse-grained tuning. The asymptotic key rate is included as a reference.} \label{fig:DMCV}
\end{figure}
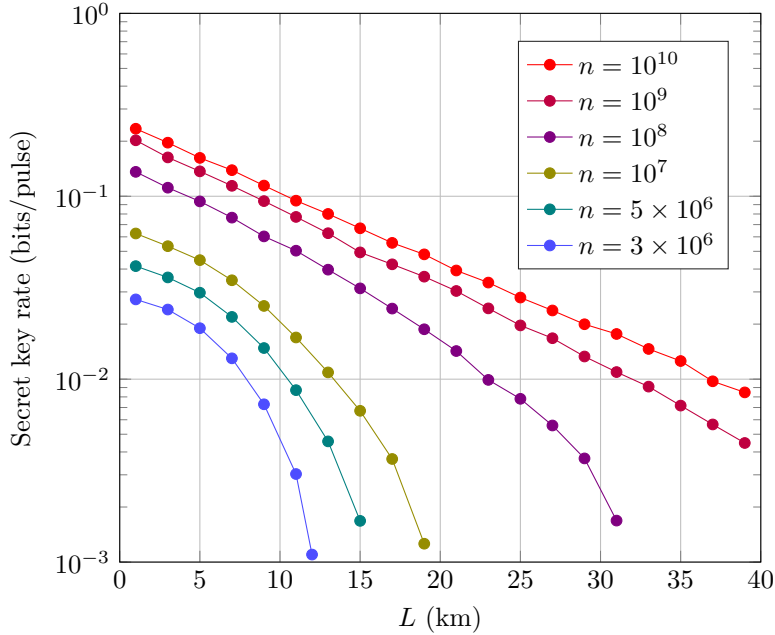 

In order to do the facial reduction, we first note that $\g$ is an isometry \footnote{This is not true when postselection is included, such as in Ref. \cite{kanitschar2023finite, Navarro2025}.}. Therefore $\gmap$ is the identity map, and equation \eqref{eq:greduction} is satisfied with $V = \id$ and $W_\g = G$. In its turn, $\z$ is already a strictly positive map, so equation \eqref{eq:zreduction} is satisfied with $\zmap = \z$ and $W_\z = \id$. Then, the isometry $S$ is given by $S = W_\z^\dagger W_\g = G$, and so the facial reduction for the RényiQKD cone \eqref{eq: qkd cone} is done. For the FastRényiQKD cone \eqref{eq: fast qkd cone} we need in addition the facial reduction of $\zg$, which is trivial, since $\zg$ is already strictly positive \cite{Lorente2025quantumkey}. Thus, $\zgmap = \zg$ and equation \eqref{eq:zgreduction} is satisfied with $W_{\z\g}=\id$, implying $S=G$ once again.

Fig \ref{fig:DMCV} provides several finite secret key rates, using the modulation presented in Ref. \cite{pascualgarcia2024} while adjusting all other parameters to their optimal value. In particular, we observe that key rates can be generated with $n = 5 \times 10^6$ at 15 km, which improves in orders of magnitude over Ref. \cite{pascualgarcia2024}, where $n = 5 \times 10^8$ was reported as the minimum block size. On the other hand, our method provides an enhanced numerical scalability with respect to other approaches \cite{Araujo2023quantumkey,Lorente2025quantumkey} and does not require a min-tradeoff linearization, which should allow a removal of the cutoff assumption without incurring in unfortunate penalties \cite{primaatmaja2024}. 

\section{Benchmarks} \label{Subsec:Benchmarks}

We conclude our numerical analysis by showcasing the performance of our algorithm in contrast with the repository Renyi-security-framework \cite{git_waterloo} and the package QICS \cite{QICS} (used in Refs. \cite{kamin25MEATsecurity} and \cite{he2025} respectively), which are also capable of optimizing the sandwiched Rényi entropy.

\begin{table}[h]
\centering 
\setlength{\tabcolsep}{6pt} 
\renewcommand{\arraystretch}{1.2} 
\begin{tabular}{lcccccccc} 
\toprule
& \multicolumn{2}{c}{FastRényiQKD}   &  \multicolumn{2}{c}{RényiQKD} & \multicolumn{2}{c}{QICS \cite{QICS}} \\
\cmidrule(lr){2-3} \cmidrule(lr){4-5}  \cmidrule(lr){6-7} 
$N_c$ & Iterations & Time $[s]$ & Iterations & Time $[s]$ & Iterations & Time $[s]$\\
\midrule
5 & 67 & 16.87 & 70 & 594.7 & 53 & 1731  \\ 
7 & 54 & 49.64 & 80 & 3374 & 60 & 8701 \\ 
10 & 60 & 347.1 & 73 & 18430 & \multicolumn{2}{c}{Out of RAM} \\ 
\bottomrule
\end{tabular}
\caption{Average iterations and time to optimize the cones corresponding to the DMCV protocol. For each value of $N_c$, the iterations and time are averaged over 5 runs with the same parameters: $L=13$ km, $n=10^7$, $f=1.0$, $\Delta s = 1.5$, and $\Delta = 4.0$. In the more relevant case $N_c = 10$ \cite{lin2019asymptotic,hu2022,kanitschar2023finite}, QICS results in the computer running out of RAM due to the higher dimension of the cones, which forces it to use SWAP memory, making the optimization very slow. The benchmarks were performed using ConicQKD.jl v0.3.0 and QICS v1.1.3 \cite{QICS} on a workstation equipped with an Apple M4 Max and 128 GB DDR5 RAM.}
\label{table:conic-performance dmcv}
\end{table}

In Table \ref{table:conic-performance dmcv}, we compare the time and iterations taken to optimize the cones corresponding to the DMCV protocol between our implementation and QICS \cite{QICS}. Since the Rényi cone in QICS does not include facial reduction, a small depolarization noise ($\epsilon=10^{-8}$) was introduced in $\g$ to avoid null eigenvalues, which would make the gradient undefined and the optimization unstable. Moreover, given that QICS exploits the sparsity in the optimization of the Rényi entropy cone, we also exploited the block structure of the $\zmap$ in the FastRényiQKD and RényiQKD cones to make a fair comparison. 

\begin{table}[h]
\centering 
\setlength{\tabcolsep}{6pt} 
\renewcommand{\arraystretch}{1.2} 
\begin{tabular}{lcccc} 
\toprule
& \multicolumn{2}{c}{FastRényiQKD} & \multicolumn{2}{c}{Rényi-security-framework \cite{git_waterloo}} \\
\cmidrule(lr){2-3} \cmidrule(lr){4-5}
$n$ & Iterations & Time $[s]$ & Iterations & Time $[s]$ \\
\midrule
$10^9$ & 62 & 0.133 & 60 & 34.427 \\ 
$10^7$ & 42 & 0.098 & 60 & 32.945  \\ 
$10^5$ & 36 & 0.049 & 60 & 19.682 \\ 
\bottomrule
\end{tabular}
\caption{Average iterations and computational time over 20 runs of the qubit BB84 protocol at transmittance $\chi=8$ dB and different block-sizes. FastRényiQKD benchmarks were executed with ConicQKD.jl v0.3.0, whereas the Rényi-security-framework \cite{git_waterloo} was evaluated using MATLAB R2024a. All computations were performed on a workstation equipped with an 11th Gen Intel Core i5-11500 processor and 16 GB of RAM.}
\label{table:conic-performance}
\end{table}

As we can see, our RényiQKD cone \eqref{eq: qkd cone} is considerably faster than QICS. This is due to a factor of 4 reduction in the output dimension of $\g$ caused by the facial reduction. Additionally, switching to the FastRényiQKD \eqref{eq: fast qkd cone}, which is not available in QICS, provides a dramatic speedup with respect to the RényiQKD cone. Regarding stability, in some runs (for example, for $N_c = 7$, $n=10^7$, and $L=9$ km) the optimization in QICS failed to converge, unless we increased the depolarization noise.

In Table \ref{table:conic-performance}, we provide a numerical comparison between our FastRényiQKD cone \eqref{eq: fast qkd cone} and the repository Renyi-security-framework \cite{git_waterloo} used by Ref. \cite{kamin25MEATsecurity}, which implements the Frank-Wolfe algorithm \cite{frankwolfe56,winick2018reliable}. Unlike our approach, this method needs to add an artificial depolarization noise ($\epsilon=10^{-10}$ in this specific realization) to prevent undefined gradients. The values reported represent the average iterations and computation time required to solve the two optimization problems leading to the final key rate (see Appendix \ref{App:VarLength}), averaged over 20 runs of the variable-length qubit BB84 protocol at transmittance $\chi=8$ dB. The results indicate that our conic method is orders of magnitude faster than the Frank-Wolfe approach.

\section{Conclusion}\label{Sec:Conclusion}

This work introduces the most advanced framework for QKD, combining the improved privacy amplification via Rényi hashing \cite{Dupuis23RenyiHashLemma} and the security proof provided by the MEAT \cite{arqand2025MEAT}, which admits the calculation of secret key rates under finite-size effects and general security by means of sandwiched Rényi entropies. The core technical contribution is the introduction of two convex cones that allow the fast and reliable optimization of the conditional Rényi entropy via non-symmetric conic optimization.

We have tested this framework using instances of different protocols, including the full-discrete QPSK protocol from the DMCV family \cite{baeuml2023security,pascualgarcia2024}, which is well-known for its resource-intensive numerical implementation \cite{primaatmaja2024}. In said case, we have observed secret key rates at smaller block sizes thanks to the tightness of the finite-size analysis of the MEAT. Further improvements can be achieved by using more sophisticated approaches, such as a fine-tuned postselection on Bob's measurements \cite{kanitschar2023finite, Navarro2025}. Furthermore, we have applied our method to the qubit BB84 protocol using the variable-length framework described in Ref. \cite{kamin25MEATsecurity}. We benchmarked our conic approach in these two protocols against QICS \cite{QICS} and the Frank-Wolfe algorithm \cite{frankwolfe56,winick2018reliable} used in previous implementations, respectively. We observed orders of magnitude faster computation time, as well as improved stability, and slightly better secret key rates and critical transmittances.

Regarding future work, it would be interesting to apply our conic methods to compute key rates in more complex and realistic protocols, such as DMCV without a cutoff assumption \cite{Navarro2025} or BB84 using decoy states. From a practical perspective, we note that our framework opens the possibility of performing QKD implementations using reduced block sizes and ensuring long-distance general security without the repetition-rate limitations of previous entropy accumulation techniques \cite{arqand2024generalized,metger2022security}, which paves the path for a practical, reliable, and large-scale implementation for QKD.

\section*{Acknowledgements}

We thank Lars Kamin and Ernest Tan for insightful discussions about entropy accumulation, variable-length approaches, and theoretical aspects of QKD. MN and CPG were supported by the Government of Spain (Severo Ochoa CEX2019-000910-S, FUNQIP and NextGeneration EU PRTR-C17.I1) and European Union (QSNP, 101114043). MN acknowledges funding from the European Union’s Horizon Europe research and innovation programme under the MSCA Grant Agreement No. 101081441. CPG has received funding from the European Union’s Digital Europe Programme under the project QUARTER (101091588) and Horizon 2020 Research and Innovation Programme under the project QSNP (101114043), and from the European Innovation Council’s Horizon Europe EIC Accelerator Programme under the project MIQRO (101161539).
The research of AGL, PVP, and MA was supported by the Q-CAYLE project, funded by the European Union-Next Generation UE/MICIU/Plan de Recuperación, Transformación y Resiliencia/Junta de Castilla y León (PRTRC17.11), and also by the Department of Education of the Junta de Castilla y León and FEDER Funds (Reference: CLU-2023-1-05). PVP has also been funded under the UVa 2024 predoctoral contract, co-financed by Banco Santander. MA was also supported by the Spanish Agencia Estatal de Investigación, Grant Nos.~RYC2023-044074-I and PID2024-161725OA-I00.

\section*{Code availability}
The implementation of the cones introduced in this work is available in the package ConicQKD.jl \cite{git}. The results presented in Sec \ref{Sec:Experiments} can be reproduced using the repository in Ref. \cite{git_examles}. The comparative results and benchmarks presented in Fig \ref{fig:bb84comparisonWaterloo} and Table \ref{table:conic-performance} were obtained using the repository in Ref. \cite{git_waterloo}, which is heavily based on the Open QKD Security package \cite{burniston_2024_14262569}.

\appendix
\counterwithin*{equation}{section}         
\renewcommand{\theequation}{\thesection\arabic{equation}}

\section{Marginal-constrained entropy accumulation theorem}\label{App:MEAT}
This section provides a theoretical security analysis based on the marginal-constrained entropy accumulation theorem (MEAT) \cite{arqand2025MEAT} and the required notation changes to adapt it to our approach. As a first step, we define the corresponding results of the MEAT itself, considering that no public announcements are made before all quantum signals have been measured. In addition, we apply the special implementation provided by Ref. \cite[Theorem 4.2b]{arqand2025MEAT}, in which only the classical register $C$ is considered for test rounds.

\begin{proposition} \label{prop:MEAT} \cite[Corollary 4.2, Theorem 4.2a]{arqand2025MEAT}
For each $j \in \{1,...,n\}$, take a state $\sigma^{j} \in \mc{D}(A_{j})$ and a CPTP map $\tilde{\mc{M}}_j: [AB]_{j} \to [SCI]_j$, such that $C_j$ and $I_j$ are classical. Let $\rho_{[SCI]^{(n)} \hat{E}} = \tilde{\mc{M}}_n \circ ... \circ \tilde{\mc{M}}_1 (\bar{\omega}) $ for $\bar{\omega} \in \mc{D} ([AB]^{(n)} \hat{E})$, verifying $\bar{\omega}_{A_1^{n}}= \sigma^{1}_{A_1}\otimes ... \otimes \sigma^{n}_{A_{n}}$.

Suppose furthermore that $\rho = p_\Omega \rho_{|\Omega} + (1-p_{\Omega}) \rho_{|\neg \Omega}$ for $p_\Omega \in (0,1]$ and all $C_j$ are isomorphic to a single register $C$ with alphabet $\mc{C}$. Let $S_\Omega$ be the convex set of probability distributions defined on the alphabet $\mathcal{C}$ such that any $c^{(n)}$ with nonzero probability in $\rho_{|\Omega}$ verifies that $\mathrm{freq}_{c^{(n)}} \in S_\Omega$. Then, for any $\alpha \in (1,\infty)$,
\begin{equation}\label{eq:Renyi_h}
    H^\uparrow_{\alpha}(S^{(n)}|[CI]^{(n)}\hat{E})_{\rho_{|\Omega}} \geq n h^\uparrow_{\alpha} - \frac{\alpha}{\alpha-1} \log \frac{1}{p_{\Omega}},
    \end{equation}
 where
    \begin{equation} \label{eq:OriginalMEAT}
    h^\uparrow_{\alpha} = \inf_{{q}\in S_{\Omega}} \inf_{\nu\in\Sigma_j} \frac{\alpha}{\alpha-1} D_\mathrm{KL}({q}\|{\nu}_{C}) + \sum_{c \in \mathrm{supp}({\nu}_{C})} q(c) H_\alpha^\uparrow(S| I E)_{\nu_{|c}}
    \end{equation}
provided that $E$ is a large-enough purifying register for any $[AB]_{j}$, and $\Sigma_j$ the set of all states $\tilde{\mc{M}}_j (\omega_{[AB]_{j}E})$ for some initial $\omega \in  \mc{D}([AB]_{j}E)$ such that $\omega_{A_{j}}=\sigma_{A_{j}}^{j}$. 
\end{proposition}

Here, the register $I_j$ stores the public announcements, with alphabet $\mathcal{I}$ satisfying $\mathcal{C \subseteq I}$. In particular, the equality holds whenever there are no public announcements in key generation rounds. For further reference, we introduce a change of notation by writing the probability $p_\Omega$ as $\Pr[\Omega]$, and the set $\Sigma$ as $\Sigma_\Omega$.  

For simplicity, in the following discussion we assume that all maps are identical, i.e. $\tilde{\mc{M}}_1 = \dots = \tilde{\mc{M}}_n = \tilde{\mc{M}}$,  and consequently, so is the marginal $\sigma_A$. We then divide map $\tilde{\mc{M}}$ in two channels as 
\begin{align}
    \tilde{\mc{M}}(\cdot)_{SIC} =  p(\bot)\mc{M}^\mathrm{K}(\cdot)_{SI} \otimes \pro{\perp}_C + [1-p(\bot)]\pro{\perp}_S \otimes \mc{M}^\mathrm{T}(\cdot)_{I\tilde{C}} 
\end{align}
where $\tilde{C} := C\backslash\{\perp\}$, and $\mc{M}^\mathrm{K}(\cdot)$ and $\mc{M}^\mathrm{T}(\cdot)$ are quantum channels for key generation and parameter estimation, respectively. Given this redefined channel and provided that $\nu_C = \Tr_{SI}[\tilde{\mc{M}}(\omega)_{SCI}]$, for some $\omega \in \mc{D}(AB)$, the Kullback-Leibler divergence in eq. \eqref{eq:OriginalMEAT} can be written as 
\begin{align}
    D_\mathrm{KL}({q}\|{\nu}_{C})  &= q(\perp) \log\left(\frac{q(\perp)}{p(\bot) \mc{M}^\mathrm{K}(\omega_{AB})_{|\perp} }\right) \nonumber \\
    & \quad + \sum_{c \in \tilde{\mc{C}}} q(c) \log\left(\frac{q(c)}{[1-p(\bot)]\mc{M}^\mathrm{T}(\omega_{AB})_{|c} }\right).
\end{align}
For the right-hand side in \eqref{eq:OriginalMEAT}, the sandwiched Rényi entropy can be simplified as 
\begin{equation} \label{eq:OnlyKeyRounds}
     \sum_{c \in \mathrm{supp}({\nu}_{C})} q(c) H_\alpha^\uparrow(S|I E)_{\nu_{|c}} \geq   q(\perp) H_\alpha^\uparrow(S| IE)_{\nu_{|\perp}}.
\end{equation}
Within our discussion, this relation holds as an equality because  $S = \perp$ when $C \neq \perp$. Moreover, note that Eve also holds a copy $\tilde{I}$ of the classical register $I$, which implies 
\begin{align}
    H_\alpha^\uparrow(S|IE)_{\nu_{|\perp}} = H_\alpha^\uparrow(S|\tilde{E})_{\nu_{|\perp}}=H_\alpha^\uparrow(S|\tilde{I}E)_{\nu_{|\perp}}, \label{eq:Final H(S|IE)}
\end{align}
where $E$ represents Eve's quantum side information, and $\tilde{E}:=\tilde{I}E$ all her quantum and classical information. The conditional entropy in \eqref{eq:Final H(S|IE)} can be further lower-bounded by following Ref. \cite[Supplementary material I]{chung2025}. First, note that our final state for a key round $\rho_{RABSI}$ comes as the combination of two actions: a coherent key map, acting on $AB$ to produce a key reference $R$ and a public discussion $I$; and a pinching operation that tells Alice and Bob how to distil the raw key at register $S$ according to the key reference $R$. Hence, we have
\begin{itemize}
    \item The quantum state shared by Alice and Bob $\omega_{AB}$.
    \item The state after the key isometry $\sigma_{RABI}$, and in particular its marginal $\g(\omega_{AB}) = \Tr_I[\sigma_{RABI}]$. 
    \item The final state after the key extraction, $\rho_{SRABI} = V(\sigma_{RABI})V^\dagger$ where $V$ acts on register $R$, as well as its marginal $\rho_{SRAB} =  V(\sigma_{RAB})V^\dagger$. 
\end{itemize}
In particular, $V = \sum_{s} \ket{s}_S \otimes Z^s_{R}$ where $Z_{R}$ is a projective operator, and $R$ is an intermediate key register between Alice and Bob, which may include postselection and/or removal of no-clicks. 

Let us recall, for $\alpha\geq 0$, the up-arrow sandwiched Rényi entropy
\begin{align}
    H^\uparrow_\alpha(A|B)_\rho := - \inf_{\psi \in \mc{D}(B)} D_\alpha (\rho_{AB}\|\mathds{1}_A \otimes \psi_B),
\end{align}
and define the down-arrow sandwiched Rényi entropy
\begin{align} \label{eq:DownArrowRenyi}
    H^\downarrow_\alpha(A|B)_\rho := - D_\alpha (\rho_{AB}\|\mathds{1}_A \otimes \rho_B),
\end{align}
as well as the up-arrow Petz-Rényi entropy
\begin{align}
    \bar{H}^\uparrow_\alpha(A|B)_\rho:= - \inf_{\psi \in \mc{D}(B)} \bar{D}_\alpha (\rho_{AB}\|\mathds{1}_A \otimes \psi_B)
\end{align}
where $\bar{D}_\alpha(\rho\|\psi)$ is the Petz-Rényi divergence. Using these definitions and considering $\rho_{SRABI\tilde{E}}$ as the purification of the state $\rho_{SRABI}$, we find for $\gamma = \alpha/(2 \alpha -1)$ when $\alpha \in (1/2,\infty]$ that \cite[Supplementary material I]{chung2025}
\begin{align}
     H^\uparrow_\alpha(S|\tilde{E})_{\nu_{|\bot}} &= - H^\uparrow_{\gamma} (S|RAB)_{\nu_{|\bot}} \nonumber \\ 
     &=  \inf_{\psi \in \mc{D}(RAB)} D_{\gamma} (\rho_{SRAB}\|\mathds{1}_S \otimes \psi_{RAB}) \nonumber\\
     & = \inf_{\psi \in \mc{D}(RAB)} D_{\gamma} (V \sigma_{RAB} V^\dagger  \| \id_S \otimes \psi_{RAB})  \nonumber\\
     & \geq \inf_{\psi \in \mc{D}(RAB)} D_{\gamma} (\sigma_{RAB}\| V^\dagger [\id_S \otimes \psi_{RAB}] V)\nonumber \\
     &= \inf_{\psi \in \mc{D}(RAB)} D_{\gamma} (\sigma_{RAB}\| \sum_s Z^s_{R} \psi_{RAB} Z_{R}^s )\nonumber \\
     &= \inf_{\psi \in \mc{D}(RAB)} D_{\gamma} (\g(\omega_{AB})\| \z (\psi_{RAB})). \label{eq:DivRenyiTrue}
\end{align}
Here, we used the fact that $\tilde{I} = I$ and the duality relation shown in Ref. \cite[Theorem 3]{Tomamichel2014}, while the inequality arises due to the data processing inequality \footnote{Ref. \cite[Eqs. (35,40)]{chung2025} also uses a further reduction $\psi_{RAB} = \g(\psi_{AB})$ in the minimization. However, this assumption is not mathematically sound, as it unnecessarily restricts the optimization to a subset of the valid state space.}. We also note that Ref. \cite{chung2025} proves that the inequality can be set as an equality, which requires deriving an upper bound for the up-arrow Rényi entropy.

Alternatively, we can lower-bound the up-arrow Rényi entropy in \eqref{eq:Final H(S|IE)} using the down-arrow entropy \eqref{eq:DownArrowRenyi}. This eventually defines a new bound $h^\downarrow_\alpha$ at the finite-size analysis, which yields faster numerical results at the cost of a suboptimal secret key rate. This bound goes as follows
\begin{align}
    H^\uparrow_\alpha(S|\tilde{E})_{\nu_{|\bot}} &\geq \bar{H}^\uparrow_\alpha (S|\tilde{E})_{\nu_{|\bot}} \nonumber\\
    & = - H^\downarrow_\beta (S|RAB)_{\nu_{|\bot}} \nonumber\\
    & = D_\beta(\rho_{SRAB}\| \mathds{1}_S \otimes \rho_{RAB}) \nonumber\\
    & = D_\beta (V \sigma_{RAB} V^\dagger  \| \id_S \otimes \rho_{RAB}) \nonumber \\
    & \geq D_\beta (\sigma_{RAB}\| V^\dagger [\id_S \otimes \rho_{RAB}] V) \nonumber \\
    &= D_\beta (\sigma_{RAB}\| \sum_s Z^s_{R} \rho_{RAB} Z_{R}^s ) \nonumber \\
    &= D_\beta (\g(\omega_{AB})\| \zg (\omega_{AB})). \label{eq:DivFastRenyi}
\end{align}
For the first line, we used Ref. \cite[Corollary 4]{Tomamichel2014}, while in the second we again applied duality relations shown in Ref. \cite[Theorem 3]{Tomamichel2014} for $\beta =1/\alpha$ whenever $\alpha \in (0,\infty)$. The second inequality arises due to the data processing inequality, and the last line comes from calculating $\rho_{RAB} = \Tr_{SI}[\rho_{RABSI}]$. \\

With the new notation and bounds, we redefine Proposition \ref{prop:MEAT} as
\begin{proposition}
 Let $n \in \mathbb{N}$, $\rho_{[SCI]^{(n)}\hat{E}} = \tilde{\mathcal{M}}^\otimes (\bar{\omega})$ for some $\bar{\omega} \in \mathcal{D}([AB]^{(n)}\hat{E})$ and $\mc{\tilde{M}} : AB \to SCI$ a quantum channel, and $\sigma\in \mc{D}(A)$. Given a nonabortion event $\Omega$, we have 
\begin{align}\label{eq:PropositionMEAT}
    H^\uparrow_\alpha(S^{(n)}|[CI]^{(n)} 
    \hat{E})_{\rho | \Omega} 
    \geq n h^\uparrow_\alpha - 
    \frac{\alpha}{\alpha-1} \log \frac{1}{\Pr[\Omega]}
\end{align}
where $\alpha \in (1,\infty)$ and
\begin{align}
     h^\uparrow_\alpha \geq \inf_{q \in \hat{\Sigma}_\Omega } \inf_{\omega \in \mc{D}(AB)} \inf_{\psi \in \mc{D}(RAB)} \frac{\alpha}{\alpha -1} D_\mathrm{KL}({q}\|\tilde{\mc{M}}(\omega)_{C}) + q(\perp) D_\gamma (\mc{G}(\omega)\|\z{(\psi)} )
\end{align}
for $\hat{\Sigma}_\Omega$  the set in \eqref{eq:AcceptanceSet}, all $\omega$ such that $\Tr_B[ \omega] = \sigma_A$ and $\gamma=\alpha/(2 \alpha-1)$. Furthermore, the bound $h^\uparrow_\alpha \geq h^\downarrow_\alpha$ holds at \eqref{eq:PropositionMEAT} where
\begin{align}
     h^\downarrow_\alpha \geq \inf_{q \in \hat{\Sigma}_\Omega } \inf_{\omega \in \mc{D}(AB)} \frac{\alpha}{\alpha -1} D_\mathrm{KL}({q}\|\tilde{\mc{M}}(\omega)_{C}) + q(\perp) D_\beta (\mc{G}(\omega)\|\zg{(\omega)} ),
\end{align}
with $\beta = 1/\alpha$.
\label{prop:MEATv2}
\end{proposition}

\begin{remark}\label{remark:ClassCond}
At \eqref{eq:Final H(S|IE)}, it is possible to apply a classical conditioning on $\tilde{I}$ (e.g. to discard a key round due to a no-click). Using Ref. \cite[Proposition 9]{Mueller_Lennert_2013}, this induces 
    \begin{equation} \label{eq:ClassicalCondRenyi}
    H^\uparrow_\alpha(S| \tilde{I} E)_\rho =  \frac{\alpha}{1-\alpha} \log \left( \sum_{i \in \tilde{\mc{I}}} p(\tilde{I}= i) 2^{\frac{1-\alpha}{\alpha} H^\uparrow_\alpha(S| E)_{\rho_{|i}}} \right).
\end{equation}

This formulation constitutes an essential step in decoy-state BB84 and variable-length formulations \cite{kamin25MEATsecurity}. Within our context, and provided the aforementioned considerations, we have that the state can be decomposed as
\begin{align}
    \rho_{SRABI} = \sum_{i \in \tilde{\mc{I}} } \g^i (\omega)_{SRAB} \otimes \pro{i}_I
\end{align}
where $\g^i$ is the unnormalized CP map applied according to the value of register $I$. Let us now expand this expression by using \eqref{eq:DivRenyiTrue}. In the following, we take shorthand notation for the normalized, sandwiched Rényi divergence \eqref{eq:QRenyiDiv}
\begin{align}
    H^\uparrow_\alpha(S|E)_{\rho_{|i}} &
    \geq  \inf_{\psi \in \mc{D}(RAB)} D_{\gamma} (\g^i(\omega_{AB})\| \z (\psi_{RAB})).\nonumber \\
    &=: \frac{1}{\gamma-1}\log \left[ \frac{\bar{\Psi}_\gamma^i(\omega)}{\Tr[\g^i(\omega)]} \right],
\end{align}
where $\bar{\Psi}_\gamma^i$ already incorporates the optimization with respect to $\psi_{RAB}$, and the action of the maps $\g^i$ and $\z$. Then,
\begin{align}
     H^\uparrow_\alpha(S| \tilde{I} E)_\rho &\geq \frac{\alpha}{1-\alpha} \log \left( \sum_{i \in \tilde{\mc{I}}} \Tr[\g^i (\omega)] \left[\frac{\bar{\Psi}_\gamma^i(\omega)}{\Tr[\g^i (\omega)] } \right]^{\frac{1}{\gamma}} \right).
\end{align}
We note that each term in the sum constitutes a perspective of the function $\bar{\Psi}_\gamma^i(\omega)^\frac1\gamma$, which for $\gamma^{-1}>1$ can be modelled using the well-known power cone 
\begin{align}
    \mc{K}_p = \{(x,y,z) \in \mathbb{R} \times \mathbb{R} \times \mathbb{R}: x^p y^{1-p} \geq |z|\}.
\end{align}
This formulation provides a tighter bound than \eqref{eq:DivRenyiTrue}. Nevertheless, it does not achieve a noticeable improvement in the secret key rate, while it greatly increases the cost of the numerical implementation.

In the case of using the bound \eqref{eq:DivFastRenyi} on the exponent, we may explicitly decompose the final expression using the normalized divergence \eqref{eq:QRenyiDiv} to find
\begin{align} 
    H^\uparrow_\alpha(S| \tilde{I} E)_\rho &\geq  \frac{1}{\beta-1} \log \left( \sum_{i \in \tilde{\mc{I}}} \Tr[\g^i (\omega)]2^{\log[\Psi^i_\beta(\omega) / \Tr[\g^i (\omega)]]} \right) \nonumber  \\
    &=  \frac{1}{\beta-1} \log \left( \sum_{i \in \tilde{\mc{I}}} \Psi^i_\beta(\omega)  \right). \label{eq:Conditioning}
\end{align}
\end{remark}

\section{Proof of Theorem \ref{Th:MainTheorem}}\label{App:Proof_Th}

The security for finite-size QKD protocols is based on the leftover hash lemma \cite{tomamichel2015quantum,berta2016smooth}, which bounds the length of an extractable secret key according to the raw key given all the information accessible to Eve. Within our discussion, we use a variant of this lemma based on Rényi entropies.

\begin{proposition}
    \emph{\cite[Theorem 9]{Dupuis23RenyiHashLemma}} \label{prop:RenyiHashLemma}
    Let $\sigma_{AE} \in \mc{D}(AE)$ be a cq-state and $\{\mc{R}^h_{A \rightarrow C}(\cdot), h \in \mc{H}\}$ a family of $\lambda-$randomizing hash functions on register $A$. Then, 
    \begin{align}
        T \Big(\rho^\mc{R}_{CFE}, \frac{1}{|C|} \mathds{1}_C \otimes &\rho_{FE} \Big) \leq 
        \nonumber \\&2^{\frac{2(1-\alpha)}{\alpha}}2^{\frac{\alpha-1}{\alpha}(\log|C|-H^\uparrow_\alpha(A|E)_\sigma + 2 \log \lambda)} \label{eq:renyihashlemma}
    \end{align}
    where $\rho^\mc{R}_{CFE} = (\sum_h p(h) \mc{R}^h \otimes \pro{h}_F \otimes \mathds{1}_E)(\sigma_{AE})$, $\alpha \in (1,2)$ and $\{\ket{h}_F\}_{h \in \mc{H}}$ constitutes a basis in the register $F$ of hash functions.
\end{proposition}
This result can be adapted to our analysis by performing the following changes.

\begin{align*}
    C &\rightarrow K_B, \\
    A &\rightarrow {S}^{(n)}, \\
    E &\rightarrow E'L , \\
    \rho^{\mathcal{R}}_{CFE} &\rightarrow \rho_{K_B F E'|\Omega}, \\
    \sigma_{AE} &\rightarrow \rho_{S^{(n)} E'|\Omega},\\
    \frac{1}{|C|} \mathds{1}_C\otimes \rho_{FE} &\rightarrow \bar{\tau}_{K_B F E'|\Omega},
\end{align*}
where $E' :=[CI]^{(n)} {E}$, $\bar{\tau}$ is defined as in \eqref{eq:IdealCCQ}, $\rho$ is the state after a concatenation of $n$ channels constituting our protocol, and more explicitly, $F$ denotes the family of universal$_2$ hash functions used for error correction and privacy amplification. Furthermore, $|K_B| = 2^\ell$ since the secret key is actually a binary register of length $\ell$; and $L$ is the register of all the public information exchanged by Alice and Bob during error correction. Since we are interested in universal$_2$ families of hash functions, we set $\lambda=1$ \cite[Lemma 7]{Dupuis23RenyiHashLemma}. Therefore, \eqref{eq:renyihashlemma} can be rewritten as
\begin{align}
    T (\rho_{K_B F E'L|\Omega}, \bar{\tau}_{K_B F E'L|\Omega}) \leq  2^{\frac{1-\alpha}{\alpha} (H^\uparrow_\alpha(S^{(n)}|E'L)_{\rho|\Omega}-\ell + 2)}  \label{eq:newrenyileftoverhashlemma}
\end{align}
Let now define a scalar coefficient $\varepsilon_\mathrm{PA} > 0$  verifying that 
\begin{equation}\label{eq:BoundEpsilonPrime}
    T \left(\rho_{K_B F E'L|\Omega} , \bar{\tau}_{K_B F E'L|\Omega} \right) \leq\frac{\varepsilon_\mathrm{PA}}{\Pr[\Omega]_\rho}
\end{equation}
To ensure the above expression follows from the bound in Eq. \eqref{eq:newrenyileftoverhashlemma}, it is sufficient to require its right-hand side to be no larger than $\varepsilon_\mathrm{PA}/\Pr[\Omega]_\rho$. Solving this condition for the key length and considering that $(1-\alpha)/\alpha <0$, we have 
\begin{align} \label{eq:lenghtKey_v1}
    \ell 
    \leq H^\uparrow_\alpha &(S^{(n)}|E'L)_{\rho|\Omega} - \frac{\alpha}{\alpha-1} \log \left( \frac{1}{\varepsilon_\mathrm{PA}}\right) 
    +\frac{\alpha}{\alpha-1} \log \left( \frac{1}{\Pr[\Omega]}\right) + 2 .
\end{align}
Next, we split off register $L$ using \cite[Lemma 5.14, Lemma 5.15]{tomamichel2015quantum} and the notation leak$_\text{EC} := \log(|L|)$, such that
\begin{align}
    H^\uparrow_\alpha(S^{(n)}|E' L)_{\rho|\Omega} \geq H^\uparrow_\alpha(S^{(n)}|E')_{\rho|\Omega} - \text{leak}_\text{EC}.
\end{align}
Thus, the inequality in \eqref{eq:lenghtKey_v1} can be written as 
\begin{align}
    \ell \leq H^\uparrow_\alpha(S^{(n)}|E' )_{\rho|\Omega} - \frac{\alpha}{\alpha-1} \log \left( \frac{1}{\varepsilon_\mathrm{PA}}\right) 
    +\frac{\alpha}{\alpha-1} \log \left( \frac{1}{\Pr[\Omega]}\right)- \text{leak}_\text{EC} +2.
\end{align}
We can now apply the MEAT following Proposition \ref{prop:MEATv2}, such that the final expression of the key length is
\begin{align} \label{eq:FinalLength}
    \ell \leq&  \, n h_{\alpha} -  \frac{\alpha}{\alpha-1} \log \left(\frac{1}{\varepsilon_\mathrm{PA}} \right)- \mathrm{leak}_\mathrm{EC} + 2,
\end{align}
where $h_\alpha$ could be equal to $h^\uparrow_\alpha$ or $h^\downarrow_\alpha$. Thus, we can prove the security of Theorem \ref{Th:MainTheorem}.

\begin{proof}
    Let us consider the state after the execution of the protocol $\rho_{K_A K_B F E'L}$ and an ideal state $\bar{\tau}_{K_A K_B F E^\prime L}$. We calculate its trace distance, noting that both states are the same when the protocol aborts (albeit no key rate is produced then),
\begin{align}
    T (\rho_{K_A K_B F E^\prime L} , &\bar{\tau}_{K_A K_B F E^\prime L} ) = 
    \Pr[\Omega]_\rho T \left(\rho_{K_A K_B F E'L|\Omega} , \bar{\tau}_{K_A K_B F E^\prime L|\Omega} \right).
\end{align}
Then, we bound the resulting trace distance as follows
\begin{align} \label{eq:NonincreasingKA}
\begin{aligned} 
    T (\rho_{K_A K_B F E^\prime L|\Omega} , \bar{\tau}_{K_A K_B F E^\prime L|\Omega} )\leq & T \left(\rho_{K_A K_B F E^\prime L|\Omega} , \rho_{K_A K_B F E^\prime L|\Omega\wedge K_A = K_B} \right) 
      \\
     & + T \left(\rho_{K_A K_B F E^\prime L|\Omega\wedge K_A = K_B} , \bar{\tau}_{K_A K_B F E^\prime L|\Omega} \right)  \\
     \leq& \frac{\varepsilon_\mathrm{EC}}{\Pr[\Omega]_\rho} + T \left(\rho_{K_B F E^\prime L|\Omega} , \bar{\tau}_{K_B F E^\prime L|\Omega} \right). 
\end{aligned}
\end{align}
We applied the triangle inequality in the first line and our hypothesis of $\varepsilon_\mathrm{EC}-$correctness in the second (which is eventually reflected in the value of $\mathrm{leak}_\mathrm{EC}$), such that we can trace out register $K_A$ in the second line. For the final trace distance, we may assume the hypothetical value for $\varepsilon_\mathrm{PA} \in (0,1)$ in \eqref{eq:BoundEpsilonPrime}. Thus, 

\begin{equation}
    \Pr[\Omega]_\rho T \left(\rho_{K_B F E^\prime L|\Omega} , \bar{\tau}_{K_B F E^\prime L|\Omega} \right) \leq \varepsilon_\mathrm{PA}.
\end{equation}
Using Proposition \ref{prop:RenyiHashLemma}, we certify that the protocol is $(\varepsilon_\mathrm{EC}+ \varepsilon_\mathrm{PA})-$secure, as long as the resulting secret key has a length bounded by
\eqref{eq:FinalLength}. For any choice of parameters, the protocol verifies
\begin{align}
   \Pr[\Omega]_\rho T \left(\rho_{K_A K_B F E^\prime L|\Omega} , \bar{\tau}_{K_A K_B F E^\prime L|\Omega} \right) &\leq\varepsilon_\mathrm{EC} + \varepsilon_\mathrm{PA},
\end{align}
which constitutes our security statement, in accordance with \eqref{eq:secure}.
\end{proof}

\section{Variable-length secret key rates via conic optimization}\label{App:VarLength}

In this section, we develop a conic approach for the variable-length secret key rate for qubit BB84 using the framework described in Ref. \cite{kamin25MEATsecurity}. To do so, we first introduce a tradeoff function $f: \mathcal{C} \to \mathbb{R}$, which allows us to express the $f$-weighted Rényi entropy \cite{arqand2025MEAT}
\begin{equation}
     H^{\uparrow,f}_\alpha (S|{C}IE)_{\mathcal{M}(\omega)}= \frac{\alpha}{1-\alpha} \log \left( \sum_{{c}\in \mc{C}} p({c}) 2^{\frac{1-\alpha}{\alpha} (-f ({c})+ H^\uparrow_\alpha (S|IE))_{{\mathcal{M}(\omega)}_{|{c}}}}  \right).
\end{equation}
This function is necessary to define $\kappa$ as an ${H}^{\uparrow,f}_{\alpha}$-normalization constant corresponding to the set of all states that can be produced via a single-round map $\mc{M}$, such that
\begin{equation}
\begin{gathered}
  \kappa:= \inf_{\omega_{AB}\in\mathcal{D}(\!AB\!)} H_{\alpha}^{\uparrow,f} (S | {C}IE )_{\mathcal{M}(\omega)}, \\
\text{s.t. }\Tr_{B}[\omega_{AB}]=\sigma_{A}.
\end{gathered}
\end{equation}
Based on this, the security proof given by Ref. \cite[Theorem 12]{kamin25MEATsecurity} establishes that a variable-length protocol is secure when the final key length $\ell_\text{var}$ is bounded by a function $\hat{f}_\text{prot}(c_1^n) = \sum_i (f(c_i) + \kappa)$. As further discussed in Ref. \cite[Sec VI]{kamin25MEATsecurity}, rewriting this function in terms of the observed frequency vector and matching its expectation to a \textit{fixed} honest probability distribution $q^h$, eventually provides an average secret key rate 
\begin{align}
  \mathbb{E} \left[\ell_\text{var} \right] \ge n\bigl( f\cdot q^h +\kappa \bigr) -\mathrm{leak}_{\mathrm{EC}} -\frac{\alpha}{\alpha-1}\log\frac{1}{\varepsilon_{\mathrm{PA}}}+2,
\end{align}
where the leakage is provided by Inequality \eqref{eq:EC_P2P}, and is evaluated according to $q^h$. 
To find the optimal tradeoff function $f$, we introduce an auxiliary variable $\lambda \in \mathds{P}_\mathcal{C}$. Then, similarly to Ref. \cite{kamin25MEATsecurity}, we formulate the following fixed-length optimization program
\begin{equation}\label{eq:tradeoffMEAT}
\begin{gathered}
    \inf_{\omega, \lambda}   \frac{\alpha}{\alpha - 1} D_\mathrm{KL}(\lambda|{\mc{M}}(\omega)_{C}) +  q^h(\perp) D_{\beta} (\g(\omega)\|\zg(\omega)) \\
    \mathrm{s.t.} \quad \omega_{AB} \succeq 0, \, \Tr_B[\omega_{AB}] = \sigma_{A}\\
    \sum_{c \in {\mc{C}}} \lambda(c) = 1,\; \lambda \ge 0, \\
    \;  \lambda - q^h = 0.
\end{gathered}
\end{equation}
Using the techniques described in Sec \ref{sec:ComputeSKR}, we can reformulate this optimization as the conic program
\begin{equation}\label{eq:tradeoffConic}
\begin{gathered}
    \inf_{h_\mathrm{KL},h_\mathrm{QKD}, u, \lambda, \omega}  \frac{\alpha}{\alpha - 1} [h_\mathrm{KL} - p^\mathrm{K} h_\mathrm{QKD}] \\ %
    \mathrm{s.t.} \quad \Tr_B[\omega_{AB}] = \sigma_{A}, \; \sum_{c \in {\mc{C}}} \lambda(c) = 1,\\
    (u,\omega_{AB}) \in \mc{K}^{\beta, \gmap, \zgmap, S}_\text{FastRényiQKD} \\
    \left(h_\mathrm{KL}, \lambda, \mc{M}(\omega)_{C} \right) \in \mc{K}_\mathrm{KL} \\
    \left(h_\mathrm{QKD},1,s_\beta u + 1 - \Tr[\g^\top (\omega)]\right) \in \mathcal{K}_{\log} \\
    \lambda - q^h = 0.
\end{gathered}
\end{equation}
As stated in Ref. \cite[Lemma 4.12]{arqand2025MEAT}, the dual value of the last constraint provides an optimal choice for $f$, which can be readily introduced in the optimization for $\kappa$. Before continuing, note that $f-$weighted Rényi entropy can be lower-bounded as
\begin{align}
    H^{\uparrow,f}_\alpha (S|{C}IE)_\rho &= \frac{\alpha}{1-\alpha} \log \left( \sum_{c \in \tilde{\mc{C}} }p(c) 2^{\frac{\alpha-1}{\alpha}f (c)} + \pkey 2^{\frac{1-\alpha}{\alpha} (-f ({\perp})+ H^\uparrow_\alpha (S|IE)_{\mathcal{M}(\omega)_{|{\perp}}}})  \right) \nonumber\\
    &\geq  \frac{\alpha}{1-\alpha} \log \left( \sum_{c \in \tilde{\mc{C}} }p(c) 2^{\frac{\alpha-1}{\alpha}f (c)} + \pkey 2^{\frac{\alpha-1}{\alpha} f ({\perp})} \left[ \sum_{i \in \tilde{\mc{I}}} \Psi^i_\beta(\omega) \right] \right),
\end{align}
where in the last line we used Eq. \eqref{eq:Conditioning} from Remark \ref{remark:ClassCond}. In the particular case of qubit BB84, we have $\tilde{\mathcal{I}} = \{\top, \perp\}$. Using then the description in Sec \ref{sec:tracepreserving}, we can explicitly write the optimization of $\kappa$ in its conic form via the FastRényiQKD cone
\begin{equation}\label{eq:Conic_kappaFast}
\begin{gathered}
    \kappa \geq  \inf_{u, \omega} \frac{\alpha}{1-\alpha}  \log  \left( (1-p^{\mathrm{K}})\sum_{{c} \in \tilde{\mathcal{C}}} \Tr[\Gamma^c \omega_{AB}] 2^{\frac{\alpha-1}{\alpha}f ({c})} + p^\mathrm{K} 2^{\frac{\alpha-1}{\alpha} f ({\perp})}\left(s_\beta u + 1-\Tr[\g^\top(\omega)] \right)  \right) \\ %
        \mathrm{s.t.} \quad  \, \Tr_B[\omega_{AB}] = \sigma_{A}\\
    (u,\omega_{AB}) \in \mc{K}^{\beta, \gmap, \zgmap, S}_\text{FastRényiQKD},
\end{gathered}
\end{equation}
where $\{\Gamma^c\}_{c \in \mc{C}}$ succinctly denotes all the POVM elements used by Alice and Bob in parameter estimation rounds, weighted according to the probabilities $(\pkey,1-\pkey)$ of Bob's choices.

\section{Properties of \texorpdfstring{$S$}{S}}\label{sec:isometry}

As noted in Sec \ref{Sec:Prolegomena}, we require $\mathrm{supp}(\rho) \subseteq \mathrm{supp}(\sigma)$ in our definition of the sandwiched Rényi relative entropy. This is not necessary for the entropy itself to be well-defined in our range of interest, $\alpha \in [1/2,1)$, but we require it in order to prove that $S$ is an isometry, which is needed for the derivatives to be well-defined and the proof of self-concordance.

In terms of our maps $\g, \z$, this means that $\mathrm{supp}(\g) \subseteq \mathrm{supp}(\z)$, where
\begin{subequations}
    \begin{gather}
    \supp(\g) = \cl\left\{ \ket{\phi};\; \left\langle \phi \middle| \g(V\omega'V^\dagger) \middle| \phi \right\rangle > 0\; \forall\, \omega' \succ 0 \right\} 
	= \left\{ W_\g\mathbf{x};\; \mathbf{x}\in \C^k\right\},    \\
    \supp(\z) = \cl\left\{ \ket{\phi} ;\; \left\langle \phi \middle| \z(\psi) \middle| \phi \right\rangle > 0\; \forall\, \psi \succ 0 \right\} 
	= \left\{ W_\z\mathbf{y};\; \mathbf{y}\in \C^n\right\}
	,
    \end{gather}
\end{subequations}
\begin{proposition}\label{Sisometry}
	$S = W_\z^\dagger W_\g$ is an isometry.
\end{proposition}
\begin{proof}
	Let $\mathbf{x}\in \C^k$. $W_\g \mathbf{x}\in \supp(\g)$; since $\mathrm{supp}(\g) \subseteq \mathrm{supp}(\z)$, there exists $\mathbf{y}\in \C^n$ such that
	\begin{align}
		W_\g\mathbf{x} &= W_\z\mathbf{y}. \label{suppsubset0}
	\end{align}
	On the one hand, multiplying both sides of \eqref{suppsubset0} by $W_\g^\dagger$,
	\begin{equation}\label{suppsubset1}
		\mathbf{x} = W_\g^\dagger W_\z\mathbf{y}.
	\end{equation}
	On the other hand, multiplying both sides of \eqref{suppsubset0} by $W_\z^\dagger$,
	\begin{equation}\label{suppsubset2}
		W_\z^\dagger W_\g\mathbf{x} = \mathbf{y}.
	\end{equation}
	Then, replacing \eqref{suppsubset2} in \eqref{suppsubset1},
	\begin{equation}\label{suppsubset3}
		\mathbf{x} = W_\g^\dagger W_\z W_\z^\dagger W_\g\mathbf{x} = S^\dagger S\mathbf{x}.
	\end{equation}
	As it is true for all $\mathbf{x}\in \C^k$, it must be $S^\dagger S = \mathds{1}_k$.
\end{proof}
The proof for the $\g, \z \circ \g$ follows by replacing $W_\z$ with $W_{\z\g}$. We will also check that the map $X \mapsto S^\dagger X S$ is strictly positive, as we will need it to prove that the barrier function is self-concordant:
\begin{proposition}\label{mapS_++}
	The map $X \mapsto S^\dagger X S$ is strictly positive.
\end{proposition}
\begin{proof}
    First note that this map is positive, as $Y \succeq 0$ iff $Y = K^\dagger K$, and therefore $S^\dagger Y S = S^\dagger K^\dagger K S \succeq 0$. Now note that $Z \succ 0$ iff there exists $\varepsilon > 0$ such that $Z - \varepsilon \id \succeq 0$. Then
	\begin{equation}
	    0 \preceq S^\dagger(Z-\varepsilon\id)S = S^\dagger Z S - \varepsilon\id,
	\end{equation}
    which implies that $S^\dagger Z S \succ 0$.
\end{proof}

\section{Self-concordance}\label{sec:selfconcordance}

In this section, we adapt the proof given in Ref. \cite{he2025} to prove that the barrier functions \eqref{eq:truebarrier} and \eqref{eq:fastbarrier} are \emph{logarithmically homogeneous self-concordant barriers} (see the definition below) for their respective cones, \eqref{eq: qkd cone} and \eqref{eq: fast qkd cone}. In particular, we are interested in $\alpha \in \left[ \frac12, 1 \right)$, and therefore, $s_\alpha = \operatorname{sign}(\alpha-1) = -1$.

\begin{definition}\label{Def: LHSCB}\cite[definitions 1.3.2, 5.1.1]{Nesterov2018}, \cite[definition 2.3.2]{Nesterov1994}
 	Let $\cal V$ be a finite-dimensional real vector space and let $\cal K\subset V$ a closed convex cone. A $\nu$-\emph{logarithmically homogeneous self-concordant barrier} for $\cal K$ is a $\mathscr{C}^3$ function $f: \interior{\cal K} \to \R$ such that
    \begin{itemize}
        \item For all $x_0\in \partial \mc{K}.$
        \begin{equation}\label{eq: Def barrera}
	   	   \lim_{
		          x\rightarrow x_0 
            } f(x) = +\infty.
	    \end{equation}
        \item $f$ is $\nu$-logarithmically homogeneous:
        \begin{equation}\label{eq: Def LH}
	   	   f(\lambda x) = f(x) - \nu\log(\lambda)
	    \end{equation}
        for all $x\in\interior{\cal K}$, for all $\lambda>0$.
        \item $f$ is self-concordant:
        \begin{equation}\label{eq: Def SC}
	   	   \big| D^3 f(x)[h,h,h] \big| \leq 2\big( D^2 f(x)[h,h] \big)^{3/2}
	    \end{equation}
        for all $x\in\interior{\cal K}$, for all $h\in\cal V$.
    \end{itemize}
\end{definition}

In the proofs of self-concordance of our barriers, we use the following concepts:

\subsection{Conic monotony and conic convexity. Compatibility}

Remember that a closed convex cone $\cal K\subset V$ in a real vector space defines a partial ordering in $\cal V$ given by:
$$x\preccurlyeq_{\cal K} y \Leftrightarrow y-x\in \cal K.$$

\begin{definition}\cite[\S 2.1]{he2025}\label{KK-Mon}
	Let $\cal V, V'$ be finite-dimensional real vector spaces and let $\cal K\subset V,\, K'\subset V'$ closed convex cones. A map $f: \cal V \to V'$ is $(\cal K, K')$-\emph{monotone} if for all $x,y\in \cal V$ such that $$x\preccurlyeq_{\cal K} y$$
	then, $$f(x)\preccurlyeq_{\cal K'} f(y).$$
\end{definition}

\begin{definition}\label{KK-Conv}\cite[\S 2.1]{he2025}
	Let $\cal U, V$ be finite-dimensional real vector spaces and let $\cal K\subset V$ a closed convex cone. A map $f: \cal U \to V$ is $\cal K$-\emph{convex} if for all $x,y\in \cal U$, $\lambda\in [0,1]$, 
    \begin{equation}
    f\big( \lambda x + (1-\lambda)y \big) \preccurlyeq_{\cal K} \lambda f(x) + (1-\lambda) f(y).
    \end{equation}
	A map $f: \cal U \to V$ is $\cal K$-\emph{concave} if $-f$ is $\cal K$-convex. 
\end{definition}

To define operator monotony and operator convexity, we establish the notation
\begin{equation}
	\HH^d_{\sigma\subset I} := \{ X\in \HH^d; \; \sigma(X)\subset I \}
\end{equation}
for an interval $I\subset \R$, where $\sigma(X)$ denotes the set of eigenvalues of $X$.

\begin{definition}\label{OpMon}
	A real function $f:(a,b)\subset \R\to \R$ is \emph{operator-monotone} if for all $d\in \N$, for all $X,Y\in \HH^d_{\sigma\subset (a,b)}$ with $X\preceq Y$, $$f(X)\preceq f(Y).$$\end{definition}

\begin{lemma}
	Let $I\subset \R$ be an interval. The set $\HH^d_{\sigma\subset I}$ is convex.
\end{lemma}
\begin{proof}
	Let $X,Y\in \HH^d_{\sigma\subset I}$ and $\lambda\in [0,1]$. Let $Z = \lambda X + (1-\lambda) Y.$
	As a consequence of Weyl's theorem (\cite[Theorem 4.3.1]{HornJohnson}),
	\begin{equation}
		\left.\begin{aligned}
			\min \sigma(Z) &\geq \lambda\min \sigma(X) + (1-\lambda)\min \sigma(Y) \\
			\max \sigma(Z) &\leq \lambda\max \sigma(X) + (1-\lambda)\max \sigma(Y)
		\end{aligned}\right\},
	\end{equation}
	where 
    \begin{equation}
        \min \sigma(X), \min \sigma(Y), \max \sigma(X), \max \sigma(Y)\in I.
    \end{equation}
	We conclude that $\big[\! \min \sigma(Z) , \max \sigma(Z) \big]\subset I$, so $\sigma(Z)\subset I$.
\end{proof}

\begin{definition}\label{OpConv}A real function $f:(a,b)\subset \R$ is \emph{operator-convex} if for all $d\in \N$, for all $X,Y\in \HH^d_{\subset (a,b)}$, $\lambda\in [0,1]$, 
\begin{equation}
f\big( \lambda X + (1-\lambda)Y \big) \preceq \lambda f(X) + (1-\lambda) f(Y)
\end{equation}A function $f:(a,b)\subset \R$ is \emph{operator-concave} if $-f$ is operator-convex.\end{definition}

\begin{definition}\label{Kb-Compat}\cite[Definition 5.1.1]{Nesterov1994}, \cite[Definition 2.3]{he2025}.
	Let $\cal U, V$ be finite-dimensional real vector spaces and let $\cal K\subset V$ a closed convex cone. Let $C\subset \cal U$ be a closed convex subset and $f:\interior{C} \to \cal V$ a $\claseC^3$ function.
	Given a real number $\beta>0$, $f$ is $(\cal K, \beta)$-\emph{compatible} with $C$ if
	\begin{itemize}
		\item $f$ is $\cal K$-concave. 
		\item For all $x\in \interior{C}$, for all $h\in \cal U$ such that $x\pm h\in C,$
		\begin{equation}\label{kb-Compat}
			D^3 f(x)[h,h,h] \preccurlyeq_{\cal K} -3\beta D^2 f(x)[h,h].
		\end{equation}
	\end{itemize}
\end{definition}

\begin{definition}\label{Kb-CompatG}\cite[Definition 5.1.2]{Nesterov1994}.
	Let $f:\interior{C} \to \cal V$ verifying hypothesis of definition \ref{Kb-Compat} and let $G: \interior{C} \to \R$ a self-concordant barrier for $C$.
	Given a real number $\beta>0$, it is said that $f$ is $(\cal K, \beta)$-\emph{compatible} with the barrier $G$ if
	\begin{itemize}
		\item $f$ is $\cal K$-concave. 
		\item For all $x\in \interior{C}$, for all $h\in \cal U$ such that $x\pm h\in C,$
		\begin{equation}\label{kb-CompatG}
			D^3 f(x)[h,h,h] \preccurlyeq_{\cal K} -3\beta D^2 f(x)[h,h] \big( D^2 G(x)[h,h] \big)^\frac{1}{2}.
		\end{equation}
	\end{itemize}
\end{definition}

Concepts in definitions \ref{Kb-Compat} and \ref{Kb-CompatG} are closely related, as shown in Ref. \cite[Remark 5.1.2]{Nesterov1994}:

\begin{proposition}\label{Kb-Compat<=>G}
	Let $f:\interior{C} \to \cal V$ verifying hypothesis of definitions \ref{Kb-Compat} and \ref{Kb-CompatG}, and let 
	$\beta>0$.
	\begin{itemize}
	    \item[(i)] If $f$ is $(\cal K, \beta)$-compatible with $C$, it is $(\cal K, \beta)$-compatible with any barrier of $C$.
		\item[(ii)] Conversely, if $f$ is $(\cal K, \beta)$-compatible with a $\theta$-self-concordant barrier $G: \interior{C} \to \R$ for $C$, it is $\big( {\cal K}, (3\theta+1)\beta \big)$-compatible with $C$.
	\end{itemize}
\end{proposition}

What we want now is to prove that the function
\begin{equation}\label{eq: psi pre true-fast renyi}
	\Psi_{\alpha,S}(G,Z) := \big\|Z^\frac{1-\alpha}{2\alpha}S G^\frac12\big\|^{2\alpha}_{2\alpha}
	=
	\Tr\left[\big( G^\frac12 S^\dagger Z^\frac{1-\alpha}{\alpha} S G^\frac12 \big)^\alpha\right]
\end{equation} 
is $(\R_\geq,1)$-compatible so that we can prove that the functions 
\eqref{eq:truepsi} and \eqref{eq:fastpsi}
also are. We will do that by doing the composition with the appropriate linear map applying the following lemma:
\begin{lemma}\label{LinCompatibility}\cite[Lemma 2.5]{he2025}, \cite[Lemma 5.1.3 (iii)]{Nesterov1994}.
	Let $\cal U, V$ be finite-dimensional real vector spaces and let $C\subset \cal V$ be a closed convex subset.
	Let $f:\interior{C} \to \HH^n$ be a $\claseCk{3}$ $\HH^n_\succeq$-concave function. Suppose $f$ is $(\HH^n_\succeq, \beta)$-compatible with $C$, $\beta>0$.
	\begin{itemize}
		\item[(i)] Let ${\cal L}: \HH^n \to \HH^m$ a positive linear map. 
		Then, ${\cal L}\circ f: \interior{C} \to \HH^n \to \HH^m$ is $(\HH^m_\succeq, \beta)$-compatible with $C$.
		\item[(ii)] Let $\varphi: \cal U \to V$ be an affine map such that $\interior{C}\cap \im\varphi\neq \emptyset$. Then, $f\circ \varphi: \inv{\varphi}(\interior{C}) \to \interior{C} \to \HH^n$ is $(\HH^n_\succeq, \beta)$-compatible with $\inv{\varphi}(C)$.
	\end{itemize}
\end{lemma}

\subsection{\texorpdfstring{$(\R_\geq,1)$}{(R\_>,1)}-compatibility}

We are proving the $(\R_\geq,1)$-compatibility of \eqref{eq: psi pre true-fast renyi} this way:

\begin{theorem}\cite[Theorem 3.1]{he2025}\label{Kerry3.1}
	Let $\cal U, V$ be finite-dimensional real vector spaces. Let $C\subset \cal U$ be a closed convex set and $\cal K\subset V$ a closed convex cone. Let $\Phi:\interior{C} \to\cal V$ be a $\claseCk{3}$ function. Let $\cal K^*$ be the dual cone, and suppose that for all $z\in\cal K^*$ and for all $x\in \interior{C}$, $h\in\cal U$ such that $x\pm h\in \interior{C}$, the function
	\begin{equation}\label{kerry3.1}
		\begin{array}{ccccl}
			\phi & : & (-1,1) & \rightarrow & \R \\
			&& t & \mapsto & \inprod{z}{\Phi(x+th)}
		\end{array}
	\end{equation}
	is operator-concave in $(-1,1)$. Then $\Phi:\interior{C} \to\cal V$ is $({\cal K},1)$-compatible with $C$.
\end{theorem}

Let 
\begin{equation}\label{kerry3.1psi}
	\begin{array}{ccccl}
		\psi & : & (-1,1) & \rightarrow & \R \\
		&& t & \mapsto & \Psi_{\alpha,S}(G+tH,Z+tK)
	\end{array}
\end{equation}
where $(G,Z)\in \HH^k_\succ\times \HH^n_\succ$, $(H,K)\in \HH^k\times \HH^n$ are such that $(G\pm H,Z\pm K)\in \HH^k_\succ\times \HH^n_\succ$. The idea is to prove that $\psi$ is operator-concave in $(-1,1)$ so that we can apply Theorem \ref{Kerry3.1}.\\

For the following lemma, we define the the upper half-plane
$
\C^+ := \{ z\in \C; \; \operatorname{Im}(z) > 0 \}.
$

\begin{lemma}\label{Psihat}
	Let 
	\begin{equation}
		\begin{array}{ccccl}
			\hat\psi & : & (1,+\infty) & \rightarrow & \R \\
			&& t & \mapsto & t \psi\left(\frac1t\right)
		\end{array}
	\end{equation}
	$\hat\psi$ has an analytic continuation to $\C^+$ such that $\hat\psi(\C^+)\subset \C^+$.
\end{lemma}
\begin{proof}
	Note that this is a particular case of Ref. \cite[Theorem 2.1 (1)]{HiaiTr2} with $p=\frac{1-\alpha}{\alpha}$, $q=1$, $s=\alpha$, $\Psi(X) = S^\dagger X S$, and $\Phi$ the identity map.
	First of all, let us write $\hat\psi$:
	\begin{align}
		\hat\psi(t) &
		= \Tr\Big[ \big( ( tG + H )^\frac12 S^\dagger ( tZ + K )^\frac{1-\alpha}{\alpha} S ( tG + H )^\frac12 \big)^\alpha \Big].
	\end{align}
	
	Let us remember the \emph{cartesian decomposition} (\cite[Sec 4.4]{FZhangMatTheory}) of a matrix $C\in \C^{d\times d}$:
	the Hermitian matrices
	\begin{equation}
		\begin{aligned}
			\operatorname{Re}(C) &= \frac{C+C^\dagger}{2}, \\
			\operatorname{Im}(C) &= \frac{C-C^\dagger}{2i},
		\end{aligned}
	\end{equation}
	are called \emph{real part} and \emph{imaginary part} of $C$, respectively.
	In this context, $C = \operatorname{Re}(C) + i\operatorname{Im}(C)$ and they are the unique decomposition of $C = A + iB$ with $A,B$ Hermitian matrices.
    
	We define the set (\cite[Se 1]{HiaiTr1})
	\begin{equation}
		{\cal I}_d^{+}
		:= \big\{ A+iB; \; A,B\in \HH^d, B\succ 0 \big\}\subset \C^{d\times d}.
	\end{equation}
	If $c\in \C^+$, $cG+H\in {\cal I}_k^{+}$ and $cZ+K\in {\cal I}_n^{+}$. By Ref. \cite[Lemma 1.2]{HiaiTr1}, then 
	$( cG + H )^\frac12\in {\cal I}_k^{+}$ and $( cZ + K )^\frac{1-\alpha}{\alpha}\in {\cal I}_n^{+}$, because $0 < \frac12, \frac{1-\alpha}{\alpha} \leq 1$. As the map $X \mapsto S^\dagger X S$ is strictly positive (see proposition \ref{mapS_++}), $S^\dagger ( tZ + K )^\frac{1-\alpha}{\alpha} S\in {\cal I}_n^{+}$.
    \\

	Let us call
	\begin{equation}
		F(t) = ( tG + H )^\frac12 S^\dagger ( tZ + K )^\frac{1-\alpha}{\alpha} S ( tG + H )^\frac12
	\end{equation}
	and $\gamma = \frac{1-\alpha}{\alpha} + 1 = \frac{1}{\alpha}\in \left(1,2\right]\subset \left(0,2\right]$. By Ref. \cite[proof of Theorem 2.1, eq (2.5)]{HiaiTr2}, if $\sigma\big(F(c)\big)$ is the set of the eigenvalues of $F(c)$,
	$
	\sigma\big(F(c)\big)\subset \Gamma_{\gamma\pi}
	$
	for $c\in \C^+$, where
	\begin{equation}
		\Gamma_{\Theta} = \{ re^{i\theta}; \; r>0, 0<\theta<\Theta \},\; \Theta\in (0,2\pi).
	\end{equation}
	Therefore,  
	$
	\sigma\big(F(c)^\alpha\big)\subset \Gamma_{\pi} = {\cal I}_d^{+}
	$
	(remember that $\gamma = \frac1\alpha$). Then,
	$
	\Tr[(F(c)^\alpha]\in {\cal I}_d^{+}.
	$
\end{proof}

Once we have lemma \ref{Psihat}, by Lowner's theorem (\cite[Lemma 2.1]{he2025}), we conclude that $\hat\psi
$ is operator-monotone in $(1,+\infty)$. Now we can get that $\hat\psi
$ is operator concave in $(0,1)$ by Ref. \cite[Lemma 4.3]{he2025}, since
\begin{align}
	\hat{\hat{f}}(t) &= t\hat f\left(\frac1t\right) = t\frac1t f\left(\frac{1}{\frac1t}\right) = f(t).
\end{align}

\begin{lemma}\label{SwapOpConcave}
	Let $\phi$ the function \eqref{kerry3.1} in the conditions of Theorem \ref{Kerry3.1}. If $\phi$ is operator-concave in $(0,1)$, it is also operator-concave in $(-1,0)$.
\end{lemma}
\begin{proof}
	Let us recall $\phi_+$ the function \eqref{kerry3.1} $\phi$.
	Because of the conditions of Theorem \ref{Kerry3.1}, the function \eqref{kerry3.1} is well defined in $(-1,1)$ swapping $h$ for $-h$. We call $\phi_-$ this new function.
	\begin{equation*}\label{}		
		\begin{array}{ccccccl}			
			\phi_- & : & (-1,0) & \rightarrow & (0,1) & \overset{\phi_+}{\longrightarrow} & \R \\			
			&& t & \mapsto & -t & \mapsto & \phi_+(-t) 
		\end{array}	
	\end{equation*}
	Let $X,Y\in \HH^d$ whose eigenvalues are in $(-1,0)$. 
	By hypothesis, $\phi_+$ is operator-concave in $(0,1)$. 
	Then, given $\lambda\in [0,1]$
	\begin{align*}
		\phi_-\big( \lambda X + (1-\lambda) Y \big) &
		= \phi_+\big( -(\lambda X + (1-\lambda) Y) \big) 
		= \phi_+\big( \lambda (-X) + (1-\lambda) (-Y) \big) \\ &
		\succeq \lambda \phi_+(-X) + (1-\lambda) \phi_+(-Y) = \lambda \phi_-(X) + (1-\lambda) \phi_-(Y).
	\end{align*}
\end{proof}

By lemma \ref{SwapOpConcave}, the function \eqref{kerry3.1psi} $\psi$ is also operator-concave in $(-1,0)$. Using Ref. \cite[Lemma 4.4]{he2025} for all $d\in \N$, we get that $\psi$ is operator-concave in $(-1,1)$. Finally, applying Theorem \ref{Kerry3.1}, we conclude that the function
\eqref{eq: psi pre true-fast renyi}
is $(\R_\geq,1)$-compatible with $\HH^k_\succeq\times \HH^n_\succeq$ for $\alpha\in\left[\frac12, 1\right)$.

\subsection{Proof of self-concordance for the RényiQKD cone}

Let us rewrite the function \eqref{eq:truepsi} in terms of \eqref{eq: psi pre true-fast renyi}:
\begin{equation}
	\Psihat(\rho, \sigma) = \Psi_{\alpha,S}\circ \lmap_{\gmap,\zmap}(\rho, \sigma),
\end{equation}
where $\lmap_{\gmap,\zmap}(\rho,\sigma) = \big( \gmap(\rho),\zmap(\sigma) \big)$. Applying Lemma \ref{LinCompatibility} (ii), $\Psihat(\rho, \sigma)$ is $(\R_\geq,1)$-compatible with $\HH^q_\succeq\times \HH^m_\succeq$.
Once we know that, 
we can prove that the function
\begin{equation}
	f(u, \rho, \sigma) = -\log(\Psihat(\rho, \sigma) - u) - \logdet(\rho) - \logdet(\sigma) \label{eq:truebarrier hypo}
\end{equation}
is a self-concordant barrier for $\cl\hypo\big( \Psihat(\rho, \sigma) \big)$ by Ref. \cite[Lemma 2.4]{he2025}:

\begin{theorem}\label{SCKerry}
	Let $\cal U, V$ be finite-dimensional real vector spaces. 
	Let $C\subset \cal U$ be a closed convex set and $\cal K\subset V$ a closed convex cone. 
	Let $G : \interior{C}\subset\cal U \to \R$ be a $\claseCk{2}$ $\nu$-SCB for $C$ and $F : \interior{C}\subset \cal U \to V$ $(\cal K, \beta)$-compatible with $C$, $\beta\geq 0$. 
	Let $H : \interior{\cal K} \to \R$ be a $\eta$-SCB for $\cal K$. 
	Then, the function
	\begin{displaymath}
		\begin{array}{ccccl}
			f &:& \interior{C} \times {\cal V} & \rightarrow & \R\cup \{+\infty\} \\
			&& (x,y) & \mapsto & H\big( F(x) - y \big) + \beta^3 G(x)
		\end{array}
	\end{displaymath}
	is a $(\eta + \beta^3 \nu)$-SCB for $\cl\hypo F$, where
	$$\hypo F = \left\{ (y,x)\in {\cal W}\times \interior{C}\; ;\, y\preccurlyeq_{\cal K} F(x) \right\}.$$
\end{theorem}
This is a particular case of \cite[Theorem 5.4.4]{Nesterov2018}, noticing proposition \ref{Kb-Compat<=>G} (i).

\begin{corollary}
	The function \eqref{eq:truebarrier hypo}
	is a $(1+q+m)$-SCB for $\cl\hypo\big( \Psihat(\rho, \sigma) \big).$
\end{corollary}
\begin{proof}
	We apply Theorem \ref{SCKerry} with
	\begin{itemize}
		\item ${\cal U} = \HH^q\times \HH^m$, $C = \HH^q_\succeq\times \HH^m_\succeq$, $G(\rho, \sigma) = -\logdet(\rho) -\logdet(\sigma)$, which is a $(q+m)$-SCB for $\HH^q_\succeq\times \HH^m_\succeq$ $(\nu = q+m)$.
		\item ${\cal V} = \R$, ${\cal K} = \R_\geq$.
		\item $F = \Psihat : \HH^q_\succ\times \HH^m_\succ \to \R$, which is $(\R_\geq, 1)$-compatible with $\HH^q_\succeq\times \HH^m_\succeq$ $(\beta = 1)$.
		\item $H = -\log$, which is a $1$-SCB for $\R_\geq$ $(\eta = 1)$.
	\end{itemize}
\end{proof}

\subsection{Proof of self-concordance for the FastRényiQKD cone}

Following an analogous argument, we can rewrite the function \eqref{eq:fastpsi} in terms of \eqref{eq: psi pre true-fast renyi}:
\begin{equation}
	\Psihat(\rho) = \Psi_{\alpha,S}\circ \lmap_{\gmap,\zgmap}(\rho),
\end{equation}
where $\lmap_{\gmap,\zgmap}(\rho) = \big( \gmap(\rho),\zgmap(\rho) \big)$. Applying again Lemma \ref{LinCompatibility} (ii), $\Psihat(\rho)$ is $(\R_\geq,1)$-compatible with $\HH^q_\succeq$.
We can prove then that the function
\begin{equation}
	f(u, \rho) = -\log(\Psihat(\rho) - u) - \logdet(\rho) - \logdet(\sigma) \label{eq:fastbarrier hypo}
\end{equation}
is a self-concordant barrier for $\cl\hypo\big( \Psihat(\rho) \big)$:

\begin{corollary}
	The function \eqref{eq:fastbarrier hypo}
	is a $(1+q)$-SCB for $\cl\hypo\big( \Psihat(\rho) \big).$
\end{corollary}
\begin{proof}
	We apply Theorem \ref{SCKerry} with
	\begin{itemize}
		\item ${\cal U} = \HH^q$, $C = \HH^q_\succeq$, $G(\rho) = -\logdet(\rho)$, which is a $q$-SCB for $\HH^q_\succeq$ $(\nu = q)$.
		\item ${\cal V} = \R$, ${\cal K} = \R_\geq$.
		\item $F = \Psihat : \HH^q_\succ \to \R$, which is $(\R_\geq, 1)$-compatible with $\HH^q_\succeq$ $(\beta = 1)$.
		\item $H = -\log$, which is a $1$-SCB for $\R_\geq$ $(\eta = 1)$.
	\end{itemize}
\end{proof}

\section{Derivatives}\label{sec:derivatives}

\subsection{Background}

Real functions can be extended to act on Hermitian matrices by diagonalizing the matrix and applying the function to each (real) eigenvalue. More precisely, let $f:\mathbb{R} \mapsto \mathbb{R}$ and $X \in \mathbb{H}^n$ with spectral decomposition $X = U\operatorname{diag}(\vec{\lambda})U^\dagger$, where $\vec{\lambda} \in \mathbb{R}^n$ is the vector of eigenvalues. Then we define $f(X)$ as
\begin{equation}
    f(X) = U\operatorname{diag}(f(\vec{\lambda}))U^\dagger,
\end{equation}
where $f$ acts on $\vec{\lambda}$ elementwise. We will call $f : \mathbb{H}^n \to \mathbb{H}^n$ a \textit{spectral function}.

The barriers of our cones, eqs. \eqref{eq:truebarrier} and \eqref{eq:fastbarrier}, are defined in terms of spectral functions. In order to differentiate spectral functions, we need to extend the usual rules of multivariable calculus to these functions. In this section, we collect some of these rules for convenience, but they can also be found in matrix analysis books like Refs. \cite{Higham2008, Hiai2014}. We will refer to these properties in the computation of the gradient, Hessian and third derivatives of the barrier.

First, let's recall some rules valid for general functions. The Fréchet derivative of the composition of two functions applied to a vector direction $H$ obeys the chain rule:
\begin{equation}\label{eq:chain rule}
    D(f\circ g)(X)[H] = Df(g(X))[Dg(X)[H]] \:,
\end{equation}
and the derivative of a product obeys the product rule:
\begin{equation}\label{eq:product rule}
    D(fg)(X)[H] = Df(X)[H]g(X) + f(X)Dg(X)[H] \:.
\end{equation}
In particular, we will use it when the product is composed with the trace. Since the trace is a linear operator, its differential is itself and we obtain
\begin{equation}\label{eq:trace product rule}
    D(\Tr \circ fg)(X)[H] = \inprod{Df(X)[H]}{g(X)} +  \inprod{f(X)}{Dg(X)[H]} \:.
\end{equation}
Turning our attention to spectral functions, the derivative is given by
\begin{equation}\label{eq:frechet derivative}
    Df(X)[H] = U(f^{[1]}(\Lambda) \odot (U^\dagger H U))U^\dagger \:,
\end{equation}
where $X = U\Lambda U^\dagger$ is the diagonalization of $X$, $f^{[1]}$ are the first divided differences of $f$ and $\odot$ is the Hadamard or element-wise matrix product.

Another important property is that the derivative of a spectral function composed with the trace is equal to:
\begin{equation}\label{eq:trace derivative}
    D(\Tr \circ f)(X)[H] = \inprod{H}{f'(X)} \:,
\end{equation}
where $\langle A, B\rangle = \Tr(AB^\dagger)$ is the Hilbert-Schmidt inner product.
We will call this equation the \textit{trace derivative}. We can combine the chain rule and the trace derivative to obtain:
\begin{equation}\label{eq:trace chain rule}
    D(\Tr\circ f\circ g)(X)[H] = D(\Tr\circ f)(g(X))[Dg(X)[H]] = \inprod{Dg(X)[H]}{f'(g(X)} \:.
\end{equation}
These equations will be used extensively in the upcoming derivations.

To compute the Hessian and third derivatives of the barriers, we will need higher-order Fréchet derivatives of spectral functions. The second-order derivative applied to two vector directions is given by
\begin{equation}
    D^2 f(X)[H,K] = U \left[ \sum_{k=1}^n f_k^{[2]}(\Lambda)\odot \Big( U^\dagger H U \ket{k}\bra{k} U^\dagger K U + U^\dagger K U \ket{k}\bra{k} U^\dagger H U \Big) \right] U^\dagger \:,
\end{equation}
where $X = U\Lambda U^\dagger$ and $f_k^{[2]}(\Lambda) = \left[f^{[2]}(\lambda_i,\lambda_j,\lambda_k) \right]_{i,j = 1}^n$ (second divided differences of $f$).

The third derivative applied twice to the same vector direction is
\begin{align}
    D^3 f(X)[H, H, K] = U\Bigg[2\,\sum_{k,\ell =1}^n f_{k,\ell}^{[3]}(\Lambda)\odot \big[ U^\dagger H U \ket{k}\bra{k} U^\dagger H U \ket{\ell}\bra{\ell} U^\dagger K U + \nonumber \\[1mm]
    + U^\dagger H U \ket{k}\bra{k} U^\dagger K U \ket{\ell}\bra{\ell} U^\dagger H U + U^\dagger K U \ket{k}\bra{k} U^\dagger H U \ket{\ell}\bra{\ell} U^\dagger H U \big]\Bigg] U^{\dagger} ,
\end{align}
where $X = U\Lambda U^\dagger$ and $f_{k \ell}^{[3]}(\Lambda) = \left[f^{[3]}(\lambda_i,\lambda_k, \lambda_\ell, \lambda_j) \right]_{i,j = 1}^n$\, (third divided differences of $f$).

Finally, to isolate some of the vector directions later on, we will use the property that Fréchet derivatives of spectral functions are self-adjoint:
\begin{align}\label{eq:self adjointness derivative}
    \inprod{A}{Df(X)[H]} & = \inprod{A}{U(f^{[1]}(\Lambda) \odot (U^\dagger H U))U^\dagger} \nonumber \\[1mm]
     & = \inprod{U^\dagger A U}{f^{[1]}(\Lambda) \odot (U^\dagger H U)} \nonumber \\[1mm]
     & = \inprod{f^{[1]}(\Lambda) \odot (U^\dagger A U)}{U^\dagger H U} \nonumber \\[1mm]
     & = \inprod{U(f^{[1]}(\Lambda) \odot (U^\dagger A U))U^\dagger}{H} = \inprod{Df(X)[A]}{H} \:.
\end{align}

By applying this property successive times, we can prove that higher-order derivatives are also self-adjoint:
\begin{align}
    \inprod{A}{D^2f(X)[H, K]} & = \inprod{D^2f(X)[H, A]}{K} \:, \label{eq:self adjointness second derivative} \\[1mm]
    \inprod{A}{D^3f(X)[H, K, L]} & = \inprod{D^3f(X)[H, K, A]}{L}  \label{eq:self adjointness third derivative} \:.
\end{align}

\subsection{RényiQKD cone}

\subsubsection{Gradient}\label{sec:gradient qkd cone}

The gradient of the barrier \eqref{eq:truebarrier} is 
\begin{align}
    & \nabla_h f(h, \rho, \sigma) = -\frac{1}{z} \:, \\
    & \nabla_\rho f(h, \rho, \sigma) = \frac{s_\alpha}{z} \, \nabla_\rho \Psihat (\rho, \sigma) - \rho^{-1} \:, \\
    & \nabla_\sigma f(h, \rho, \sigma) = \frac{s_\alpha}{z} \, \nabla_\sigma \Psihat (\rho, \sigma) - \sigma^{-1} \:,    
\end{align}
where $z = h-s_\alpha \Psihat (\rho, \sigma)$ and $\Psihat$ is given by Eq. \eqref{eq:truepsi}.

For convenience let us define the intermediate variables $G := \gmap(\rho)$ and $Z := \zmap(\sigma)$, and the functions $g(x) = x^\alpha$ and $h(x) = x^\frac{1-\alpha}{\alpha}$, so that we can write
\begin{equation}\label{eq:goodequivalenceGZ}
    \Psihathat(G,Z) := \Psihat(\rho,\sigma) = \Tr\left[g\big(G^\frac12 S^\dagger h(Z) S G^\frac12\big)\right] = \Tr\left[g\big(Z_S^\frac12 G Z_S^\frac12\big) \right] \:,
\end{equation}
where we defined $Z_S := S^\dagger h(Z)S$ for brevity. The two expressions we give for $\Psihathat(G,Z)$ come from eq. \eqref{eq:goodequivalence}, and as mentioned before, the operators appearing there are positive definite, which is necessary for the derivatives to be well-defined.

The chain rule \eqref{eq:chain rule} gives us the derivatives of $\Psihat$ in terms of the derivatives of $\Psihathat$:
\begin{gather}
    D_\rho\Psihat(\rho, \sigma)[H] = D_G\Psihathat(G,Z)[D_\rho G(\rho)[H]] \:, \\
    D_\sigma\Psihat(\rho, \sigma)[H] = D_Z\Psihathat(G,Z)[D_\sigma Z(\sigma)[H]] \:.
\end{gather}
Applying the chain rule and the trace derivative \eqref{eq:trace chain rule} to the expression $\Tr\left[g\big(Z_S^\frac12 G Z_S^\frac12\big) \right]$ in eq. \eqref{eq:goodequivalenceGZ} gives us
\begin{align}
    D_G\Psihathat(G,Z)[H_G] &= \inprod{Z_S^\frac12 H_G Z_S^\frac12}{g'\big(Z_S^\frac12 G Z_S^\frac12\big)} = \inprod{H_G}{Z_S^\frac12 g'\big(Z_S^\frac12 G Z_S^\frac12 \big) Z_S^\frac12} \:.\label{eq:gooddg}
\end{align}
Now, using the expression $\Tr\left[g\big(G^\frac12 S^\dagger h(Z) S G^\frac12\big)\right]$ from Eq. \eqref{eq:goodequivalenceGZ} to compute the derivative with respect to $Z$, we obtain
\begin{align}
    D_Z\Psihathat(G,Z)[H_Z] &= \inprod{G^\frac12 S^\dagger Dh(Z)[H_Z] S G^\frac12}{g'\big(G^\frac12 Z_S G^\frac12\big)} \\
    &= \inprod{  Dh(Z)[H_Z] }{S G^\frac12 g'\big(G^\frac12 Z_S G^\frac12\big)G^\frac12 S^\dagger} \label{eq:gooddz}\\
    &= \inprod{ H_Z }{ Dh(Z)[S G^\frac12 g'\big(G^\frac12 Z_S G^\frac12\big)G^\frac12 S^\dagger]} \:,
\end{align}
where in the last step we used the self-adjointness of the spectral derivative \eqref{eq:self adjointness derivative}. Now, substituting $G=\gmap(\rho)$, $Z=\zmap(\sigma)$, $H_G = D_\rho G(\rho)[H_\rho] = \gmap(H_\rho)$ and $H_Z = D_\sigma Z(\sigma)[H_\sigma] = \zmap(H_\sigma)$, we obtain
\begin{gather}
    D_\rho\Psihat(\rho, \sigma)[H_\rho] = \inprod{H_\rho}{\gmap^\dagger(Z_S^\frac12 g'\big(Z_S^\frac12 G Z_S^\frac12 \big) Z_S^\frac12)} \:, \\
    D_\sigma\Psihat(\rho, \sigma)[H_\sigma] = \inprod{H_\sigma}{\zmap^\dagger(Dh(Z)[S G^\frac12 g'\big(G^\frac12 Z_S G^\frac12\big)G^\frac12 S^\dagger])} \:,
\end{gather}
and therefore the gradient is
\begin{gather}
    \nabla_\rho\Psihat(\rho, \sigma) = \gmap^\dagger(Z_S^\frac12 g'\big(Z_S^\frac12 G Z_S^\frac12 \big)Z_S^\frac12) \:, \label{eq:truegradient rho} \\
    \nabla_\sigma\Psihat(\rho, \sigma) = \zmap^\dagger(Dh(Z)[S G^\frac12 g'\big(G^\frac12 Z_S G^\frac12\big)G^\frac12 S^\dagger]) \:. \label{eq:truegradient sigma}
\end{gather}
The gradient can be significantly simplified if we write it in terms of the singular value decomposition of $Z_S^\frac12 G^\frac12 = U \Lambda V^\dagger$:
\begin{gather}
    \nabla_\rho\Psihat(\rho, \sigma) = \gmap^\dagger(Z_S^\frac12 U g'\big(\Lambda^2\big) U^\dagger  Z_S^\frac12) \:, \\
    \nabla_\sigma\Psihat(\rho, \sigma) = \zmap^\dagger(Dh(Z)[S G^\frac12 V g'\big(\Lambda^2\big) V^\dagger G^\frac12 S^\dagger]) \:.
\end{gather}

\subsubsection{Hessian}

The Hessian of the barrier \eqref{eq:truebarrier} is 
\begin{align}
	& \nabla^2_{hh} f(h, \rho, \sigma) = \frac{1}{z^2} \:, \\[1.5mm]
	& \nabla^2_{h\rho} f(h, \rho, \sigma) = - \frac{s_\alpha}{z^2} \, \nabla_\rho \Psihat (\rho, \sigma) \:, \\[1.25mm]
    & \nabla^2_{h\sigma} f(h, \rho, \sigma) = - \frac{s_\alpha}{z^2} \, \nabla_\sigma \Psihat (\rho, \sigma) \:, \\[1.25mm]
	& \nabla^2_{\rho\rho} f(h, \rho, \sigma) = H_\rho \mapsto \frac{1}{z^2} \inprod{ \nabla_\rho \Psihat (\rho, \sigma)}{H_\rho} \, \nabla_\rho \Psihat (\rho, \sigma) + \frac{s_\alpha}{z} \nabla^2_{\rho\rho} \, \Psihat (\rho, \sigma) H_\rho + \rho^{-1} H_\rho \rho^{-1}, \\[1.25mm]
    & \nabla^2_{\sigma\sigma} f(h, \rho, \sigma) = H_\sigma \mapsto \frac{1}{z^2} \inprod{ \nabla_\sigma \Psihat (\rho, \sigma)}{H_\sigma} \, \nabla_\sigma \Psihat (\rho, \sigma) + \frac{s_\alpha}{z} \nabla^2_{\sigma\sigma} \, \Psihat (\rho, \sigma)  H_\sigma + \sigma^{-1} H_\sigma \sigma^{-1}, \\[1.25mm]
    & \nabla^2_{\rho\sigma} f(h, \rho, \sigma) = H_\sigma \mapsto \frac{1}{z^2} \inprod{ \nabla_\sigma \Psihat (\rho, \sigma)}{H_\sigma} \, \nabla_\rho \Psihat (\rho, \sigma) + \frac{s_\alpha}{z} \nabla^2_{\rho\sigma} \, \Psihat (\rho, \sigma) H_\sigma \:, \\
    & \nabla^2_{\sigma\rho} f(h, \rho, \sigma) = H_\rho \mapsto \frac{1}{z^2} \inprod{ \nabla_\rho \Psihat (\rho, \sigma)}{H_\rho} \, \nabla_\sigma \Psihat (\rho, \sigma) + \frac{s_\alpha}{z} \nabla^2_{\sigma\rho} \, \Psihat (\rho, \sigma) H_\rho \:,      
\end{align}
where $z = h-s_\alpha\Psihat (\rho, \sigma)$.
\vspace{1mm}

Similarly to the derivation of the gradient, we obtain $D_{GG}\Psihathat(G,Z)$ by differentiating \eqref{eq:gooddg} over $G$:
\begin{equation}\label{eq:DGG}
    D_{GG}^2\Psihathat(G,Z)[H_{G_1}, H_{G_2}] = \inprod{H_{G_1}}{Z_S^\frac12 Dg'(Z_S^\frac12 G Z_S^\frac12)[Z_S^\frac12 H_{G_2} Z_S^\frac12] Z_S^\frac12}
\end{equation}
We obtain $D_{ZZ}^2\Psihathat(G,Z)$ by differentiating \eqref{eq:gooddz} over $Z$ with the product rule \eqref{eq:trace product rule}:
\begin{multline}\label{eq:DZZ}
    D_{ZZ}^2\Psihathat(G,Z)[H_{Z_1}, H_{Z_2}] = \inprod{D^2h(Z)[H_{Z_1}, H_{Z_2}] }{S G^\frac12 g'\big(G^\frac12 Z_S G^\frac12\big)G^\frac12 S^\dagger} + \\
    +  \inprod{Dh(Z)[H_{Z_1}] }{S G^\frac12 D g'\big(G^\frac12 Z_S G^\frac12\big)\Big[G^\frac12 S^\dagger Dh(Z)[H_{Z_2}]S G^\frac12\Big]G^\frac12 S^\dagger} \:.
\end{multline}
Using He's Lemma 5.1 \cite{he2025}, we have that
\begin{align}
    D_Z\Psihathat(G,Z)[H_Z] &= \inprod{  Dh(Z)[H_Z] }{S G^\frac12 g'\big(G^\frac12 Z_S G^\frac12\big)G^\frac12 S^\dagger} \\
                            &= \inprod{  Dh(Z)[H_Z] }{S Z_S^{-\frac12} \tilde g\big(Z_S^\frac12 G Z_S^\frac12\big)Z_S^{-\frac12} S^\dagger} \:,
\end{align}
where $\tilde g(x) = xg'(x)$. With this second expression we can obtain $D_{ZG}^2$ easily by differentiating it over $G$:
\begin{equation}\label{eq:DZG}
    D_{ZG}^2\Psihathat(G,Z)[H_{Z_1}, H_{G_2}] = \inprod{  Dh(Z)[H_{Z_1}] }{S Z_S^{-\frac12} D\tilde g\big(Z_S^\frac12 G Z_S^\frac12\big)[Z_S^\frac12 H_{G_2} Z_S^\frac12]Z_S^{-\frac12} S^\dagger} \:.
\end{equation}

Now we can substitute $G=\gmap(\rho)$, $Z=\zmap(\sigma)$ and the vector directions $H_{G_i} = \gmap(H_{\rho_i})$ and $H_{Z_i} = \zmap(H_{\sigma_i})$. To obtain the product of the Hessian times a vector we need to isolate the first vector directions, $H_{\rho_1}$ and $H_{\sigma_1}$:
\begin{equation}\label{eq:hessian rho rho}
     \nabla^2_{\rho\rho}\Psihat (\rho, \sigma) \, H_{\rho_2} = \gmap^\dagger(Z_S^\frac12 Dg'(Z_S^\frac12 G Z_S^\frac12)[Z_S^\frac12 \gmap(H_{\rho_2}) Z_S^\frac12] Z_S^\frac12) \:,
\end{equation}
\begin{equation}\label{eq:hessian sigma rho}
    \nabla^2_{\sigma \rho}\Psihat (\rho, \sigma) \, H_{\rho_2} = \zmap^\dagger(Dh(Z)[S Z_S^{-\frac12} D\tilde g\big(Z_S^\frac12 G Z_S^\frac12\big)[Z_S^\frac12 \gmap(H_{\rho_2}) Z_S^\frac12]Z_S^{-\frac12}S^\dagger]) \:,
\end{equation}
\begin{equation}\label{eq:hessian rho sigma}
    \nabla^2_{\rho \sigma}\Psihat (\rho, \sigma) \, H_{\sigma_2} = \gmap^\dagger(Z_S^\frac12  D\tilde g\big(Z_S^\frac12 G Z_S^\frac12\big)[Z_S^{-\frac12} S^\dagger Dh(Z)[\zmap(H_{\sigma_2})] S Z_S^{-\frac12}] Z_S^\frac12) \:,
\end{equation}
\begin{multline}\label{eq:hessian sigma sigma}
    \nabla^2_{\sigma \sigma}\Psihat (\rho, \sigma) \, H_{\sigma_2} = \zmap^\dagger\left(D^2h(Z)[\zmap(H_{\sigma_2}), S G^\frac12 g'\big(G^\frac12 Z_S G^\frac12\big)G^\frac12 S^\dagger]\right)  \\[1mm]
    + \zmap^\dagger(Dh(Z)[S G^\frac12 D g'\big(G^\frac12 Z_S G^\frac12\big)\Big[G^\frac12 S^\dagger Dh(Z)[\zmap(H_{\sigma_2})]S G^\frac12\Big]G^\frac12 S^\dagger]) \:.
\end{multline}
And we obtain the Hessian, $\nabla^2 \Psihat$, by vectorizing (see Appendix A of Ref. \cite{Lorente2025quantumkey}) the previous equations and thus isolating the remaining vector directions, $H_{\rho_2}$ and $H_{\sigma_2}$.

\subsubsection{Third derivatives}\label{sec:third derivatives qkd cone}

The third derivatives of the barrier \eqref{eq:truebarrier} are 
\begin{align}
	\nabla^3_{hhh} f(h, \rho, \sigma) &= -\frac{2}{z^3} \:, \\[1.5mm]
	\nabla^3_{hh\rho} f(h, \rho, \sigma) &= \frac{2 s_\alpha}{z^3} \, \nabla_\rho \widehat\Psi_\alpha (\rho, \sigma) \:, \\[1.25mm]
    \nabla^3_{hh\sigma} f(h, \rho, \sigma) &= \frac{2 s_\alpha}{z^3} \, \nabla_\sigma \widehat\Psi_\alpha (\rho, \sigma) \:, \\[1.25mm]
	\nabla^3_{h\rho\rho} f(h, \rho, \sigma) &= H_\rho \mapsto -\frac{2 s_\alpha^2}{z^3} \inprod{ \nabla_\rho \widehat\Psi_\alpha (\rho, \sigma)}{H_\rho} \, \nabla_\rho \widehat\Psi_\alpha (\rho, \sigma) - \frac{s_\alpha}{z^2} \nabla^2_{\rho\rho} \, \widehat\Psi_\alpha(\rho, \sigma) H_\rho \:, \\[1.25mm]
    \nabla^3_{h\sigma\sigma} f(h, \rho, \sigma) &= H_\sigma \mapsto -\frac{2 s_\alpha^2}{z^3} \inprod{ \nabla_\sigma \widehat\Psi_\alpha (\rho, \sigma)}{H_\sigma} \, \nabla_\sigma \widehat\Psi_\alpha (\rho, \sigma) - \frac{s_\alpha}{z^2} \nabla^2_{\sigma\sigma} \, \widehat\Psi_\alpha(\rho, \sigma) H_\sigma \:, \\[1.25mm]
    \nabla^3_{h\sigma\rho} f(h, \rho, \sigma) &= H_\rho \mapsto -\frac{2 s_\alpha^2}{z^3} \inprod{ \nabla_\rho \widehat\Psi_\alpha (\rho, \sigma)}{H_\rho} \, \nabla_\sigma \widehat\Psi_\alpha (\rho, \sigma) - \frac{s_\alpha}{z^2} \nabla^2_{\sigma\rho} \, \widehat\Psi_\alpha(\rho, \sigma) H_\rho \:, \\[1.25mm]
    \nabla^3_{\rho\rho\rho} f(h, \rho, \sigma) &= (H_\rho, H_\rho) \mapsto \frac{2 s_\alpha^3}{z^3} \inprod{ \nabla_\rho \widehat\Psi_\alpha (\rho, \sigma)}{H_\rho}^2  \, \nabla_\rho \widehat\Psi_\alpha (\rho, \sigma) + \nonumber \\[1.25mm]
    &+ \frac{s_\alpha^2}{z^2} \inprod{\nabla^2_{\rho\rho}\widehat\Psi_\alpha H_\rho}{H_\rho} \, \nabla_\rho\widehat\Psi_\alpha(\rho, \sigma) + \frac{2 s_\alpha^2}{z^2} \inprod{\nabla_{\rho}\widehat\Psi_\alpha}{H_\rho} \, \nabla^2_{\rho\rho} \widehat\Psi_\alpha(\rho, \sigma) H_\rho \nonumber \\[1.25mm]
    & +\frac{s_\alpha}{z}  \nabla^3_{\rho\rho\rho} \widehat\Psi_\alpha(\rho, \sigma)[\,\cdot\,, H_\rho, H_\rho] - 2\rho^{-1} H_\rho \rho^{-1} H_\rho \rho^{-1} \:,\\[1.25mm]
    \nabla^3_{\sigma\sigma\sigma} f(h, \rho, \sigma) &= (H_\sigma, H_\sigma) \mapsto \frac{2 s_\alpha^3}{z^3} \inprod{ \nabla_\sigma \widehat\Psi_\alpha (\rho, \sigma)}{H_\sigma}^2  \, \nabla_\sigma \widehat\Psi_\alpha (\rho, \sigma) + \nonumber \\[1.25mm]
    &+ \frac{s_\alpha^2}{z^2} \inprod{\nabla^2_{\sigma\sigma}\widehat\Psi_\alpha H_\sigma}{H_\sigma} \, \nabla_\sigma\widehat\Psi_\alpha(\rho, \sigma) + \frac{2 s_\alpha^2}{z^2} \inprod{\nabla_{\sigma}\widehat\Psi_\alpha}{H_\sigma} \, \nabla^2_{\sigma\sigma} \widehat\Psi_\alpha(\rho, \sigma) H_\sigma \nonumber \\[1.25mm]
    & +\frac{s_\alpha}{z}  \nabla^3_{\sigma\sigma\sigma} \widehat\Psi_\alpha(\rho, \sigma)[\,\cdot\,, H_\sigma, H_\sigma] - 2\sigma^{-1} H_\sigma \sigma^{-1} H_\sigma \sigma^{-1}\:,
\end{align}
\begin{align}
    \nabla^3_{\rho\rho\sigma} f(h, \rho, \sigma) &= (H_\rho, H_\sigma) \mapsto \frac{2 s_\alpha^3}{z^3} \inprod{ \nabla_\rho \widehat\Psi_\alpha (\rho, \sigma)}{H_\rho}\inprod{ \nabla_\sigma \widehat\Psi_\alpha (\rho, \sigma)}{H_\sigma}  \, \nabla_\rho \widehat\Psi_\alpha (\rho, \sigma)  \nonumber \\[1.25mm]
    &+ \frac{s_\alpha^2}{z^2} \inprod{\nabla^2_{\rho\sigma}\widehat\Psi_\alpha H_\sigma}{H_\rho} \, \nabla_\rho\widehat\Psi_\alpha(\rho, \sigma) + \frac{s_\alpha^2}{z^2} \inprod{\nabla_{\rho}\widehat\Psi_\alpha}{H_\rho} \, \nabla^2_{\rho\sigma} \widehat\Psi_\alpha(\rho, \sigma) H_\sigma \nonumber \\[1.25mm]
    & \frac{s_\alpha^2}{z^2} \inprod{\nabla_{\sigma}\widehat\Psi_\alpha}{H_\sigma} \, \nabla^2_{\rho\rho} \widehat\Psi_\alpha(\rho, \sigma) H_\rho +\frac{s_\alpha}{z}  \nabla^3_{\rho\rho\sigma} \widehat\Psi_\alpha(\rho, \sigma)[\,\cdot\,, H_\rho, H_\sigma] \:,
\end{align}
where $z = h-s_\alpha\widehat\Psi_\alpha (\rho)$. The remaining derivatives are completely analogous. Here we simplified the third derivative $\nabla^3_{\rho\rho\rho} f$ (and similarly for other terms) by applying it twice to the same vector because that is what the conic solver uses.

To obtain $\nabla^3 \Psihat(\rho, \sigma)$ we first differentiate the second derivatives of $\Psihathat(G,Z)$ and then substitute $G$ and $Z$. We obtain $D_{GGG}\Psihathat(G,Z)$ by differentiating \eqref{eq:DGG} over $G$:
\begin{equation}
    D_{GGG}^3\Psihathat[H_{G_1}, H_{G_2}, H_{G_3}] = \inprod{H_{G_1}}{Z_S^\frac12 D^2g'(Z_S^\frac12 G Z_S^\frac12)[Z_S^\frac12 H_{G_2} Z_S^\frac12, Z_S^\frac12 H_{G_3} Z_S^\frac12] Z_S^\frac12} \:.
\end{equation}
We obtain $D_{ZGG}^3\Psihathat(G,Z)$ by differentiating \eqref{eq:DZG} over $G$:
\begin{equation}
\begin{aligned}
    D_{ZGG}\Psihathat[H_{Z_1}, H_{G_2}, H_{G_3}] &= \inprod{  Dh(Z)[H_{Z_1}] }{S Z_S^{-\frac12} D^2\tilde g\big(Z_S^\frac12 G Z_S^\frac12\big)[Z_S^\frac12 H_{G_2} Z_S^\frac12, Z_S^\frac12 H_{G_3} Z_S^\frac12]Z_S^{-\frac12} S^\dagger} \\[1mm]
    &= \inprod{  H_{Z_1} }{Dh(Z)\left[S Z_S^{-\frac12} D^2\tilde g\big(Z_S^\frac12 G Z_S^\frac12\big)[Z_S^\frac12 H_{G_2} Z_S^\frac12, Z_S^\frac12 H_{G_3} Z_S^\frac12]Z_S^{-\frac12} S^\dagger \right]} \:.
\end{aligned}
\end{equation}
Because the third derivatives are continuous, $D^3_{GZG}\Psihathat = D_{ZGG}^3\Psihathat$. In order to solve the first vector direction and obtain $\nabla^3_{GZG}\Psihathat[\,\cdot\,,H_{Z_2},H_{G_3}]$, it is convenient to reorganize $D^3_{GZG}\Psihathat$:
\begin{equation}
    D_{GZG}^3\Psihathat(G,Z)[H_{G_1}, H_{Z_2}, H_{G_3}] = \inprod{H_{G_1}}{Z_S^\frac12 D^2\tilde g\big(Z_S^\frac12 G Z_S^\frac12\big)[Z_S^{-\frac12} S^\dagger Dh(Z)[H_{Z_2}] S Z_S^{-\frac12}, Z_S^\frac12 H_{G_3} Z_S^\frac12]Z_S^\frac12} \:,
\end{equation}
and analogously for $D_{GGZ}^3\Psihathat$.

To obtain $D_{GZZ}\Psihathat$, we follow Ref. \cite[Lemma 5.1]{he2025} and use the expression
\begin{equation}\label{swMatrix}
	\inprod{A}{X f(X Y^2 X) X} = \inprod{A}{\inv{Y} (Y X^2 Y) f(Y X^2 Y) \inv{Y}}\:.
\end{equation}
With this equation, we can rewrite 
\begin{equation}
    D_G\Psihathat(G,Z)[H_G] = \inprod{H_G}{Z_S^\frac12 g'\big(Z_S^\frac12 G Z_S^\frac12 \big) Z_S^\frac12} = \inprod{H_G}{G^{-\frac12} \tilde g\big(G^\frac12 Z_S G^\frac12 \big) G^{-\frac12}}\:,
\end{equation}
where $\tilde g(x) = xg'(x)$, and differentiate twice with respect to $Z$:
\begin{equation}
\begin{aligned}
    & D_{GZZ}\Psihathat[H_{G_1}, H_{Z_2}, H_{Z_3}] = \inprod{G^{-\frac12} H_{G_1} G^{-\frac12}}{D\tilde g(G^\frac12 Z_S G^\frac12)[G^\frac12 S^\dagger D^2h[H_{Z_2}, H_{Z_3}] S G^\frac12]} \\[2mm]
    \qquad &+\inprod{G^{-\frac12} H_{G_1} G^{-\frac12}}{D^2\tilde g(G^\frac12 Z_S G^\frac12)[G^\frac12 S^\dagger Dh[H_{Z_2}] S G^\frac12, G^\frac12 S^\dagger Dh[H_{Z_3}] S G^\frac12]}  \:.
\end{aligned}
\end{equation}

Again, to obtain $\nabla^3_{ZGZ}\Psihathat[\,\cdot\,,H_{G_2},H_{Z_3}]$ we simply reorganize $D_{GZZ}\Psihathat$:
\begin{equation}
\begin{aligned}
    & D_{ZGZ}\Psihathat[H_{Z_1}, H_{G_2}, H_{Z_3}] =
    \inprod{H_{Z_1}}{D^2h[S G^\frac12 D\tilde g(G^\frac12 Z_S G^\frac12)[G^{-\frac12} H_{G_2} G^{-\frac12}]G^\frac12 S^\dagger, H_{Z_3}]} \\[2mm]
    \qquad &+\inprod{H_{Z_1}}{Dh[ S G^\frac12D^2\tilde g(G^\frac12 Z_S G^\frac12)[G^{-\frac12} H_{G_2} G^{-\frac12}, G^\frac12 S^\dagger Dh[H_{Z_3}] S G^\frac12]G^\frac12 S^\dagger]}  \:.
\end{aligned}
\end{equation}
And analogously for $D_{ZZG}\Psihathat$.

Finally, we differentiate Eq. \eqref{eq:DZZ} with respect to Z using the product rule \eqref{eq:trace product rule} and use the self-adjointness of the spectral derivatives \eqref{eq:self adjointness third derivative} to obtain 
\begin{align}
    &D_{ZZZ}\Psihathat[H_{Z_1}, H_{Z_2}, H_{Z_3}] = \inprod{H_{Z_1}}{D^3h(Z)[S G^\frac12 g'\big(G^\frac12 Z_S G^\frac12\big)G^\frac12 S^\dagger, H_{Z_2}, H_{Z_3}]} \nonumber\\
    &+ \inprod{H_{Z_1} }{D^2h(Z)\left[H_{Z_2}, S G^\frac12 Dg'\big(G^\frac12 Z_S G^\frac12\big)[G^\frac12 S^\dagger Dh(Z)[H_{Z_3}] S G^\frac12]G^\frac12 S^\dagger\right]}  \nonumber\\
    &+ \inprod{H_{Z_1} }{D^2h(Z)\left[S G^\frac12 D g'\big(G^\frac12 Z_S G^\frac12\big)\Big[G^\frac12 S^\dagger Dh(Z)[H_{Z_2}]S G^\frac12\Big]G^\frac12 S^\dagger, H_{Z_3}\right]}  \\ 
    &+ \inprod{H_{Z_1}}{Dh(Z)\left[S G^\frac12 D g'\big(G^\frac12 Z_S G^\frac12\big)\Big[G^\frac12 S^\dagger D^2h(Z)[H_{Z_2}, H_{Z_3}]S G^\frac12\Big]G^\frac12 S^\dagger\right]} \nonumber\\
    &+ \inprod{H_{Z_1}}{Dh(Z)\left[S G^\frac12 D^2 g'\big(G^\frac12 Z_S G^\frac12\big)\Big[G^\frac12 S^\dagger Dh(Z)[H_{Z_2}]S G^\frac12, G^\frac12 S^\dagger Dh(Z)[H_{Z_3}]S G^\frac12\Big]G^\frac12 S^\dagger\right]} \:. \nonumber
\end{align}

Now we can simply substitute $G=\gmap(\rho)$, $Z=\zmap(\sigma)$ and the vector directions $H_{G_i} = \gmap(H_{\rho_i})$ and $H_{Z_i} = \zmap(H_{\sigma_i})$. We then solve the first vector direction to recover the third derivatives $\nabla^3 \Psihat(\rho, \sigma)$ applied to two vectors (actually, twice the same vector), which is what the conic solver needs.

\subsection{FastRényiQKD cone}

\subsubsection{Gradient}\label{sec:gradient fast qkd cone}

The gradient of the barrier \eqref{eq:fastbarrier} is 
\begin{align}
    & \nabla_h f(h, \rho) = -\frac{1}{z} \:, \\
    & \nabla_\rho f(h, \rho) = \frac{s_\alpha}{z} \, \nabla_\rho \Psihat (\rho) - \rho^{-1} \:,
\end{align}
where $z = h-s_\alpha\Psihat (\rho)$ and $\Psihat$ is given by Eq. \eqref{eq:fastpsi}.

The gradient of $\Psihat(\rho)$ can be obtained from the gradient of $\Psihat(\rho, \sigma)$ (computed in \ref{sec:gradient qkd cone}) just by composing $\Psihat(\rho, \sigma)$ with the projection $(\rho, \sigma) \mapsto (\rho, \rho)$ and replacing $\zmap$ with $\zgmap$:
\begin{align}\label{eq:grad fast psi}
    \nabla_\rho\Psihat(\rho) &= \gmap^\dagger(Z_S^\frac12 g'\big(Z_S^\frac12 G Z_S^\frac12 \big)Z_S^\frac12) + \zgmap^\dagger(Dh(Z)[S G^\frac12 g'\big(G^\frac12 Z_S G^\frac12\big)G^\frac12 S^\dagger]) \\[1mm]
    &= \gmap^\dagger(Z_S^\frac12 U g'\big(\Lambda^2\big) U^\dagger  Z_S^\frac12) + \zgmap^\dagger(Dh(Z)[S G^\frac12 V g'\big(\Lambda^2\big) V^\dagger G^\frac12 S^\dagger]) \:,
\end{align}
where in the last equality we used the singular value decomposition of $Z_S^\frac12 G^\frac12 = U \Lambda V^\dagger$.

\subsubsection{Hessian}

The Hessian of the barrier \eqref{eq:fastbarrier} is 
\begin{align}
	& \nabla^2_{hh} f(h, \rho) = \frac{1}{z^2} \:, \\[1.5mm]
	& \nabla^2_{h\rho} f(h, \rho) = - \frac{s_\alpha}{z^2} \, \nabla_\rho \Psihat (\rho) \:, \\[1.25mm]
	& \nabla^2_{\rho\rho} f(h, \rho) = H \mapsto \frac{s_\alpha^2}{z^2} \inprod{ \nabla_\rho \Psihat (\rho)}{H} \, \nabla_\rho \Psihat (\rho) + \frac{s_\alpha}{z} \nabla^2_{\rho\rho} \, \Psihat(\rho) H + \rho^{-1} H \rho^{-1} \:,
\end{align}
where $z = h-s_\alpha\Psihat (\rho)$.

To obtain the product of the Hessian times a vector $H$, we can simply compose $\Psihat(\rho, \sigma)$ with the projection $(\rho, \sigma) \mapsto (\rho, \rho)$, replace $\zmap$ with $\zgmap$, and recycle the equations for the Hessian of $\Psihat(\rho, \sigma)$, eqs. \eqref{eq:hessian rho rho}-\eqref{eq:hessian sigma sigma}:
\begin{equation}
    \nabla^2\Psihat(\rho) \, H = \nabla^2_{\rho \rho}\Psihat(\rho, \rho) \, H + \nabla^2_{\sigma \rho}\Psihat(\rho, \rho) \, H + \nabla^2_{\rho \sigma}\Psihat(\rho, \rho) \, H + \nabla^2_{\sigma \sigma}\Psihat(\rho, \rho) \, H \:.
\end{equation}
To obtain the Hessian, we can vectorize $\nabla^2\Psihat(\rho) \, H$ and isolate the vector direction $H$.

\subsubsection{Third derivatives}

Finally, the conic solver needs the third derivatives of the barrier \eqref{eq:fastbarrier} are
\begin{align}
	\nabla^3_{hhh} f(h, \rho) &= -\frac{2}{z^3} \:, \\[1.5mm]
	\nabla^3_{hh\rho} f(h, \rho) &= \frac{2 s_\alpha}{z^3} \, \nabla_\rho \Psihat (\rho) \:, \\[1.25mm]
	\nabla^3_{h\rho\rho} f(h, \rho) &= H \mapsto -\frac{2 s_\alpha^2}{z^3} \inprod{ \nabla_\rho \Psihat (\rho)}{H} \, \nabla_\rho \Psihat (\rho) - \frac{s_\alpha}{z^2} \nabla^2_{\rho\rho} \, \Psihat(\rho) \,H \:, \\[1.25mm]
    \nabla^3_{\rho\rho\rho} f(h, \rho) &= (H_1, H_2) \mapsto \frac{2 s_\alpha^3}{z^3} \inprod{ \nabla_\rho \Psihat (\rho)}{H_1} \inprod{ \nabla_\rho \Psihat (\rho)}{H_2} \, \nabla_\rho \Psihat (\rho) \nonumber \\[1.25mm]
    &+ \frac{s_\alpha^2}{z^2} \inprod{\nabla^2_{\rho\rho}\Psihat \,H_2}{H_1} \, \nabla_\rho\Psihat(\rho) +\frac{s_\alpha^2}{z^2} \inprod{\nabla_{\rho}\Psihat}{H_1} \, \nabla^2_{\rho\rho} \Psihat(\rho)\,H_2 \nonumber \\[1.25mm]
    & + \frac{s_\alpha^2}{z^2} \inprod{\nabla_{\rho}\Psihat}{H_2} \nabla^2_{\rho\rho} \, \Psihat(\rho) \, H_1  +\frac{s_\alpha}{z}  \nabla^3_{\rho\rho\rho} \Psihat(\rho)[\,\cdot\,, H_1, H_2] \nonumber \\[1.25mm]
    & - \rho^{-1} H_2 \rho^{-1} H_1 \rho^{-1} - \rho^{-1} H_1 \rho^{-1} H_2 \rho^{-1}\:,
\end{align}
where $z = h-s_\alpha\Psihat (\rho)$. Since the conic solver will use $\nabla^3 f$ applied twice to the same vector direction $H$, we can simplify the last derivative:
\begin{equation}
\begin{aligned}
    \nabla^3_{\rho\rho\rho} f(h, \rho) [\,\cdot\,, H, H] &= \frac{2 s_\alpha^3}{z^3} \inprod{ \nabla_\rho \Psihat (\rho)}{H}^2 \, \nabla_\rho \Psihat (\rho) \\[1.25mm]
    &+ \frac{s_\alpha^2}{z^2} \inprod{\nabla^2_{\rho\rho}\Psihat \,H}{H} \, \nabla_\rho\Psihat(\rho) + 2 \frac{s_\alpha^2}{z^2} \inprod{\nabla_{\rho}\Psihat}{H} \, \nabla^2_{\rho\rho} \Psihat(\rho)\,H \\[1.25mm]
    & +\frac{s_\alpha}{z}  \nabla^3_{\rho\rho\rho} \Psihat(\rho)[\,\cdot\,, H, H] - 2\rho^{-1} H \rho^{-1} H \rho^{-1}\:.
\end{aligned}
\end{equation}

Finally, as in the previous sections, by composing $\Psihat(\rho, \sigma)$ with the projection $(\rho, \sigma) \mapsto (\rho, \rho)$ and taking $H := H_\sigma = H_\rho$, we can recycle the third derivatives of Sec \ref{sec:third derivatives qkd cone}:
\begin{equation}
\begin{aligned}
    \nabla^3_{\rho\rho\rho} \Psihat(\rho)[\,\cdot\,, H, H] &= \nabla^3_{\rho\rho\rho} \Psihat(\rho, \rho)[\,\cdot\,, H, H] + \nabla^3_{\sigma\rho\rho} \Psihat(\rho, \rho)[\,\cdot\,, H, H] \\[1mm]
    &+ \nabla^3_{\rho\sigma\rho} \Psihat(\rho, \rho)[\,\cdot\,, H, H] + 
    \nabla^3_{\rho\rho\sigma} \Psihat(\rho, \rho)[\,\cdot\,, H, H] \\[1mm]
    &+ \nabla^3_{\sigma\sigma\rho} \Psihat(\rho, \rho)[\,\cdot\,, H, H] + 
    \nabla^3_{\sigma\rho\sigma} \Psihat(\rho, \rho)[\,\cdot\,, H, H] \\[1mm]
    &+ \nabla^3_{\rho\sigma\sigma} \Psihat(\rho, \rho)[\,\cdot\,, H, H] +
    \nabla^3_{\sigma\sigma\sigma} \Psihat(\rho, \rho)[\,\cdot\,, H, H] \:.
\end{aligned}
\end{equation}

\section{Starting point}\label{sec:startingpoint}

\subsection{RényiQKD cone}

The starting point for the RényiQKD cone \eqref{eq: qkd cone} is given by the solution of
\begin{equation}
    (h,\rho,\sigma) = -\nabla f(h,\rho,\sigma),
\end{equation}
where $f(h,\rho,\sigma)$ is the barrier function \eqref{eq:truebarrier}. Recalling the derivatives from Appendix \ref{sec:gradient qkd cone}, we obtain the following system of non-linear equations:
\begin{subequations}\label{eq:initialcentral}
\begin{align}
h &= \frac1z, \label{eq:h}  \\
\rho &= \frac1z \nabla_\rho z + \rho^{-1}, \label{eq:rho} \\
\sigma &= \frac1z \nabla_\sigma z + \sigma^{-1}, \label{eq:sigma}
\end{align}
\end{subequations}
where $z = h - s_\alpha\Psihat$. 

If $d_\rho$ and $d_\sigma$ are the dimensions of $\rho$ and $\sigma$, we can plug in the ansatz $\rho = \gamma \id_{d_\rho}, \sigma = \delta \id_{d_\sigma}$. In the particular case where $\gmap(\id_{d_\rho}) = \gmap^\dagger(\id_{d_\rho}) = \id_{d_\rho}$ and $\zmap(\id_{d_\sigma}) = \zmap^\dagger(\id_{d_\sigma}) = \id_{d_\sigma}$, Eq.s \eqref{eq:truegradient rho} and \eqref{eq:truegradient sigma} give us
\begin{gather}
\nabla_{\rho}\Psihat(\rho,\sigma) = \alpha \de{\frac\gamma\delta}^{\alpha-1}\id_{d_\rho} \\
\nabla_{\sigma}\Psihat(\rho,\sigma) = (1-\alpha) \de{\frac\gamma\delta}^{\alpha}\id_{d_\sigma}
\end{gather}
In this particular case we get $\Psihat = d_\rho\de{\frac\gamma\delta}^{\alpha} \delta$, so we substitute that in Eq.s \eqref{eq:initialcentral}, obtaining
\begin{subequations}
\begin{gather}
h^2 - s_\alpha h d_\rho\de{\frac\gamma\delta}^{\alpha} \delta -1 = 0 \label{eq:h2}\\
-s_\alpha h \alpha \de{\frac\gamma\delta}^{\alpha-1} = \gamma - \frac1\gamma \label{eq:rho2} \\
-s_\alpha h (1-\alpha) \de{\frac\gamma\delta}^{\alpha} = \delta - \frac1\delta \label{eq:sigma2}
\end{gather}
\end{subequations}
We eliminate $h$ by combining \eqref{eq:h2} and \eqref{eq:rho2} and also by combining \eqref{eq:rho2} and \eqref{eq:sigma2}, obtaining
\begin{gather}
    (\gamma^2 -1)^2 \gamma^{-2\alpha} (\delta^2)^{\alpha-1} + d_\rho\alpha(\gamma^2-1) - \alpha^2 = 0 \label{eq:gamma} \\
    \frac{\gamma^2-1}{\alpha} = \frac{\delta^2-1}{1-\alpha} \label{eq:delta} 
\end{gather}
We eliminate $\delta^2$ from Eq. \eqref{eq:gamma} using \eqref{eq:delta}, obtaining an equation on a single variable that can then be solved numerically by Newton's method. We substitute back in Eq. \eqref{eq:delta} to obtain $\delta$, and in Eq. \eqref{eq:h} to obtain $h$.

In general, computing the initial point for an arbitrary choice of $\gmap,\zmap$ is not feasible. Note, however, that this strategy always produces a valid starting point, since $\gamma\id_{d_\rho} \succ 0, \delta\id_{d_\sigma} \succ 0$ and Eq. \eqref{eq:h} implies that $h > s_\alpha\Psihat$, and therefore $(h,\gamma\id_{d_\rho},\delta\id_{d_\sigma})$ belongs to the interior of the cone.

If the solver is going to be run repeatedly for a fixed choice of $\gmap, \zmap$, one can compute numerically the central point for these maps and supply the result to the solver.

\subsection{FastRényiQKD cone}

The starting point for the FastRényiQKD cone \eqref{eq: fast qkd cone} is given by the solution of  
\begin{equation}
    (h,\rho) = -\nabla f(h,\rho),
\end{equation}
where $f(h, \rho)$ is the barrier \eqref{eq:fastbarrier}. Substituting the derivatives from Appendix \ref{sec:gradient fast qkd cone} we obtain the equations
\begin{align}
h &= \frac1z, \\
\rho &= \frac1z \nabla_\rho z + \rho^{-1}, \label{eq:startrho}
\end{align}
where $z = h - s_\alpha\Psihat$.

The first equation, together with the requirement that $h \ge s_\alpha\Psihat$, implies
\begin{equation}\label{eq:optimalh}
    h = \frac{s_\alpha\Psihat}{2} + \sqrt{1 + \frac{\Psihat^2}{4}}
\end{equation}
Using the ansatz $\rho = \gamma \id$ and in the particular case that $\gmap(\id) = \gmap^\dagger(\id) = \zgmap(\id) = \zgmap^\dagger(\id) = \id,$ Eq. \eqref{eq:grad fast psi} gives us
\begin{equation}
\nabla_{\rho}\Psi(\rho) = \alpha \id + (1-\alpha)\id = \id.
\end{equation}
Eq. \eqref{eq:startrho} then boils down to
\begin{equation}
-s_\alpha h = \gamma - \frac1\gamma
\end{equation}
In this particular case, we get $\Psihat = \gamma d$, so we solve
\begin{equation}
    -\frac{\gamma d}{2} -s_\alpha \sqrt{1 + \frac{\gamma^2d^2}{4}} = \gamma - \frac1\gamma
\end{equation}
to obtain
\begin{equation}
    \gamma = \sqrt{\frac{d+3}{2d+2}-\frac{s_\alpha}{2}\sqrt{1 + \frac{4}{(d+1)^2}}} \:.
\end{equation}
As with the RényiQKD cone, computing the initial point for an arbitrary choice of $\gmap,\zgmap$ is not feasible. Note, however, that choosing $\rho = \gamma\id$ and setting $h$ to be given by Eq. \eqref{eq:optimalh} is always a valid starting point, since $\gamma\id \succ 0$ and $\frac{s_\alpha\Psihat}{2} + \sqrt{1 + \frac{\Psihat^2}{4}} > s_\alpha\Psihat$, and thus it belongs to the interior of the cone.

\printbibliography

\end{document}